\documentclass[a4paper,11pt,DIV=calc]{scrartcl}
\usepackage{hyperref}
\usepackage[utf8]{inputenc}
\usepackage{amsmath,amsfonts,amssymb,amsthm,mathtools}
\usepackage{amsmath}
\usepackage{amsfonts}
\usepackage{amssymb}
\usepackage{dsfont}
\usepackage{mathrsfs}
\usepackage[T1]{fontenc}
\usepackage[ngerman,english]{babel}
\usepackage[inline]{enumitem}
\usepackage{latexsym,graphicx}
\usepackage[all]{xy}
\usepackage{mathtools}
\usepackage{color}
\usepackage{authblk}
\usepackage{amscd}
\usepackage{microtype}
\numberwithin{equation}{section}

\DeclareMathOperator{\tr}{Tr}

\newtheorem{theorem}{Theorem}[section]
\newtheorem{proposition}{Proposition}[section]
\newtheorem{lemma}{Lemma}[section]
\newtheorem{corollary}{Corollary}[section]
\theoremstyle{definition}
\newtheorem{remark}{Remark}[section]

\theoremstyle{definition}

\newcommand{\half}{\mbox{$\frac{1}{2}$}}
\newcommand{\beq}{\begin{equation}}
\newcommand{\eeq}{\end{equation}}
\usepackage{pdfsync}

\begin{document}

\title{Bose-Einstein Condensation in a Dilute, Trapped Gas at Positive Temperature}

\author{Andreas Deuchert, Robert Seiringer, Jakob Yngvason}

\date{April 18, 2018}

\maketitle

\begin{abstract} 
We consider an interacting,  dilute Bose gas trapped in a harmonic potential at a positive temperature.  The system is analyzed in a  combination of a thermodynamic and a Gross-Pitaevskii (GP) limit where the trap frequency $\omega$, the temperature  $T$ and the particle number $N$ are related by $N \sim (T / \omega)^{3} \to\infty$ while the scattering length is so small that the interaction energy per particle around the center of the trap is of the same order of magnitude as the spectral gap in the trap. 

We prove that the difference between the canonical free energy of the interacting gas and the one of the noninteracting system can be obtained by minimizing the GP energy functional. We also prove Bose-Einstein condensation in the following sense: The one-particle density matrix of any approximate minimizer of the canonical free energy functional is to leading order given by that of the noninteracting gas but with the free condensate wavefunction replaced by the GP minimizer. 
\end{abstract}

\setcounter{tocdepth}{2}
\tableofcontents

\section{Introduction and main results}
\label{sec:introductionandmainresults}
\subsection{Background and summary}
\label{sec:background}

Proving Bose-Einstein condensation (BEC) rigorously for systems of interacting particles has for a long time been a major challenge in Mathematical Physics. The experimental realization of BEC  in trapped alkali gases in 1995 \cite{WieCor1995, Kett1995} triggered numerous mathematical investigations of the properties of dilute Bose gases. Building on work of Dyson on hard-core bosons from 1957 \cite{Dyson}, the first proof of an asymptotically accurate lower bound for the ground state energy of a dilute Bose gas in the thermodynamic limit was achieved in \cite{LiYng1998}. Together with the upper bound given in \cite{RobertGPderivation}, it  established firmly the leading order behavior of the ground state energy. Perhaps more important than the result itself were the techniques of the proof which formed the basis of much subsequent work.  

For dilute, trapped gases as prepared in typical experiments the Gross-Pitaevskii (GP) limit for the ground state is relevant. This limit is characterized by the requirement that the interaction energy per particle is kept of the same order of magnitude as the spectral gap in the trap. Mathematically, this can be achieved by either scaling the trap or the interaction potential suitably as the particle number tends to infinity.
In \cite{RobertGPderivation,LiSei2002,LiSei2006} it was proved that the ground state energy of a Bose gas is in this limit equal to the minimum of the GP energy functional. Additionally, the  projection onto the minimizer of this functional is the limit of  the one-particle density matrix of the gas, proving complete Bose-Einstein condensation in the ground state. The dynamics of a system in the GP limit, on the other hand, can be described by the time-dependent GP equation, see \cite{ErdSchlYau2009,ErdSchlYau2010,BenOlivSchl2015,Pickl2015}. For a more extensive list of references to the mathematical analysis of dilute Bose gases we refer to \cite{Themathematicsofthebosegas,Rou2015,BenPorSchl2015}. 

While ground states provide a good description of quantum gases at very low temperatures in first approximation, the understanding of finite temperature effects in cold gases is also essential. For instance, the spectacular emergence of a peak in the momentum distribution out of a maxwellian thermal cloud in the experiments \cite{WieCor1995,Kett1995} as the temperature falls below a critical value, cannot be explained from the ground state properties alone.
For describing such phenomena one has to consider the  Gibbs state and the free energy of the system rather than the ground state and the corresponding energy. For a dilute,  homogeneous Bose gas, the leading order behavior of the free energy has been established, see \cite{Sei2008} for the lower bound and \cite{Yin2010} for the upper bound. The techniques developed in \cite{LiYng1998,RobertGPderivation} have also been extended to treat fermions, both for the ground state \cite{LiSeiSol2005} and for the free energy at finite temperature \cite{RobertFermigas}. 
We mention also the papers \cite{LewinNamRougerie2015,LewinNamRougerie2017} and \cite{FKSS2017}   where Gibbs states of  Bose gases with mean-field interactions are studied.
A general proof of Bose-Einstein condensation in dilute gases remains elusive, however.

In this paper we consider the Gibbs state of an interacting Bose gas in a harmonic trap at positive temperatures in a combination of a thermodynamic and  a GP limit. We show that in this limit the free energy becomes equal to that of  the ideal gas plus a correction given by the GP energy of the condensate. Moreover, we show that the one-particle density matrix of the system is well approximated by that of the ideal gas with the noninteracting condensate wavefunction replaced by the minimizer of the GP energy functional. This  proves, in particular, Bose-Einstein condensation at positive temperatures with the same transition temperature and the same condensate fraction as for the ideal gas to leading order.
\subsection{Notation}
\label{sec:notation}
For functions $a$ and $b$ depending on the particle number or other parameters, we use the notation $a \lesssim b$ to say that there exists a constant $C>0$ independent of the parameters such that $a \leq C b$. If $a \lesssim b$ and $b \lesssim a$ we write $a \sim b$ and $ a \simeq b$ means that $a$ and $b$ are equal to leading order in the limit considered. 
\subsection{The model}
\label{sec:model}
We consider a system of $N$ bosons trapped in a three-dimensional harmonic potential with trap frequency $\omega$.
The one-particle Hilbert space is $\mathcal{H}_1 = L^2(\mathbb{R}^3)$ and that for the whole system is the $N$-fold symmetric tensor product $\mathcal{H}_N = \otimes_{\text{s}}^N L^2(\mathbb{R}^3)$. On $\mathcal{H}_N$ we define the Hamiltonian of the system by\footnote{In our units the mass is $m=\half$ and $\hbar=1$ so that $\hbar^2/2m=1$ and $ m/2=\hbox{$\frac 14$}$.  Moreover, Boltzmann's constant $k_{\text{B}}$ is taken to be 1.}
\begin{equation}
H_N = \sum_{i=1}^N \left( -\Delta_i +\hbox{$\frac 14$}\omega^2 x_i^2 - \hbox{$\frac 32$} \omega \right) + \sum_{1 \leq i < j \leq N} v_N(x_i-x_j). \label{eq:main1}
\end{equation}
In this formula $\Delta_i$ denotes the Laplacian acting on the $i$-th particle and we have  subtracted 
$\hbox{$\frac 32$} \omega$ for convenience so that the ground state energy of the harmonic oscillator is zero. The interaction potential is given by 
\beq v_N(x) = \omega N^2 v( \omega^{1/2} N x )\label{eq:intpot} 
\eeq
with a nonnegative, radial and measurable function $v$, independent of $N$. 
A simple scaling argument shows that if $a_v$ is the (dimensionless) scattering length of $v$ then the scattering length $a_N$ of $v_N$ is 
\beq 
a_N=a_v\, \omega^{-1/2}  N^{-1}\,.\label{eq:scaling}
 \eeq 
The scattering length is a combined measure of the range and the strength of a potential and its definition is recalled in Section~\ref{sec:scatteringlength}. To be able to include hard-core interactions, we allow $v$ to take the value $+\infty$. If $v$ happens to be infinite on a set of nonzero measure the domain of the Hamiltonian has to be restricted to functions that vanish on this set. We require that $v$ is integrable outside some finite ball in order for the scattering length $a_v$ to be finite.

To motivate the scaling \eqref{eq:intpot} we recall that the energy per particle of a dilute gas of density $\rho$ with 
 interaction potential $v_N$ is to first approximation proportional to  $\rho a_N$ \cite{LiYng1998, RobertGPderivation}. If  $\rho\sim N\/\ell_{\text{osc}}^{-3}$, where    $\ell_{\text{osc}} \sim \omega^{-1/2}$ denotes the length scale of the trap, we see that
\beq \rho a_N\sim \omega
\label{eq:int} \eeq
i.e., the interaction energy per particle 
is of the order of the spectral gap $\omega$ as $N\to\infty$. Note also that the dimensionless \lq\lq gas parameter\rq\rq\ $\rho a_N^3$ tends to zero as $N^{-2}$ so the scaling \eqref{eq:intpot} amounts to considering a special case of a dilute limit.

\subsection{The Gross-Pitaevskii energy functional and the GP limit}
\label{sec:GP}
The GP energy functional $\mathcal{E}^{\text{GP}}$ with trapping potential as in Eq.~\eqref{eq:main1} and scattering length $a$ is defined by
\begin{equation}
\mathcal{E}^{\text{GP}}(\phi) = \int_{\mathbb{R}^3} \left( \vert \nabla \phi(x) \vert^2 + \left( \tfrac 14 \omega^2 x^2 - \tfrac 32\omega \right) \vert \phi(x) \vert^2 + 4 \pi a \vert \phi(x) \vert^4 \right) \text{d}x.
\label{eq:mainresult1}
\end{equation} 
Its ground state energy  is
\begin{equation}
E^{\text{GP}}(N,a,\omega) = \inf_{\Vert \phi \Vert^2_{L^2\left(\mathbb{R}^3 \right)}=N} \mathcal{E}^{\text{GP}}(\phi).
\label{eq:mainresult2}
\end{equation}
The unique minimizer of $\mathcal{E}^{\text{GP}}$ will be denoted by $\phi^{\text{GP}}_{N,a}$. To keep the notation simple, we suppress its dependence on $\omega$. The energy and the minimizer satisfy the scaling relations
\beq 
E^{\text{GP}}(N,a,\omega)=\omega N E^{\text{GP}}(1,N \omega^{1/2} a,1),\quad \quad\phi^{\text{GP}}_{N,a}=N^{1/2}\phi^{\text{GP}}_{1,Na} \label{eq:GPscaling}
 \eeq
For a detailed discussion of the mathematical properties of the GP functional and its relation to the ground state of the Hamiltonian \eqref{eq:main1} we refer to \cite{LiSei2002} (see also \cite{LiSei2006,Themathematicsofthebosegas,RobertGPderivation,NRS}) where the following is proved: 

In the {\em GP limit}, where $a=a_N$ as in Eq.~\eqref{eq:scaling} and $N\to\infty$, the ground state energy per particle of the many-body Hamiltonian \eqref{eq:main1} converges to $ \omega E^{\text{GP}}(1,N \omega^{1/2} a_N,1)= \omega E^{\text{GP}}(1,a_v,1 )$. Moreover, the normalized one-particle density matrix of the ground state wavefunction converges in trace norm to the projector onto $\phi^{\text{GP}}_{1,a_v \omega^{-1/2}}$. 

Since the GP minimizer differs considerably from the gaussian ground state of the harmonic oscillator  if $a_v$ is large, the interaction can leave a clear mark on the density profile of the gas despite the high dilution imposed by the GP limit, as seen in experiments \cite{Hau1998, Dalfovo_etal1999}.  In fact, for large $a_v$ (\lq\lq Thomas-Fermi limit\rq\rq) the profile has approximately an inverse parabolic shape of extension $\sim a_v^{2/5}\ell_{\text{osc}}$.

Our goal is to generalize these results to equilibrium states at positive temperatures when the GP limit is combined with the natural thermodynamic limit in the trap.  The definition of the latter and the heuristics behind our main results can be deduced from a comparison of the length scales involved in the problem, as discussed next.

\subsection{Length scales, thermodynamic limit, heuristics}
\label{sec:lengthscales}

For a noninteracting gas at inverse temperature $\beta=T^{-1}$ the following length scales are relevant:
\begin{itemize}
\item The extension $\ell_{\text{osc}}\sim\omega^{-1/2}$ of the ground state of the harmonic oscillator.
\item The thermal de Broglie wavelength $\ell_{\text{th}}\sim\beta^{1/2}$.
\item The extension of the thermal cloud in the trap, $R_{\text{th}}
\sim \omega^{-1}\beta^{-1/2}$, obtained by equating the potential energy $\omega^2R_{\text{th}}^2$ and the thermal kinetic energy $\beta^{-1}$.
\item The mean particle distance $d_{\text{th}}\sim N^{-1/3}R_{\text{th}}$ in the thermal cloud.
\end{itemize}

The {\em thermodynamic limit} is defined by keeping the ratio $d_{\text{th}}/\ell_{\text{th}}$ fixed as $N\to\infty$, i.e. by the condition
\beq
N(\beta\omega)^3\quad\text{fixed.}
\eeq
The thermodynamic limit requires in particular
$(\beta\omega)\sim N^{-1/3}\to 0$ and thus $\omega\ll T$.
If $d_{\text{th}}\lesssim \ell_{\text{th}}$, i.e., $N(\beta\omega)^{3}\gtrsim 1$, thermal de Broglie wave packets overlap and condensation can be expected. This heuristics is confirmed by the standard analysis of the ideal Bose gas in the harmonic trap \cite{PitaevskiiStringari,PethickSmith,ChatterjeeDiakonis}: 

Bose-Einstein condensation takes place if the temperature $T$ is smaller than the critical temperature $T_{\text{c}}$ given by
\begin{equation}
T_{\text{c}}(N,\omega) = \omega \left( \frac{N}{\zeta(3)} \right)^{1/3}.
\label{eq:idealbosegas2}
\end{equation}  
Here $\zeta$ is the zeta-function and $\zeta(3) =1.202\dots$ To be more precise, denote by $N_0(\beta,N,\omega)$ the expected number of particles occupying the ground state of the harmonic oscillator in the canonical ensemble. If $N\to\infty$ with $N(\beta\omega)^3$ fixed
the condensate fraction is given by
\begin{equation}
\lim \frac{N_0(\beta,N,\omega)}{N} =  \left[ 1 - \left( \frac{T}{T_{\text{c}}} \right)^3 \right]_+
\label{eq:idealbosegas3}
\end{equation}
with $[t]_+ = \max\lbrace t,0 \rbrace$. 
The condition for the right-hand side of Eq.~\eqref{eq:idealbosegas3} to be larger than zero, i.e., $T< T_{\text c}$, is equivalent to
\beq\label{eq:beccond}
\lim N(\beta\omega)^3>1.202.
\eeq
In this case the ground state of the harmonic oscillator is macroscopically occupied in the ideal Bose gas.
For $T>T_c$ on the other hand, i.e., if
\beq\label{eq:beccond2}
\lim N(\beta\omega)^3<1.202,
\eeq
there is no condensation.
These formulas are most conveniently derived in the grand canonical ensemble.

For an assessment of the effects of interactions the following observation is crucial: The length scales $R_{\text{th}}$ and $\ell_{\text{osc}}$ become separated if $N\to\infty$,
\beq
\ell_{\text{osc}}/R_{\text{th}}\sim (\beta\omega)^{1/2}\sim N^{-1/6}.\label{eq:scalesep}
\eeq
The average density of the condensate,
\beq \rho_0\sim N_0/\ell_{\text{osc}}^3\label{eq:rho0},\eeq
and that of the thermal cloud,
\beq \rho_{\text{th}}\sim N_{\text{th}}/R_{\text{th}}^3\sim (\beta\omega)^{3/2} N_{\text{th}}/\ell_{\text{osc}}^3 \label{eq:rhocloud},\eeq
with $N_{\text{th}}:=N-N_0$, are therefore widely different in the condensation regime $0<T<T_c$ where $N_0$ and $N_{\text{th}}$ are comparable:
\beq
 \rho_{\text{th}}/\rho_0\sim (\beta\omega)^{3/2}\sim N^{-1/2}.\label{eq:scalesep2}
\eeq
The same holds for the ratios of the interaction energy per particle, $a_N\rho_{\text{th}}$ and $a_N\rho_0$ respectively. We remark that in the presence of a condensate described by a GP minimizer with large $a_v$ it would be more precise to replace $\ell_{\text{osc}}$ in Eq.~\eqref{eq:rho0}--\eqref{eq:rhocloud} by $a_v^{2/5}\ell_{\text{osc}}$. Also, since the ball $|x|\lesssim a_v^{2/5}\ell_{\text{osc}}$ is essentially excluded for the thermal cloud, 
$R_{\text{th}}$ in Eq.~\eqref{eq:rhocloud} should  be replaced approximately by
$R_{\text{th}}+a_v^{2/5}\ell_{\text{osc}}$.  As long as $a_v$ stays $O(1)$, however,  the asymptotic behavior in Eq.~\eqref{eq:scalesep2} remains valid.

The {\em separation of scales} expressed by Eq.~\eqref{eq:scalesep}  leads to the following expectations for the combined GP and thermodynamic limit:
\begin{itemize}
\item The thermal cloud of the ideal gas remains essentially intact.

\item BEC takes place for $T<T_c$ and the condensate can be described by the GP minimizer, residing close to the center of the trap.
\end{itemize}
Transforming this heuristic picture into a mathematical proof is the subject of this paper.
In order to state our main results precisely we need a few more definitions.

\subsection{Gibbs state, free energy and the concept of BEC}
\label{sec:Gibbs}

The canonical Gibbs state for the Hamiltonian \eqref{eq:main1} is
 \beq \Gamma^\mathrm{G}_{N} = Z(\beta,N,\omega)^{-1} e^{-\beta H_N}\label{eq:Gibbsstate}\eeq
with $Z(\beta,N,\omega)=\tr_{\mathcal{H}_N} \!\left[ e^{- \beta H_N } \right]$ the canonical partition function.
The free energy of the system at inverse temperature $\beta=T^{-1}$ is given by
\begin{equation}
F(\beta, N, \omega) = -\tfrac{1}{\beta} \ln\left( \tr_{\mathcal{H}_N} \left[ e^{- \beta H_N } \right] \right).
\label{eq:main2}
\end{equation}
The trace in Eqs.~\eqref{eq:Gibbsstate}--\eqref{eq:main2} is taken over $\mathcal{H}_N$, that is, over the subspace of permutation symmetric functions in $L^2(\mathbb{R}^{3N})$. In the following we will drop the index $\mathcal{H}_N$ and just write $\tr$ for this trace. By $F_0(\beta,N,\omega)$ we denote the free energy of the ideal Bose gas in the harmonic trap, that is, the one for $v = 0$.

 A useful characterization of the free energy is via the Gibbs variational principle. Denote by $\mathcal{S}_N$ the set of states on $\mathcal{H}_N$ that have finite energy with respect to $H_N$. In other words, consider the set of all linear operators $\Gamma$ on $\mathcal{H}_N$ with $0 \leq \Gamma \leq 1$, $\tr[\Gamma] = 1$ and $\tr[H_N \Gamma] < \infty$.\footnote{Here and elsewhere in the paper, we shall interpret $\tr H \Gamma$ for positive operators $H$ and states $\Gamma$ as $\tr H^{1/2} \Gamma H^{1/2}$, which is always well-defined if one allows the value $+\infty$. In particular, finiteness of $\tr H\Gamma$ does not require that $H\Gamma$ is trace class, only that $H^{1/2}\Gamma H^{1/2}$ is.}  Also denote by $S(\Gamma) = -\tr[\Gamma \ln(\Gamma)]$ the entropy of a state $\Gamma \in \mathcal{S}_N$. The free energy functional is defined to be 
\begin{equation}
\mathcal{F}_N(\Gamma) = \tr[H_N \Gamma] - TS(\Gamma),
\label{eq:main3}
\end{equation} 
and using this definition, the free energy can be written as 
\beq \label{gibbsv}
F(\beta,N,\omega) = \min_{\Gamma \in \mathcal{S}_N} \mathcal{F}_N(\Gamma).
\eeq 
The unique minimizer of $\mathcal{F}_N$ is given by the canonical Gibbs state defined in Eq.~\eqref{eq:Gibbsstate}.

The reduced one-particle density matrix  of a state $\Gamma_N \in \mathcal{S}_N$ is defined via the integral kernel
\begin{equation}
\gamma_N(x,y) = \tr_{\mathcal{H}_N}[ a^{*}_y a_x \Gamma_N ]. \label{eq:mainresult3}
\end{equation}
Here $a^{*}_x$ and $a_x$ denote creation and the annihilation operators (actually operator-valued distributions) of a particle at point $x$,  fulfilling the canonical commutation relation $[a_x,a_y^*] = \delta(x-y)$.  Alternatively, $\gamma_N$ can be defined as $N$ times the partial trace of $\Gamma_N$ over $N-1$ particle variables. 

Finally, {\em Bose-Einstein condensation} for a sequence of states $\Gamma_N$ means, by definition, that
\beq
\liminf_{N\to\infty}\frac 1N\sup_{\Vert \psi\Vert_2=1}\langle \psi,\gamma_N \psi\rangle >0.\label{eq:BECdef}
\eeq

\subsection{The main theorem}
\label{sec:mainresult}

\begin{theorem}
\label{thm:main}
Assume that $v$ is a nonnegative, radial and measurable function which is integrable outside some finite  ball. Let $H_N$ be the Hamiltonian \eqref{eq:main1} with interaction potential $v_N$ given by  Eq.~\eqref{eq:intpot}.  Let $F(\beta,N,\omega)$ be the corresponding free energy, $F_0(\beta,N,\omega)$ the free energy of the ideal gas, and $N_0(\beta,N,\omega)$ the expected number of particles occupying the ground state of the harmonic oscillator in the canonical Gibbs state of the ideal Bose gas. In the combined thermodynamic and  GP limit, that is, for $N \to\infty $, $(\beta \omega)^{-3} \sim N$ and $a_N$ as in Eq.~\eqref{eq:scaling} with $a_{v}$ fixed, we have 
\begin{equation}
\lim \tfrac{1}{\omega N} \left\vert F(\beta,N,\omega) - F_0(\beta,N,\omega) - E^{\mathrm{GP}}(N_0,a_N,\omega) \right\vert = 0.
\label{eq:mainresult4}
\end{equation}
Moreover, for any sequence of states $\Gamma_N \in \mathcal{S}_N$ with 
\begin{equation}
\lim \tfrac{1}{\omega N} \left\vert \mathcal{F}_N(\Gamma_N) - F_0(\beta,N,\omega) -E^{\mathrm{GP}}(N_0,a_N,\omega) \right\vert = 0
\label{eq:mainresult5}
\end{equation}
we have
\begin{equation}
\lim \tfrac{1}{N} \left\Vert \gamma_N - \left( \gamma_{N,0} - N_0 \vert \varphi_0 \rangle\langle \varphi_0 \vert +  \vert \phi^{\mathrm{GP}}_{N_0, a_N} \rangle\langle \phi^{\mathrm{GP}}_{N_0, a_N} \vert \right) \right\Vert_1 = 0.
\label{eq:mainresult6}
\end{equation} 
Here $\gamma_{N,0}$ denotes the one-particle density matrix of the noninteracting canonical Gibbs state, $\varphi_0$ is the normalized ground state wavefunction of the harmonic oscillator Hamiltonian $h = -\Delta + \hbox{$\frac 14$}\omega^2 x^2 -  \hbox{$\frac 32$}\omega$, and $\Vert \cdot \Vert_1$ stands for the trace norm. 
Finally, 
\beq\label{eq:mainresult7}
\lim \tfrac{1}{N} \left\Vert
 \gamma_N-\vert \phi^{\mathrm{GP}}_{N_0, a_N} \rangle\langle \phi^{\mathrm{GP}}_{N_0, a_N} \vert
\right\Vert= 0
\eeq
where $\Vert \cdot \Vert$ is the operator norm. In particular, Bose-Einstein condensation takes place  with the same transition temperature $T_c$ and the same condensate fraction as for the ideal gas to leading order, with the GP minimizer macroscopically occupied while the occupation of every state orthogonal to it is $o(N)$.
\end{theorem}

\subsection{Remarks}
\label{sec:remarks}

\textbf {1.} 
The bounds leading to Theorem~\ref{thm:main} are uniform in $(\beta \omega)^{-3} N^{-1}$ as long as this quantity remains in a compact interval $[c,d]$, $0 \leq c<d<\infty$ of the real line. That is, we need not require $\lim T/T_{\mathrm{c}}$ to exist, only to stay bounded. In particular, Theorem~\ref{thm:main} continuously extrapolates to the known result at $T=0$.

\textbf {2.}
We have the following uniformity of our bounds in the scattering length: Assume $v(x) = a_v^{-2} \tilde{v}(x/a_v)$ for some potential $\tilde{v}$ with scattering length equal to one. By scaling, the scattering length of $v$ is given by $a_v$. Then the bounds in Theorem~\ref{thm:main} are uniform for $a_v \in (0,d]$ with $0 < d < \infty$.

\textbf {3.} 
The free energy $F_0(\beta,N,\omega)$ of the ideal Bose gas in the harmonic trap is 
of order $N\beta^{-1}\sim\omega (\beta \omega)^{-4}$, see Section~\ref{sec:freeenergyidealgas} below. The GP energy, on the other hand, is of order $N\omega \sim \omega (\beta \omega)^{-3} $ which is the scale up to which we have to control the free energy $F(\beta,N,\omega)$ of the interacting gas.

\textbf {4.} 
Theorem~\ref{thm:main} is stated and proven for the explicit choice of $\frac 14\omega^2 x^2$ as a trapping potential. This choice is mainly for notational simplicity, but is also motivated by the fact that the harmonic trap is the physically most relevant one. Rotational symmetry is not important, however. The treatment of an anisotropic trap requires only slight notational modifications and the interpretation of $\omega$ as the geometric mean of the principal frequencies of the parabolic potential.

\textbf {5.} 
The impressive cover picture of the first Bose-Einstein condensates in the July 1995 issue of Science, where the paper \cite{WieCor1995} appeared, shows the momentum distribution rather than the spatial distribution of the trapped gas. The momentum distribution of the thermal cloud is approximately an isotropic maxwellian of width $\sim \beta^{-1/2}$. The condensate momentum distribution, which is the modulus squared of the Fourier transform of the GP minimizer, is anisotropic because the trap was anisotropic. The width of the peak in momentum space is $\sim \omega^{1/2}$ and thus narrower than the thermal cloud by a factor
$(\beta\omega)^{1/2}$. Our Theorem confirms this picture rigorously for the first time.

\textbf {6.} 
The techniques used in the proof of Theorem~\ref{thm:main} carry over with moderate adjustments to the case of trapping potentials behaving as $|x|^{\alpha}$ with $\alpha<\infty$ for large $|x|$. The key point is that such potentials still lead to an asymptotic power law behavior of the eigenvalues of the related Schr\"odinger operator and cause a separation of length scales between the condensate and the thermal cloud.  It should be noted, however, that if $ \alpha>2$ the exponent 1/6 in \eqref{eq:scalesep} is replaced by a smaller exponent and a clear  separation of scales thus requires even larger values of $N$.\footnote{Even if  $N=10^6$ the scale separation in \eqref{eq:scalesep} is only by a factor 1/10.} If $\alpha\to \infty$ the whole system is confined in a box and the condensate and the thermal cloud are no longer spatially separated. Treating such systems will require a different approach from the one of the present paper. This is an important open problem because traps with very large $\alpha$ have recently become available in experiments \cite{Gau2013,Lopes2017}.

\textbf {7.} 
The fact that the transition temperature for BEC and the condensate fraction for the interacting gas stay the same as for the ideal gas relies essentially on the diluteness of the system,  expressed through the  scaling \eqref{eq:intpot}, and the $N\to\infty$ limit.
Under less restrictive conditions finite size corrections can be expected and have been seen in experiments \cite{Tamm_etal 2011}. Extending our results to capture these effects requires a proof of BEC beyond the GP limit, 
a difficult unsolved problem.
 
\textbf {8.} 
Since we work in the canonical ensemble, explicit expressions for the free energy $F_0(\beta,N,\omega)$, the one-particle density matrix $\gamma_{N,0}$ and the condensate fraction $N_0/N$ in the ideal Bose gas are not available. However, Theorem~\ref{thm:main} remains valid if these expressions are replaced by their corresponding grand canonical versions, which we recall in Section~\ref{sec:freeenergyidealgas}. This is due to the fact that the difference between the two ensembles is negligible in the limit of consideration here (see Section~\ref{sec:freeenergyidealgas} and the discussion in the Appendix for details).

\textbf {9.}
Theorem~\ref{thm:main} is also valid if we replace $F(\beta, N, \omega)$ by its grand canonical analogue $F^{\mathrm{gc}}(\beta, N, \omega)$ where a chemical potential $\mu$ is chosen such that the expected number of particles  equals $N$. We know from the Gibbs variational principle that $F^{\mathrm{gc}}(\beta, N, \omega) \leq F(\beta, N, \omega)$ holds. Hence the upper bound for the canonical free energy directly implies the upper bound for its grand canonical version. On the other hand, the proofs of the lower bound for the free energy and of the asymptotics of the one-particle density matrix are carried out in a way that is directly applicable to the grand canonical ensemble. 

\textbf {10.}
The strategy of the proof of Theorem~\ref{thm:main} also applies in two space dimensions with the obvious modifications, compare with \cite[Chapter~6.2]{Themathematicsofthebosegas}. 

\subsection{Supplementary 1: The scattering length}
\label{sec:scatteringlength}
Let us quickly summarize the basic facts about the scattering length. A more detailed discussion can be found in \cite[Appendix~C]{Themathematicsofthebosegas}. Our assumptions on $v$ guarantee that the zero energy scattering equation
\begin{equation}
-\Delta f(x) + \frac{1}{2} v(x) f(x) = 0 \quad \text{ with } \quad \lim_{|x| \to \infty} f(x) = 1
\label{eq:scatteringlength1}
\end{equation}  
has a unique solution. It satisfies
\begin{equation}
f(x) \simeq 1 - \frac{a}{|x|} \quad \text{ for } | x | \to \infty 
\label{eq:scatteringlength2}
\end{equation}
for some constant $a > 0$ which is called the scattering length of $v$. The scattering length has the natural interpretation of a combined measure for the range and the strength of the interaction potential $v$. If $v$ happens to be a hard core potential with range $r$ for example, one finds $a = r$, whereas for weak, integrable potentials $a\approx(8\pi)^{-1}\int_{\mathbb{R}^3} v(x) \text{d} x$.  For dilute quantum gases where collisions can be described in a low energy approximation, the leading order contribution to the scattering amplitude comes from $s$-wave scattering. In this approximation particles are scattered in every direction with the same probability and the scattering cross-section is given by $4 \pi a^2$.

In the case of nonnegative potentials, the scattering length can be characterized via the following variational principle. Denote by $X$ the set of functions in $H^1_{\textrm{loc}}(\mathbb{R}^3)$ with $\phi(x) \to 1$ for $|x| \to \infty$. Then the scattering length is given by
\begin{equation}
4 \pi a = \inf_{\phi \in X} \int_{\mathbb{R}^3} \left( \vert \nabla \phi(x) \vert^2 + \frac{1}{2} v(x) \vert \phi(x) \vert^2 \right) \text{d}x.
\label{eq:scatteringlength3}
\end{equation} 
Eq.~\eqref{eq:scatteringlength3} implies $8 \pi a \leq \int_{\mathbb{R}^3} v(x) \text{d}x$. By a trial state argument one can improve this inequality and show that it is strict if $v$ is not identically zero. 

\subsection{Supplementary 2: The chemical potential and the free energy  of the ideal Bose gas}
\label{sec:freeenergyidealgas}
For typical quantities of interest, as for example the free energy, there do not exist simple closed form expressions in the canonical ensemble. Nevertheless, in the thermodynamic limit as defined in Section~\ref{sec:lengthscales}, these quantities are very close to those computed in the grand canonical ensemble, which allows for explicit computations. We start by introducing several grand canonical quantities, and subsequently discuss their relation to their canonical versions.

The ideal Bose gas in the harmonic trap is described by the one-particle Hamiltonian
\begin{equation}
h=-\Delta + \frac{\omega^2 x^2}{4} - \frac{3 \omega}{2}.
\label{eq:idealbosegas1a}
\end{equation}
The expected number of particles in the condensate and the thermal cloud are given by
\begin{equation}
N_0^{\mathrm{gc}} = \frac{1}{e^{-\beta \mu_0}-1} \quad \text{ and } \quad N_{\mathrm{th}}^{\mathrm{gc}} = \sum_{n=1}^{\infty} \frac{g(n)}{e^{-\beta \mu_0} e^{\beta \omega n }-1},
\label{eq:idealbosegasparticlenumbers}
\end{equation}
respectively. Here $g(n)=(n+1)(n+2)/2$ is the degeneracy of the energy level $\omega n$ of $h$. For $\beta \omega \ll 1$ the sum in the above equation can be interpreted as a Riemann sum and one finds $N^{\mathrm{gc}}_{\mathrm{th}} \sim (\beta \omega)^{-3}$. The expected number of particles in the gas is $\overline{N} = N_0^{\mathrm{gc}} + N_{\mathrm{th}}^{\mathrm{gc}}$. It will be adjusted by the chemical potential $\mu_0$ such that $\overline{N} = N$ holds. To be more precise, we choose $\mu_0= \mu_0(\beta,N,\omega)$ such that
\begin{equation}
N = \sum_{n=0}^{\infty} \frac{g(n)}{e^{-\beta \mu_0} e^{\beta \omega n }-1}.
\label{eq:idealbosegas1}
\end{equation}
If $N_0^\mathrm{gc}\gg 1$  one has $-(\beta \mu_0)^{-1} \simeq N_0^{\mathrm{gc}}$. If $T<T_{\mathrm{c}}(1 - \epsilon)$ for some  $\epsilon > 0$, the chemical potential  behaves as $-\mu_0 \simeq T (N (1-(T/T_{\mathrm{c}})^3))^{-1}$, see Eq.~\eqref{eq:idealbosegas3}. On the other hand, for $T > T_{\mathrm{c}}(1 + \epsilon)$ one has $-\mu_0 \simeq \eta T$, where $\eta$ is  the unique solution of the equation
\begin{equation}
\left( \frac{T}{T_{\mathrm{c}}} \right)^3 \frac{1}{2} \int_{0}^{\infty} \frac{x^2}{e^{x+\eta}-1} \text{d}x = \zeta(3).
\label{eq:idealbosegas1c}
\end{equation}
The grand canonical free energy is given by  
\begin{equation}
F_0^{\mathrm{gc}}(\beta,N,\omega)=\frac{1}{\beta} \sum_{n=0}^{\infty} g(n) \ln\left( 1 - e^{\beta \mu_0} e^{-\beta \omega n} \right) + \mu_0 N.
\label{eq:gcfree} 
\end{equation}
In the same way as for the expected number of particles in the thermal cloud, the sum in the above equation can be interpreted as a Riemann sum and one finds $F_0^{\mathrm{gc}}(\beta,\mu_0,\omega) \sim -\omega (\beta \omega)^{-4}$, compare with \cite[Eqs.~(10.19)--(10.22)]{PitaevskiiStringari}. Note, however, that the error one makes by approximating the sum in Eq.~\eqref{eq:gcfree} by an integral is of order $\omega (\beta \omega)^{-3} \sim \omega N$ which is the order of the GP energy. Hence, we cannot do this replacement and have to work with the sum. The same is true in the case of $N^{\mathrm{gc}}_{\mathrm{th}}$. 

The canonical free energy $F_0$ is given by \eqref{eq:main2} with $v=0$ in $H_N$. We shall show in Corollary~\ref{cor:freeenergy} in the Appendix that $ | F_0(\beta,N,\omega) - F_0^{\mathrm{gc}}(\beta,N,\omega) | \leq T O( \ln N)$, so
\begin{equation}
F_0(\beta,N,\omega) = \frac{1}{\beta} \sum_{n=0}^{\infty} g(n) \ln\left( 1 - e^{\beta \mu_0} e^{-\beta \omega n} \right) + \mu_0 N + T O( \ln (N))
\label{eq:idealbosegas1b}
\end{equation}
with the last term much smaller than the main contribution in the thermodynamic limit. Moreover, Corollary~\ref{lem:particlenumbers} tells us that $\vert N_0 - N_0^{\mathrm{gc}} \vert \lesssim (\beta \omega)^{-3/2} (\ln N)^{1/2} + (\beta \omega)^{-1} \ln N$. The same bound holds for the expected numbers of particles in the thermal cloud.

\subsection{The proof strategy}
\label{sec:proofstrategy}
The rigorous mathematical implementation of the heuristic picture behind Theorem~\ref{thm:main} is technically rather involved.  For the convenience of the reader we describe here briefly the main steps before turning to the  proof in the remaining sections.

Section~\ref{sec:upperbound} contains an upper bound on the free energy $F(\beta,N,\omega)$ that has the correct asymptotic form \eqref{eq:mainresult4}. It utilizes the Gibbs variational principle \eqref{gibbsv}, so the task is to construct an appropriate trial state. The latter consists on the one hand of a pure state, describing the condensate particles, which are confined to a ball of some size $R$ with $\omega^{-1/2}\ll R \ll \omega^{-1}\beta^{-1/2}$, i.e., large compared to the oscillator length but small compared to the  length scale of the thermal cloud. On the other hand, the remaining particles, constituting the thermal cloud, will be described by a suitably modified Gibbs state confined to the complement of the ball. For the condensate, we can use the known zero-temperature results to obtain the GP energy. Even though the particle interactions affect the free energy to the order we are interested in only through the condensate, we cannot simply use the noninteracting Gibbs state as a trial state for the thermal cloud, since we want to allow for nonintegrable interaction potentials (e.g., having a hard-core) which can have infinite energy in that state. We thus  have to appropriately modify the trial state, avoiding configurations where the particles are too close. This creates some technical complications which have to be dealt with carefully. For various bounds, it turns out to be necessary to compare certain expressions for the ideal Bose in the canonical and grand canonical ensembles, respectively. The relevant estimates  are collected in Appendix A.

In Section~\ref{sec:lowerbound} we shall give a lower bound on the free energy, which together with the upper bound proves Eq.~\eqref{eq:mainresult4}. We shall use the technique of Fock-space localization to spatially divide the system into two, one confined to a ball of radius $R$ (chosen as above to satisfy $\omega^{-1/2}\ll R \ll \omega^{-1}\beta^{-1/2}$) and one confined to the complement. Inside the ball, the effect of the positive temperature is of lower order, and we can again utilize the known zero-temperature results to obtain a bound on the energy of these particles, as well as on the one-particle density matrix, which displays Bose-Einstein condensation into the GP minimizer. For the system in the complement of the ball, we can drop the interaction terms (using their positivity) to obtain the free energy of the ideal gas as a lower bound. 

To obtain information on the one-particle density matrix of the interacting Gibbs state (or approximate Gibbs state) we develop in Section~\ref{sec:densitymatrix} a novel lower bound on the free energy functional for an ideal Bose gas quantifying its coercivity. More precisely, in Lemma~\ref{lem:relativeentropy} we show that any approximate minimizer of the Gibbs free energy functional is, in a suitable sense, close to the actual minimizer. In combination with the result on Bose-Einstein condensation for the system inside the ball of radius $R$, this allows us to prove Eqs.~\eqref{eq:mainresult6} and~\eqref{eq:mainresult7}. 

Finally, Appendix~A collects certain properties of the ideal Bose gas that we need in our proofs. 

\section{Proof of the upper bound}
\label{sec:upperbound}
\subsection{The variational ansatz}
\label{subsec:ansatz}
In this section we construct a trial state $\Gamma_N$ whose free energy has the correct asymptotics \eqref{eq:mainresult4}. In the case of the ideal Bose gas in the harmonic trap, the characteristic length scale of the condensate is $\omega^{-1/2}$ while for the thermal cloud it is $\omega^{-1/2} (\beta \omega)^{-1/2}$ which is much larger in the limit we consider. The main idea of the proof is based on the expectation that in the GP limit this picture does not change if an interaction is turned on. What also does not change to leading order is the expected number of particles in the condensate and in the thermal cloud. Since the thermal cloud is therefore  much more dilute than the condensate, the free energy of the particles outside the condensate is not affected by the interaction to the same order of magnitude as the condensate. The following analysis makes this intuition precise. 

The first step in the construction of our trial state $\Gamma_N$ is to decompose  space into three  disjoint parts, a ball $B(R)$ with radius $R>0$, an annulus $A(R,R+\ell)$ with radii $R$ and $R+\ell$ and the complement of a ball with radius $R+\ell$. All those sets are assumed to be centered around zero. For later convenience we will refer to them as {\em Region~1, 2 and 3}, respectively. The length $R$ will be chosen such that $\omega^{-1/2} \ll R \ll \omega^{-1/2} (\beta \omega)^{-1/2}$, i.e.,  between the length scale of the condensate and the one of the thermal cloud. Choosing $R$ like this, we will be able to spatially separate the system into two parts, a condensate living in $B(R)$ and a thermal cloud living in $B(R+\ell)^{\mathrm{c}}$ without affecting each of them too much. In Region~$2$ there will be no particles in the trial state. The length $\ell \sim \omega^{-1/2} N^{-1}$ is chosen such that there is no hard core interaction between particles in Region~$1$ and particles in Region~$3$. 

The one-particle Hilbert space and the Fock space naturally decompose as 
\begin{equation}
L^2(\mathbb{R}^3) \cong  \underbrace{L^2(B(R))}_{\mathcal{H}_1} \oplus \underbrace{L^2(A(R,R+\ell))}_{\mathcal{H}_2} \oplus \underbrace{L^2(B(R+\ell)^{\mathrm{c}})}_{\mathcal{H}_3}
\label{eq:upperbound1}
\end{equation}
and $\mathcal{F}(L^2(\mathbb{R}^3)) \cong \mathcal{F}(\mathcal{H}_1) \otimes \mathcal{F}(\mathcal{H}_2) \otimes \mathcal{F}(\mathcal{H}_3)$, respectively. The heuristics in Section~\ref{sec:lengthscales} tells us that we can neglect the contribution from the thermal cloud in Region~$1$ since it is too dilute to contribute with a macroscopic number of particles. The condensate will be described by the ground state of the Hamiltonian
\begin{equation}
H^{\mathrm{D}}_{\leq R} = \sum_{i=1}^{N_0} \left( - \Delta^\mathrm{D}_{i,\leq R} + \frac{\omega^2 x_i^2}{4} - \frac{3 \omega}{2} \right) + \sum_{1 \leq i < j \leq N_0} v_N(x_i - x_j)
\label{eq:upperbound2}
\end{equation}
acting on $L_{\mathrm{sym}}^2(B(R)^{N_0})$,   the space of permutation symmetric square integrable functions depending on $N_0$ variables. Here, $\Delta^{\mathrm{D}}_{i,\leq R}$ denotes the Laplacian on $B(R)$ with Dirichlet boundary conditions, acting on the $i$-th particle, and  $N_0=N_0(\beta,N,\omega)$ is the expected number of particles in the condensate of an ideal Bose gas in the canonical ensemble. Strictly speaking $N_0$ is not necessarily an integer and we should rather choose $\lceil N_0 \rceil$, the smallest integer larger than  or equal to $N_0$, in the definition of $H^{\mathrm{D}}_{\leq R}$ instead. In the end, this will lead to $1/N$ corrections which do not cause any additional difficulties, however. In order not to complicate the presentation unnecessarily, we therefore assume that $N_0$ is an integer. By $\Psi^{\mathrm{D}}_{\leq R}$ we denote the unique  ground state wavefunction of $H^{\mathrm{D}}_{\leq R}$ with energy $E^{\mathrm{D}}_{\leq R}$. The above construction will allow us to use existing results for the ground state of a Bose gas in a trap. 

In Region~$3$ on the other hand, the condensate does not contribute to the free energy to leading order because its extension $\omega^{-1/2}$ is much smaller than $R$. To describe the thermal cloud, we define the  one-particle Hamiltonian 
\begin{equation}
h^{\mathrm{D}}_{\geq R+\ell} = - \Delta^{\mathrm{D}}_{\geq R + \ell} + \frac{\omega^2 x}{4} - \frac{3 \omega}{2},
\label{eq:upperbound3a}
\end{equation}
where $\Delta^{\mathrm{D}}_{i,\geq R + \ell}$ denotes the Laplacian on $B(R+\ell)^{\mathrm{c}}$ with Dirichlet boundary conditions. We also define the noninteracting $N_\mathrm{th}$-particle operator with energy cut-off $\Lambda$
\begin{equation}
H^{\mathrm{D},\Lambda}_{\geq R+\ell} = \sum_{i=1}^{N_{\mathrm{th}}} \mathds{1}\left( h^{\mathrm{D}}_{\geq R+\ell,i} \leq \Lambda \right) h^{\mathrm{D}}_{\geq R+\ell,i}
\label{eq:upperbound3b}
\end{equation}
acting on $L^2_{\mathrm{sym}}((B(R+\ell)^{\mathrm{c}})^{N_{\mathrm{th}}})$ with $N_{\mathrm{th}} = N - N_0$ and $h^{\mathrm{D}}_{\geq R+\ell,i}$  the operator in Eq.~\eqref{eq:upperbound3a} acting on the $i$-th particle. By $\mathds{1}( h^{\mathrm{D}}_{\geq R+\ell,i} \leq \Lambda )$ we denote the spectral projection onto the subspace of $\mathcal{H}_3$ where $h^{\mathrm{D}}_{\geq R+\ell,i}$ is at most $\Lambda$. The cut-off $\Lambda$ in Eq.~\eqref{eq:upperbound3b} is introduced for technical reasons, which will be explained in the text preceding Lemma~\ref{lem:denominators} in Subsection \ref{subsec:thermalcloud_upper} below. Let the many-particle projection $P_{N_{\mathrm{th}}}$ be defined by
\begin{equation}
P_{N_{\mathrm{th}}} = \prod_{i=1}^{N_{\mathrm{th}}} \mathds{1}\left( h^{\mathrm{D}}_{\geq R+\ell,i} \leq \Lambda \right).
\label{eq:upperbound3c}
\end{equation}
Its range consists of linear combinations of symmetrized products of eigenfunctions of $h^{\mathrm{D}}_{\geq R+\ell}$, where each of these one-particle functions has energy at most $\Lambda$. By
\begin{equation}
\Gamma_{\geq R + \ell}^{\mathrm{D},\Lambda} = \frac{e^{-\beta H^{\mathrm{D},\Lambda}_{\geq R+\ell}} P_{N_{\mathrm{th}}}}{\tr \left[e^{ - \beta H^{\mathrm{D},\Lambda}_{\geq R+\ell} }P_{N_{\mathrm{th}}} \right]}
\label{eq:upperbound4}
\end{equation}
we denote the canonical Gibbs state associated with the Hamiltonian $H^{\mathrm{D},\Lambda}_{\geq R+\ell}$. 

Since the interaction potential may include a hard core repulsion between the particles we have to add a correlation structure to the state $\Gamma_{\geq R + \ell}^{\mathrm{D},\Lambda}$. For this purpose we define the Jastrow-type function \cite{Jastrow}
\begin{equation}
F(x_1,\ldots,x_{N_{\mathrm{th}}}) = \prod_{1 \leq i < j \leq N_{\mathrm{th}}} f_b( x_i - x_j  ) \quad \text{ with } \quad f_b(x) = \begin{cases} f_0(|x|)/f_0(b) & \text{ for } |x| < b \\ 1 & \text{ for } |x| \geq b,  \end{cases}
\label{eq:upperbound5}
\end{equation}
where $b$ is a parameter to be determined and $f_0(|x|)$ is the unique solution of the zero-energy scattering equation \eqref{eq:scatteringlength1} with $v$ replaced by $v_N$. 
Since $f_0$ is an increasing function, $0 \leq f_b(x) \leq 1$ for all $x\in\mathbb{R}^3$. 
The parameter $b$ will be chosen to be larger than the scattering length $a_N$ but of the same order of magnitude.

We expand $\Gamma_{\geq R + \ell}^{\mathrm{D},\Lambda} = \sum_{\alpha = 1}^\infty \lambda_{\alpha} \vert \Psi_{\alpha} \rangle\langle \Psi_{\alpha} \vert$ where we choose the functions $\Psi_{\alpha}$ as symmetrized products of eigenfunctions of the one-particle Hamiltonian $h^{\mathrm{D}}_{\geq R +\ell}$ with energy at most $\Lambda$. The energy of $\Psi_{\alpha}$ is denoted by $E_{\alpha}$, that is, $H^{\mathrm{D}}_{\geq R + \ell} \Psi_{\alpha} = E_{\alpha} \Psi_{\alpha}$. The modified state $\tilde{\Gamma}_{\geq R + \ell}^{\mathrm{D},\Lambda}$ is defined as
\begin{equation}
\tilde{\Gamma}_{\geq R + \ell}^{\mathrm{D},\Lambda} = \sum_{\alpha = 1}^\infty \lambda_{\alpha} \vert \Phi_{\alpha} \rangle\langle \Phi_{\alpha} \vert \quad \text{ where } \quad \Phi_{\alpha} = \frac{F \Psi_{\alpha}}{\left\Vert F \Psi_{\alpha} \right\Vert}.
\label{eq:upperbound6}
\end{equation}
Here and in the following, $\Vert \Psi \Vert$ denotes the $L^2$-norm of $\Psi$. 

After these preparations we can now finally define our trial state $\Gamma_N$ on $\mathcal{F}(\mathcal{H}_1) \otimes \mathcal{F}(\mathcal{H}_2) \otimes \mathcal{F}(\mathcal{H}_3)$ to be
\begin{equation}
\Gamma_N = \vert \Psi^{\mathrm{D}}_{\leq R} \rangle\langle \Psi^{\mathrm{D}}_{\leq R} \vert \otimes \vert \Omega \rangle\langle \Omega \vert \otimes \tilde{\Gamma}_{\geq R+ \ell}^{\mathrm{D},\Lambda},
\label{eq:upperbound7}
\end{equation}
where $\Omega$ denotes the Fock space vacuum in $\mathcal{F}(\mathcal{H}_2)$. 

Let
\begin{equation}
H^{\mathrm{D}}_{\geq R + \ell} = \sum_{i=1}^{N_{\mathrm{th}}} h^{\mathrm{D}}_{\geq R+\ell,i}
\label{eq:upperbound3}
\end{equation}
be the noninteracting $N_\mathrm{th}$-particle Hamiltonian in the region $B(R+\ell)^{\mathrm{c}}$ without the cut-off. Because the first two factors in Eq.~\eqref{eq:upperbound7} do not contribute to the entropy of $\Gamma_N$, its free energy with respect to the original Hamiltonian \eqref{eq:main1} is given by
\begin{equation}
\tr \left( H_N \Gamma_N \right) - T S(\Gamma_N) = E^{\mathrm{D}}_{\leq R} + \tr\left[ \left( H^{\mathrm{D}}_{\geq R + \ell} + V_{33} \right) \tilde{\Gamma}_{\geq R + \ell}^{\mathrm{D},\Lambda} \right] - T S\left( \tilde{\Gamma}_{\geq R + \ell}^{\mathrm{D},\Lambda} \right) + \tr\left( V_{13} \Gamma_N \right).
\label{eq:upperbound8}
\end{equation}
Here $V_{ij}$ denotes the interaction between Regions $i$ and $j$. The remaining part of this section will be devoted to finding an appropriate upper bound to the right-hand side of Eq.~\eqref{eq:upperbound8}, which is an upper bound for $F(\beta,N,\omega)$ by the Gibbs variational principle \eqref{gibbsv}. In order to simplify the notation, we will from now on replace $R+\ell$ everywhere by $R$. Since $\ell \ll R$ the number $\ell$ does not enter the proofs explicitly, except when we bound the interaction energy between the condensate and the thermal cloud.

In the rest of the proof of the upper bound we assume that $(\beta\omega)^{-1} \lesssim N^{1/3}$ (i.e., $T\lesssim T_c$), 
as well as $\Lambda \gg T$ and $\omega^{-1/2} \ll R \leq \lambda \omega^{-1} \beta^{-1/2}$ for some $\lambda >0$ that we choose small enough. 

\subsection{Preparatory lemmas}
\label{subsec:prep}
The thermal cloud in Region~3, that is, in $B(R)^{\mathrm{c}}$, is described by an ideal Bose gas with the one-particle Hamiltonian $\mathds{1}(h^{\mathrm{D}}_{\geq R} \leq \Lambda) h^{\mathrm{D}}_{\geq R}$. We would like to relate its canonical free energy to that of the gas living in all of $\mathbb{R}^3$  and without a cut-off. To that end, we will first compare it to the grand canonical free energy with the help of Corollary~\ref{cor:freeenergy}. Here explicit formulas are available which allow us to quantify the change in energy caused by the Dirichlet boundary conditions at $\partial B(R)$ and by the energy cut-off $\Lambda$. As a preparation we state and prove in this subsection four Lemmas.

The first one, Lemma~\ref{lem:traces1},  is a general statement  allowing to compare traces of functions of Schr\"odinger operators with different boundary conditions. 
Lemma~\ref{lem:traces2} 
estimates the differences between traces of functions of Schr\"odinger operators with Neumann boundary conditions acting on $L^2(B(R)^{\mathrm{c}})$ and those acting on $L^2(\mathbb{R}^3)$ without boundary conditions. Together these two lemmas are used in the sequel to quantify the difference between the grand canonical free energy  of the system living in Region~3 and that of the system living in $\mathbb{R}^3$, and also the difference of the expected particle numbers. Lemmas~\ref{lem:traces1} and~\ref{lem:traces2}  also enter the proof of Lemma~\ref{lem:chemicalpotential} which concerns the effects of the boundary condition at $\partial B(R)$ and the cut-off $\Lambda$ on the chemical potential. 

To show that the interaction energy in the thermal cloud and between the thermal cloud and the condensate is of lower order we need an estimate on the $L^{\infty}$-norm of the density of the canonical ideal gas in Region~3. Using Proposition~\ref{cor:density}, we can estimate the canonical density in terms of the grand canonical density. To close the argument, we need a bound on the $L^{\infty}$-norm of the latter  showing that the system is dilute in a suitable sense. Such a bound is given in Lemma~\ref{lem:densitybound} whose proof uses also Lemma~\ref{lem:chemicalpotential}.
 
\begin{lemma}
\label{lem:traces1}
Let $f \in \mathcal{C}^2((0,\infty),\mathbb{R})$ be a convex, monotone decreasing and nonnegative function. We assume $|f'(x)| \lesssim x^{-3}$ for $x \rightarrow 0$ as well as
\begin{equation}
\int_0^{\infty}  f(x) \left( x^{1/2} + x^2 \right) \text{d}x < \infty \quad \text{ and } \quad \int_0^{\infty} | f'(x) | x^{2} \text{d}x < \infty.
\label{eq:lemupperbound10a}
\end{equation}
Denote by $h^{\mathrm{N/D}}_{\geq R}$ the harmonic oscillator Hamiltonian $-\Delta_{\geq R}^{\mathrm{N/D}} + \frac{\omega ^2 x^2}{4} - \frac{3 \omega}{2}$ acting on functions in $L^2(B(R)^{\mathrm{c}})$ with  Neumann/Dirichlet boundary conditions at $\partial B(R)$ and choose $\mu$ with $ \mu \leq C \omega$ for some $C>0$. We then have
\begin{equation}
\tr\left[ f(\beta (h_{\geq R}^{\mathrm{N}}-\mu)) \right] \leq \tr\left[ f \left( \beta (h_{\geq R}^{\mathrm{D}} -\mu) \right) \right] + O\left( \frac{1}{\beta^2 \omega^3 R^2} \right)  + O\left( \frac{R^3}{\beta^{3/2}} \right). 
\label{eq:lemupperbound10}
\end{equation} 
\end{lemma}
\begin{proof}
By the assumptions on $f$ we can write
\begin{equation}
f(x)= \int_0^{\infty}f^{''}(E)[E-x]_+ \text{d}E,
\label{eq:lemupperbound11b}
\end{equation}
where $[x]_{+}=\max\lbrace x, 0 \rbrace$. Since $f$ is convex, $f^{''}$ is nonnegative.

Let $j_1,j_2 \in \mathcal{C}^{\infty}(\mathbb{R}^3)$ be such that $j_1(x)^2+j_2(x)^2 = 1$ for all $x \in \mathbb{R}^3$. We further assume that $j_1(x)$ equals one for $x \in B(R)$ and zero for $x \in B(2R)^c$, as well as $| \nabla j_1(x) |^2 + | \nabla j_2(x) |^2 \leq 3 R^{-2}$. An application of the IMS localization formula (see e.g. \cite{SimonSchroedingerOps}) and the inequality $j_i(x) \mathds{1}(h^{\mathrm{N}}_{\geq R} \leq E) j_i(x) \leq \mathds{1}$ tell us that
\begin{align}
\tr \left[ E - \beta \left(h^{\mathrm{N}}_{\geq R}-\mu \right) \right]_{+} &= \sum_{i = 1}^2 \tr \left[ j_{i} \left( E - \beta \left(h_{\geq R}^{\mathrm{N}}-\mu \right) + \sum_{l=1}^2 (\nabla j_{l})^2 \right) j_{i} \mathds{1}\left( \beta \left( h^{\mathrm{N}}_{\geq R} - \mu \right) \leq E \right) \right] \nonumber \\
&\leq \tr \left[ \chi_1 \left( E - \beta \left(h^{\mathrm{N,D}}_{\geq R ,\leq 2R}-\mu \right) + \sum_{l=1}^2 (\nabla j_{l})^2 \right) \chi_1 \right]_{+}  \label{eq:lemupperbound12}  \\
&\hspace{0.85cm}+ \tr \left[ \chi_2 \left( E - \beta \left(h_{\geq R}^{\mathrm{D}}-\mu \right) + \sum_{l=1}^2 (\nabla j_{l})^2 \right) \chi_2 \right]_{+} \nonumber
\end{align}
holds. Here $\chi_i$ denotes the characteristic function of the support of $j_i$ and $h^{\mathrm{N,D}}_{\geq R ,\leq 2R}$ is the harmonic oscillator Hamiltonian in the annulus $A(R,2R)$ with Neumann boundary conditions on $\partial B(R)$ and Dirichlet boundary conditions on $\partial B(2R)$. To arrive at Eq.~\eqref{eq:lemupperbound12}, we used that the cut-off functions $j_i$ introduce additional Dirichlet boundary conditions for the operators under the trace. Together with Eq.~\eqref{eq:lemupperbound11b} and $| \nabla j_1(x) |^2 + | \nabla j_2(x) |^2 \leq 3 R^{-2}$, we conclude that
\begin{equation}
\tr\left[ f(\beta (h_{\geq R}^{\mathrm{N}}-\mu)) \right] \leq \tr\left[ f \left( \beta \left(h_{\geq R}^{\mathrm{D}} - \mu - 3 R^{-2} \right) \right) \right] + \tr\left[ f\left( \beta \left(  h^{\mathrm{N,D}}_{\geq R,\leq 2R} - \mu - 3 R^{-2} \right) \right) \right] \label{eq:lemupperbound13}
\end{equation}
holds.

Let us continue with the first term on the right-hand side of the above equation. Using the convexity of $f$, we see that
\begin{align}
\tr\left[ f \left( \beta \left(h_{\geq R}^{\mathrm{D}}-\mu - 3 R^{-2} \right) \right) \right] & \leq  \tr\left[ f \left( \beta (h_{\geq R}^{\mathrm{D}} -\mu) \right) \right] \label{eq:lemupperbound14} \\
& \quad -3 \beta R^{-2} \tr\left[ f' \left( \beta \left( h^{\mathrm{D}}_{\geq R} - \mu - 3 R^{-2} \right) \right) \right]. \nonumber
\end{align}
In order to give an upper bound for the second term on the right-hand side of Eq.~\eqref{eq:lemupperbound14}, we note that the $\alpha$-th eigenvalue $e_{\alpha}(h_{\geq R}^{\mathrm{D}})$ of $h_{\geq R}^{\mathrm{D}}$ can be bounded from below by $e_{\alpha}(h_{\geq R}^{\mathrm{D}}) \geq \max\lbrace e_{\alpha}(h), \omega^2 R^2/4 -3 \omega/2 \rbrace$. Here $e_{\alpha}(h)$ denotes the $\alpha$-th eigenvalue of the harmonic oscillator Hamiltonian $h$ in \eqref{eq:idealbosegas1a} acting on $L^2(\mathbb{R}^3)$. We thus have
\begin{align}
\tr\left[ f' \left( \beta \left(h^{\mathrm{D}}_{\geq R} - \mu - 3 R^{-2} \right) \right) \right] &\geq \sum_{\alpha=0}^{\alpha_0} f'\left(\beta \left( \frac{\omega^2 R^2}{4} - \frac{3\omega}{2} - \mu - 3 R^{-2} \right) \right) \label{eq:lemupperbound15} \\
&\hspace{2cm} + \sum_{\alpha > \alpha_0} f'\left(\beta \left(e_{\alpha}(h) -\mu - 3 R^{-2} \right) \right) \nonumber
\end{align}
for any $\alpha_0 \in \mathbb{N}$. If we choose $\alpha_0 \sim 1$ such that $e_{\alpha_0}(h) - \mu - 3 R^{-2} \geq \omega$ and use that $|f'(x)| \lesssim x^{-3}$ for $x \rightarrow 0$, 
we see that the first term on the right-hand side of Eq.~\eqref{eq:lemupperbound15} is at most of the order  $(\beta \omega^2 R^2)^{-3}$. Our assumptions on $R$ imply $(\beta \omega^2 R^2)^{-3} \ll (\beta \omega)^{-3}$. To treat the second term, we recall that the eigenvalue $\omega n$ of $h$ is $g(n)$-fold degenerate, where $g(n) = (n+2)(n+1)/2$. Hence, for an appropriately chosen integer $m_0 > 0$, we can write the second term on the right-hand side of Eq.~\eqref{eq:lemupperbound15} as
\begin{align}
\sum_{\alpha>\alpha_0} f'\left(\beta \left( e_{\alpha}(h) - \mu - 3 R^{-2} \right) \right) &= \sum_{n > m_0} g(n) f'\left( \beta \left(\omega n - \mu - 3 R^{-2} \right) \right) \label{eq:lemupperbound16} \\
&\gtrsim \frac{1}{(\beta \omega)^3} \int_{0}^{\infty} x^{2} f'(x) \text{d}x. \nonumber
\end{align}
To obtain the last line, we used that the sum in the line above can be interpreted as a Riemann sum approximating the integral in the last line, as well as the bound $e_{\alpha}(h) - \mu -  3 R^{-2} \geq \omega$. We now collect the results of Eqs.~\eqref{eq:lemupperbound14}--\eqref{eq:lemupperbound16} and obtain
\begin{equation}
\tr\left[ f \left( \beta \left(h_{\geq R}^{\mathrm{D}}-\mu - 3 R^{-2} \right) \right) \right] \leq \tr\left[ f \left( \beta \left(h_{\geq R}^{\mathrm{D}}-\mu \right) \right) \right] + O\left( \frac{1}{\beta^2 \omega^3 R^2} \right) \label{eq:lemupperbound16b}
\end{equation}
as an upper bound on the first term on the right-hand side of Eq.~\eqref{eq:lemupperbound13}.

It remains to give a bound on the second term on the right-hand side of Eq.~\eqref{eq:lemupperbound13}. 
Using the Weyl asymptotics  \cite[Satz~XI]{weyl}  for  the  eigenvalues of the Laplacian in $A(R,2R)$,  one sees that
\begin{equation}
e_{\alpha} \left(h^{\mathrm{N,D}}_{\geq R ,\leq 2R}\right) - \mu - 3  R^{-2} \geq \frac{\omega^2 R^2}{8} + \frac{C \alpha^{2/3}}{R^2} \label{eq:lemupperbound17}
\end{equation}
holds for some appropriately chosen constant $C>0$. This allows us to estimate the second term on the right-hand side of Eq.~\eqref{eq:lemupperbound13} by
\begin{align}
\tr\left[ f\left( \beta \left( h^{\mathrm{N,D}}_{\geq R,\leq 2R}-\mu - 3R^{-2} \right) \right) \right] &\leq \sum_{\alpha = 0}^{\infty} f\left( \beta \left( \frac{\omega^2 R^2}{8} + \frac{C \alpha^{2/3}}{R^2} \right) \right) \label{eq:lemupperbound18} \\
&\lesssim \frac{R^3}{\beta^{3/2}} \int_0^{\infty} x^{1/2} f(x) \text{d}x. \nonumber
\end{align} 
In combination, Eqs.~\eqref{eq:lemupperbound13},~\eqref{eq:lemupperbound16b} and~\eqref{eq:lemupperbound18} yield Eq.~\eqref{eq:lemupperbound10}. 
\end{proof}
\begin{lemma}
\label{lem:traces2}
Let $f : \mathbb{R}_+ \to \mathbb{R}$ be a nonnegative and monotone decreasing measurable function. We also assume that
\begin{equation}
\int_0^{\infty}  f(x) \sqrt{x} \text{d}x < \infty \quad \text{ as well as } \quad f(x) \lesssim x^{-1} \text{ for } x \to 0.
\label{eq:lemupperbound110}
\end{equation}
Denote by $h^{\mathrm{N}}_{\geq R}$ the harmonic oscillator Hamiltonian $-\Delta_{\geq R}^{\mathrm{N}} + \frac{\omega ^2 x^2}{4} - \frac{3 \omega}{2}$ acting on functions in $L^2(B(R)^{\mathrm{c}})$ with Neumann boundary conditions at $\partial B(R)$ and choose $\mu \leq c \omega$ with $c<1$. Then
\begin{equation}
\tr\left[ f(\beta (h_{\geq R}^{\mathrm{N}}-\mu)) \right] \geq \tr\left[ \mathds{1}\left( h \geq \omega \right)  f \left( \beta (h - \mu) \right) \right] - O\left( \frac{R^3}{\beta^{3/2}} \right)
\label{eq:lemupperbound111}
\end{equation} 
where $h$ is defined in \eqref{eq:idealbosegas1a}.
\end{lemma}
\begin{proof}
Denote by $h^{\mathrm{N}}_{\leq R}$ the harmonic oscillator Hamiltonian inside $B(R)$ with Neumann boundary conditions. Using the Weyl asymptotics \cite[Satz~XI]{weyl} of the eigenvalues of the Neumann Laplacian $-\Delta^{\mathrm{N}}_{\leq R}$ inside $B(R)$, we see that there exists a $C>0$ such that 
\begin{equation}
e_{\alpha}\left(h^{\mathrm{N}}_{\leq R} \right) \geq C \frac{\alpha^{2/3}}{R^2} - \frac{3 \omega}{2}
\label{eq:lemupperbound112}
\end{equation}
for $\alpha \geq 0$. Choose $\alpha_0$ to be the smallest positive integer for which the right-hand side of Eq.~\eqref{eq:lemupperbound112} is larger than or equal to $\omega + \mu$. Our assumption $\mu \leq c \omega$  implies $\alpha_0 \lesssim \omega^{3/2} R^3$. Let $h^{\mathrm{N}}_R = h^{\mathrm{N}}_{\leq R} \oplus h^{\mathrm{N}}_{\geq R}$. Using $e_{\alpha}(h^{\mathrm{N}}_R) \leq e_{\alpha}(h)$ and the monotonicity of $f$, the trace on the right-hand side of Eq.~\eqref{eq:lemupperbound111} can be bounded from above by
\begin{equation}
\sum_{\alpha=1}^{\infty} f\left( \beta \left( e_{\alpha}(h) - \mu \right) \right) \leq \sum_{\alpha=1}^{\alpha_0} f\left( \beta \left( e_{\alpha}(h) - \mu \right) \right) + \sum_{\alpha > \alpha_0}  f\left( \beta \left( e_{\alpha} \left(h^{\mathrm{N}}_R \right) - \mu \right) \right).
\label{eq:lemupperbound113}
\end{equation}
The asymptotic behavior \eqref{eq:lemupperbound110}  of $f$ at $0$ and $\mu \leq c \omega$ with $c<1$ imply that the first term on the right-hand side of the above equation can be bounded from above by $\alpha_0 f(\beta(\omega - \mu)) \lesssim \omega^{1/2} \beta^{-1} R^3$. Since $\omega^{1/2} \beta^{-1} \ll \beta^{-3/2}$, this error term is much smaller than $R^3 \beta^{-3/2}$. 

Next, we investigate the second term on the right-hand side of Eq.~\eqref{eq:lemupperbound113}. The $\alpha_0$-th eigenvalue of $h^{\mathrm{N}}_R$ is bounded from above by
\begin{equation}
e_{\alpha_0}( h^{\mathrm{N}}_R ) \leq e_{\alpha_0}( h ) \lesssim \omega \alpha_0^{1/3} \lesssim \omega^{3/2} R.
\label{eq:lemupperbound114}
\end{equation}
To obtain the second inequality, we used that the eigenvalue $\omega n$ of $h$ is $(n+1)(n+2)/2$-fold degenerate. On the other hand, $e_{\alpha}(h^{\mathrm{N}}_{\geq R}) \geq \frac{\omega^2 R^2}{4} - \frac{3 \omega}{2}$ which is much larger than the right-hand side of Eq.~\eqref{eq:lemupperbound114} by the assumption $\omega^{-1/2} \ll R$. Hence, the second term on the right-hand side of Eq.~\eqref{eq:lemupperbound113} can be written as
\begin{equation}
\sum_{\alpha > \alpha_0}  f\left( \beta \left( e_{\alpha} \left(h^{\mathrm{N}} \right) - \mu \right) \right) = \sum_{\alpha = 0}^{\infty}  f\left( \beta \left( e_{\alpha} \left(h^{\mathrm{N}}_{\geq R} \right) - \mu \right) \right) + \sum_{\alpha > \alpha_0}  f\left( \beta \left( e_{\alpha} \left(h^{\mathrm{N}}_{\leq R} \right) - \mu \right) \right). 
\label{eq:lemupperbound115}
\end{equation}
It remains to estimate the second term on the right-hand side of Eq.~\eqref{eq:lemupperbound115}. We invoke Eq.~\eqref{eq:lemupperbound112} and the monotonicity of $f$ to find 
\begin{align}
\sum_{\alpha > \alpha_0}  f\left( \beta \left( e_{\alpha} \left(h^{\mathrm{N}}_{\leq R} \right) - \mu \right) \right) &\leq \sum_{\alpha > \alpha_0}  f\left( \beta \left( \frac{C\alpha^{2/3}}{R^2} - \frac{3 \omega}{2} - \mu \right) \right) \label{eq:lemupperbound116} \\
&\lesssim \frac{R^3}{\beta^{3/2}} \int_0^{\infty} f(x) \sqrt{x} \,\text{d}x. \nonumber
\end{align}
To arrive at the second line, we used $\mu \leq c \omega$ and the fact that the sum in the first line on the right-hand side can be interpreted as a Riemann sum approximating the integral in the second line. This proves the claim.
\end{proof}
\begin{lemma}
\label{lem:chemicalpotential}
Define $\mu$ via the equation
\begin{equation}
\tr \left[\mathds{1}\left(h^{\mathrm{D}}_{\geq R} \leq \Lambda\right)  \frac{1}{e^{\beta \left(h^{\mathrm{D}}_{\geq R} - \mu) \right)} - 1} \right] = N_{\mathrm{th}}
\label{eq:lemmamu1}
\end{equation}
(with $N_\mathrm{th}$ defined after Eq.~\eqref{eq:rhocloud}). 
Then
\begin{equation}
0 \leq \mu - \mu_0 \lesssim \omega \left( \frac{1}{\omega R^2} + \beta^{1/2} \omega^{2} R^3 + \frac{e^{-\beta \Lambda /4 }}{\beta \omega} + (\beta \omega)^2 \left(N_{\mathrm{th}} - N_{\mathrm{th}}^{\mathrm{gc}} \right) \right)
\label{eq:lemmamu2}
\end{equation}
holds. Here $N_{\mathrm{th}}^{\mathrm{gc}}$ denotes the expected number of particles in the grand canonical thermal cloud, defined in \eqref{eq:idealbosegasparticlenumbers}. 
\end{lemma} 
\begin{proof}
Since $e_{\alpha}(h) \leq e_{\alpha}(h^{\mathrm{D}}_{\geq R})$, the chemical potential can only increase if we cut out $B(R)$ and impose Dirichlet boundary conditions at $\partial B(R)$. Since the cut-off $\Lambda$ also causes the chemical potential to increase, we have $\mu \geq \mu_0$.

In order to find an upper bound on $\mu - \mu_0$, we start by estimating the influence of the cut-off $\Lambda$ and write
\begin{equation}
\tr \left[ \mathds{1}\left(h^{\mathrm{D}}_{\geq R} > \Lambda\right)  \frac{1}{e^{\beta \left(h^{\mathrm{D}}_{\geq R} - \mu) \right)} - 1} \right] = \sum_{\substack {\alpha \in \mathbb{N} : \\ e_{\alpha}(h^{\mathrm{D}}_{\geq R}) > \Lambda}} \frac{1}{e^{\beta \left(e_{\alpha}\left( h^{\mathrm{D}}_{\geq R} \right) - \mu) \right)} - 1}  \label{eq:lemmamu1b}
\end{equation}
To be able to proceed, we need a rough upper bound on the chemical potential. Certainly it cannot be larger than the lowest eigenvalue of $h^{\mathrm{D}}_{\geq R}$. 
A simple trial state argument gives an upper bound $\lesssim \omega^2 R^2$  on this lowest eigenvalue and thereby on $\mu$. Since $R \leq \lambda \omega^{-1} \beta^{-1/2}$, this implies $\mu \leq \Lambda/4$ for $\Lambda/T$ large enough. Hence
\begin{equation}
\frac{1}{e^{\beta \left(e_{\alpha}\left( h^{\mathrm{D}}_{\geq R} \right) - \mu \right)} - 1} \leq 2 e^{- \beta \Lambda/4} e^{-\beta e_{\alpha} \left( h^{\mathrm{D}}_{\geq R} \right) /2}
\label{eq:lemmamu1c}
\end{equation}
which holds as long as $(1-e^{-3\beta \Lambda/4})^{-1} \leq 2$ and $e_{\alpha}( h^{\mathrm{D}}_{\geq R} ) \geq \Lambda$. Using Eq.~\eqref{eq:lemmamu1c}, Eq.~\eqref{eq:lemmamu1b},  and the inequality $e_{\alpha}( h^{\mathrm{D}}_{\geq R} ) \geq e_{\alpha}( h )$ with $h$ defined in Eq.~\eqref{eq:idealbosegas1a} we see that
\begin{equation}
\tr \left[ \mathds{1}\left(h^{\mathrm{D}}_{\geq R} > \Lambda\right)  \frac{1}{e^{\beta \left(h^{\mathrm{D}}_{\geq R} - \mu) \right)} - 1} \right] \leq  2 e^{- \beta \Lambda/4} \tr \left[  e^{-\beta h^{\mathrm{D}}_{\geq R} /2} \right] \leq 2 e^{- \beta \Lambda/4} \tr \left[  e^{-\beta h/2} \right]
\label{eq:lemmamu1d}
\end{equation}
holds. The trace on the right-hand side of the above equation can be bounded by a constant times $(\beta \omega)^{-3}$, and hence
\begin{equation}
\tr \left[ \frac{1}{e^{\beta \left(h^{\mathrm{D}}_{\geq R} - \mu) \right)} - 1} \right] - \tr \left[ \mathds{1}\left(h^{\mathrm{D}}_{\geq R} \leq \Lambda\right)  \frac{1}{e^{\beta \left(h^{\mathrm{D}}_{\geq R} - \mu) \right)} - 1} \right] \lesssim \frac{e^{-\beta \Lambda/4}}{(\beta \omega)^3} \,,
\label{eq:lemmamu1f}
\end{equation} 
which estimates the effect of the cut-off $\Lambda$.

It remains to quantify the influence of the boundary conditions at $\partial B(R)$. Together with Eq.~\eqref{eq:lemmamu1} and Eq.~\eqref{eq:lemmamu1f}, an application of Lemma~\ref{lem:traces1} with the choice $f(x)=(e^x-1)^{-1}$ tells us that
\begin{equation}
N_{\mathrm{th}} \geq \tr \left[ \frac{1}{e^{\beta(h^{\mathrm{N}}_{\geq R} - \mu)}-1}  \right] - O\left( \frac{e^{-\beta \Lambda/4}}{(\beta \omega)^3} \right) - O\left( \frac{1}{\beta^2 \omega^3 R^2} \right) - O\left( \frac{R^3}{\beta^{3/2}} \right). \label{eq:lemmamu6}
\end{equation}
Let us have a closer look at the first term on the right-hand side of Eq.~\eqref{eq:lemmamu6}. By convexity of $x\mapsto (e^x-1)^{-1}$ and $e^x/(e^x-1)^2 \geq 1/(e^x-1)$ for $x>0$, we have
\begin{align}
\tr \left[ \frac{1}{e^{\beta(h^{\mathrm{N}}_{\geq R} - \mu)}-1} - \frac{1}{e^{\beta(h^{\mathrm{N}}_{\geq R} - \mu_0)}-1} \right] &\geq \beta(\mu-\mu_0) \tr\left[ \frac{ e^{\beta(h^{\mathrm{N}}_{\geq R}-\mu_0)}}{\left(e^{\beta(h^{\mathrm{N}}_{\geq R} - \mu_0)}-1\right)^2} \right] \label{eq:lemmamu7} \\
&\geq \beta (\mu - \mu_0) \tr\left[  \frac{ 1 }{e^{\beta(h^{\mathrm{N}}_{\geq R} - \mu_0)}-1} \right]. \nonumber
\end{align}
Using that $\mu_0 <0$, we know from Lemma~\ref{lem:traces2} with the choice $f(x)=(e^x-1)^{-1}$  that
\begin{equation}
\tr\left[ \frac{ 1 }{e^{\beta(h^{\mathrm{N}}_{\geq R} - \mu_0)}-1}\right] \geq \tr\left[ \mathds{1}\left( h \geq \omega \right)  f \left( \beta (h - \mu_0) \right) \right] - O\left( \frac{R^3}{\beta^{3/2}} \right)
\label{eq:lemmamu8}
\end{equation}
holds. Note that the trace on the right-hand side of the above equation equals $N_{\mathrm{th}}^{\mathrm{gc}} \sim (\beta \omega)^{-3}$, the expected number of particles in the grand canonical thermal cloud. It dominates the error term in \eqref{eq:lemmamu8} if $R \leq \lambda \omega^{-1} \beta^{-1/2}$ with $\lambda$ small enough.  
In combination, Eqs.~\eqref{eq:lemmamu6}--\eqref{eq:lemmamu8} thus imply \eqref{eq:lemmamu2}. 
\end{proof}

\begin{lemma}
\label{lem:densitybound}
Denote by  
\begin{equation}
\varrho^{\mathrm{D},\Lambda}_{\geq R}(x) = \left[ \mathds{1}\left( h^{\mathrm{D}}_{\geq R} \leq \Lambda \right) \frac{1}{e^{\beta \left( h^{\mathrm{D}}_{\geq R} - \mu \right)}-1} \right](x,x)
\label{eq:densitybound1}
\end{equation}
the one-particle density in the grand canonical ensemble in Region~3 with cut-off $\Lambda$ and let the chemical potential $\mu$ be chosen as in \eqref{eq:lemmamu1} such that
\begin{equation}
\int_{\mathbb{R}^3} \varrho^{\mathrm{D},\Lambda}_{\geq R}(x) \text{d}x = N_{\mathrm{th}}.
\label{eq:densitybound2}
\end{equation}
Assume that $e^{-\beta \Lambda/4} \lesssim \beta \omega$ and  $\beta \omega \ln N \lesssim 1$ holds. Then, for $\omega^{-1/2} \ll R \leq \lambda \omega^{-1} \beta^{-1/2}$ with $\lambda >0$ small enough, 
\begin{equation}
\sup_{x \in \mathbb{R}^3} \varrho^{\mathrm{D},\Lambda}_{\geq R}(x) \lesssim \frac{1}{\beta^{3/2}}\,.
\label{eq:densitybound3}
\end{equation}
\end{lemma}

\begin{proof}
We start by noting that we obtain an upper bound on $\varrho^{\mathrm{D},\Lambda}_{\geq R}(x)$ if we drop the cut-off $\Lambda$ on the right-hand side of Eq.~\eqref{eq:densitybound1}. To be able to continue, we need an upper bound on the chemical potential $\mu$ which we are going to derive now. Lemma~\ref{lem:chemicalpotential} tells us that $\mu - \mu_0$ is bounded from above by the right-hand side of Eq.~\eqref{eq:lemmamu2}. Since $\mu_0 < 0$ the same expression also bounds $\mu$, that is, 
\begin{equation}
\mu \lesssim  \frac{1}{ R^2} + \beta^{1/2} \omega^{3} R^3 + \frac{e^{-\beta \Lambda /4 }}{\beta } + \beta^2 \omega^3 \left( N_{\mathrm{th}} - N_{\mathrm{th}}^{\mathrm{gc}} \right) \,. \label{eq:densitybound3b}
\end{equation}
The first term on the right-hand side of Eq.~\eqref{eq:densitybound3b} is much smaller than $\omega$ because $R \gg \omega^{-1/2}$. The second term  is smaller than $\lambda \omega^2 R^2$ because $R \leq \lambda \omega^{-1} \beta^{-1/2}$. For the third term we have $\beta^{-1} e^{-\beta \Lambda /4} \ll \omega^2 R^2$ by assumption and with Lemma~\ref{lem:particlenumbers} the fourth term can be estimated by $\beta^2 \omega^3 \vert N_{\mathrm{th}} - N_{\mathrm{th}}^{\mathrm{gc}} \vert \lesssim \omega (\beta \omega \ln N)^{1/2} + \omega (\beta \omega \ln N)$. This is much smaller than $\lambda \omega^2 R^2$ because of $\omega^2 R^2 \gg \omega$ and our assumption $\beta \omega \ln N \lesssim 1$. We conclude that $\mu \lesssim \lambda \omega^2 R^2$ holds. 

Let $V(x) = \frac{\omega^2 x^2}{4} - \frac{3 \omega}{2} - \mu$. By choosing $\lambda$ small enough, the upper bound on $\mu$ above allows to conclude that $V(x) \geq 0$ holds for $|x|\geq R$. Using the Feynman-Kac formula, see e.g. \cite{SimonFunct,BratelliRobinson2}, we have
\begin{align}
\left[ \frac{1}{e^{\beta \left( h^{\mathrm{D}}_{\geq R} - \mu \right)}-1} \right](x,y) &= \sum_{\alpha=1}^{\infty} e^{-\beta \left( h^{\mathrm{D}}_{\geq R} - \mu \right) \alpha}(x,y) \label{eq:densitybound4} \\
&= \sum_{\alpha=1}^{\infty} \int \mathds{1}_{\Omega}\left( q \right) \exp\left( - \int_0^{\beta \alpha} V(q(s)) \text{d}s \right) \text{d}W_{x,y}(q) \nonumber
\end{align}
for all $x,y \in B(R)^{\mathrm{c}}$. Here $\text{d}W_{x,y}$ denotes the Wiener measure on paths with startpoint $x$ and endpoint $y$ and $\Omega$ is the set of those paths that do not leave $B(R)^{\mathrm{c}}$. By $\mathds{1}_{\Omega}(q)$ we denote its characteristic function. Since $V(x) \geq 0$, the exponential function in the second line on the right-hand side of Eq.~\eqref{eq:densitybound4} is bounded from above by $1$. We can also drop the characteristic function of $\Omega$ in the Wiener integral to obtain an upper bound. Together with Eq.~\eqref{eq:densitybound4} this implies
\begin{equation}
\left[ \frac{1}{e^{\beta \left( h^{\mathrm{D}}_{\geq R} - \mu \right)}-1} \right](x,y) \leq \sum_{\alpha=1}^{\infty} e^{\beta \alpha \Delta}(x,y)
\label{eq:densitybound5}
\end{equation}
for all $x,y \in B(R)^{\mathrm{c}}$. Here $e^{\beta \alpha \Delta}(x,y)$ denotes the heat kernel of the Laplacian on $\mathbb{R}^3$ 
which is bounded from above by $(4 \pi \beta \alpha)^{-3/2}$, see e.g. \cite{LiebLoss}. Hence,
\begin{equation}
\left[ \frac{1}{e^{\beta \left( h^{\mathrm{D}}_{\geq R} - \mu \right)}-1} \right](x,y) \leq \left( \frac{1}{4 \pi \beta} \right)^{3/2} \sum_{\alpha=1}^{\infty} \frac{1}{\alpha^{3/2}}
\label{eq:densitybound6}
\end{equation}
which proves the claim.
\end{proof}

\subsection{The thermal cloud}
\label{subsec:thermalcloud_upper}
With the above preparations at hand we start our discussion of the free energy of the trial state in Eq.~\eqref{eq:upperbound8}  by considering the part representing the energy of the thermal cloud. In terms of the spectral decomposition of the density matrix $\tilde{\Gamma}_{\geq R}^{\mathrm{D},\Lambda}$ this energy   can be written as
\begin{equation}
\tr\left[ \left( H^{\mathrm{D}}_{\geq R} + V_{33} \right) \tilde{\Gamma}_{\geq R}^{\mathrm{D},\Lambda} \right] = \sum_{\alpha=1}^{\infty} \lambda_{\alpha} \frac{\langle F \Psi_{\alpha}, \left( H_{\geq R}^{\mathrm{D}} + V_{33} \right) F \Psi_{\alpha} \rangle}{\langle F \Psi_{\alpha}, F \Psi_{\alpha} \rangle}.
\label{eq:upperbound9}
\end{equation}
Bearing in mind that all eigenfunctions $\Psi_{\alpha}$ of $H^{\mathrm{D}}_{\geq R} $ can be chosen to be real-valued, we integrate by parts once to rewrite the kinetic energy for the $i$-th coordinate as
\begin{equation}
\int_{\mathbb{R}^{3 N_{\mathrm{th}}}} \overline{ F \Psi_{\alpha} } \nabla_i^2 F \Psi_{\alpha} \text{d}X = \int_{\mathbb{R}^{3 N_{\mathrm{th}}}} \left[ F^2 \left( \Psi_{\alpha} \nabla_i^2 \Psi_{\alpha} \right) - \Psi_{\alpha}^2 \left( \nabla_i F \right)^2 \right] \text{d}X\,,
\label{eq:upperbound10}
\end{equation}
where $\text{d}X$ is short for $\text{d}(x_1,\ldots,x_{N_\mathrm{th}})$. 
For the energy of a single $\Psi_{\alpha}$ this implies
\begin{align}
&\left\langle F \Psi_{\alpha}, \left[ \sum_{i=1}^{N_{\mathrm{th}}} \left( -\Delta_i + \frac{\omega^2 x_i^2}{4} - \frac{3 \omega}{2} \right) + \sum_{1 \leq i < j \leq N_{\mathrm{th}}} v_N(x_i-x_j) \right] F \Psi_{\alpha} \right\rangle =
\label{eq:upperbound11} \\
&\int_{\mathbb{R}^{3 N_{\mathrm{th}}}} \left\lbrace F^2 \Psi_{\alpha} \underbrace{\left[ \sum_{i=1}^{N_{\mathrm{th}}} \left( -\Delta_i + \frac{\omega^2 x_i^2}{4} - \frac{3 \omega}{2} \right) \right] \Psi_{\alpha}}_{=E_{\alpha} \psi_{\alpha}} +  \left[ \sum_{i=1}^{N_{\mathrm{th}}} (\nabla_i F)^2 + \sum_{1 \leq i < j \leq N_{\mathrm{th}}} v_N(x_i - x_j) F^2 \right] \Psi_{\alpha}^2 \right\rbrace \nonumber
\end{align}
and the whole energy can be written as
\begin{align}
&\tr\left[ \left( H^{\mathrm{D}}_{\geq R} + V_{33} \right) \tilde{\Gamma}_{\geq R}^{\mathrm{D},\Lambda} \right] = \tr\left( H^{\mathrm{D}}_{\geq R} \Gamma_{\geq R}^{\mathrm{D},\Lambda} \right) \label{eq:upperbound12} \\
&\hspace{4.5cm} +  \sum_{\alpha=1}^{\infty} \lambda_{\alpha} \frac{\int_{\mathbb{R}^{3 N_{\mathrm{th}}}} \Psi_{\alpha}^2 \left[ \sum_{i=1}^{N_{\mathrm{th}}} (\nabla_i F)^2 + \sum_{1 \leq i < j \leq N_{\mathrm{th}}} v_N(x_i - x_j) F^2 \right]}{\left\Vert F \Psi_{\alpha} \right\Vert^2}. \nonumber
\end{align}
The following Lemma provides a lower bound for the norm of $F \Psi_{\alpha}$ and thereby an upper bound on  the second term on the right-hand side of Eq.~\eqref{eq:upperbound12} as long as $(\omega+\Lambda)^{3/2} a_N b^3 N_{\mathrm{th}}^2$ is small enough. We need the cut-off $\Lambda$ in the definition of the state $\tilde{\Gamma}_{\geq R }^{\mathrm{D},\Lambda}$ in the proof of Lemma~\ref{lem:denominators}. It could be avoided if we  knew that the $L^4(\mathbb{R}^3)$-norm of the eigenfunctions of the operator $h^{\mathrm{D}}_{\geq R}$ are bounded independently of the energy, which we expect to be true. (Compare with the result in \cite{kocht} for $h$ on the whole space $\mathbb{R}^3$.) The bound would most likely grow with $R$, however, which would need to be quantified. We do not have such a bound at our disposal, and therefore need the cut-off. It will be chosen such that $\omega \ll T \ll \Lambda$ holds.
\begin{lemma}
\label{lem:denominators}
There exists a constant $C > 0$ independent of $\alpha$ such that 
\begin{equation}
\left\Vert F \Psi_{\alpha} \right\Vert^2 \geq 1 - C (\omega+\Lambda)^{3/2} a_N b^2 N_{\mathrm{th}}^2.
\label{eq:upperbound13}
\end{equation}
\end{lemma}
\begin{proof}
Spelled out in more detail, the norm of $F \Psi_{\alpha}$ reads
\begin{equation}
\left\Vert F \Psi_{\alpha} \right\Vert^2 = \int_{\mathbb{R}^{3 N_{\mathrm{th}}}} | \Psi_{\alpha} |^2 \prod_{1 \leq i < j \leq N_{\mathrm{th}}} f_b(x_i - x_j)^2 \text{d}X.
\label{eq:upperbound14}
\end{equation} 
We define $\eta_b(x) = 1 - f_b(x)^2$ and estimate
\begin{align}
\left\Vert F \Psi_{\alpha} \right\Vert^2 &\geq \int_{\mathbb{R}^{3 N_{\mathrm{th}}}} | \Psi_{\alpha} |^2 \left( 1 -  \sum_{1 \leq i < j \leq N_{\mathrm{th}}} \eta_b(x_i - x_j) \right) \text{d}X \label{eq:upperbound15} \\
&= 1 - \int_{\mathbb{R}^{6}} \eta_b(x-y) \varrho^{(2)}_{\Psi_{\alpha}}(x,y) \text{d}(x,y), \nonumber
\end{align}
where $\varrho^{(2)}_{\Psi_{\alpha}}(x,y)$ denotes the two-particle density of $\Psi_{\alpha}$. We use the fact that the $\Psi_{\alpha}$ are symmetrized products of one-particle orbitals to conclude that $\varrho^{(2)}_{\Psi_{\alpha}}(x,y) \leq  \varrho_{\Psi_{\alpha}}(x) \varrho_{\Psi_{\alpha}}(y)$ holds, where $\varrho_{\Psi_{\alpha}}$ is the one-particle density of $\Psi_{\alpha}$. This allows us to bound the integral on the right-hand side of  Eq.~\eqref{eq:upperbound15} in the following way:
\begin{equation}
\int_{\mathbb{R}^{6}} \eta_b(x-y) \varrho^{(2)}_{\Psi_{\alpha}}(x,y) \text{d}(x,y) \leq \underbrace{\int_{\mathbb{R}^3} \eta_b(x) \text{d}x}_{_{\leq \frac{4 \pi}{3} a_N b^2}} \int_{\mathbb{R}^3} \varrho_{\Psi_{\alpha}}(y)^2 \text{d}y. 
\label{eq:upperbound16}
\end{equation}
To obtain the bound for the integral of $\eta_b$, we used its explicit form and the lower bound $f_0(|x|) \geq \left[ 1-a/|x| \right]_+$, see \cite[Appendix~C]{Themathematicsofthebosegas}.

Let us have a closer look at the integral over the squared density on the right-hand side of Eq.~\eqref{eq:upperbound16}. Denote by $\lbrace \varphi_j^{\mathrm{D}} \rbrace_{j=0}^{\infty}$ a complete set of eigenfunctions of $h^{\mathrm{D}}_{\geq R}$ and estimate
\begin{equation}
\int_{\mathbb{R}^3} \varrho_{\Psi_{\alpha}}(x)^2 \text{d}x = \int_{\mathbb{R}^3} \left[ \sum_{j} \langle \varphi_j | \gamma_{\Psi_\alpha} | \varphi_j\rangle \varphi_j^{\mathrm{D}}(x)^2  \right]^2 \text{d}x \leq N_{\mathrm{th}}^2 \sup_{ \substack{j \geq 0 : \\ e_j(h^{\mathrm{D}}_{\geq R}) \leq \Lambda}} \left\Vert \varphi_j^{\mathrm{D}} \right\Vert_{L^4(\mathbb{R}^3)}^4.
\label{eq:upperbound17}
\end{equation}
In the above equation $e_j(h^{\mathrm{D}}_{\geq R})$ is the $j$-th eigenvalue of $h^{\mathrm{D}}_{\geq R}$. By the H\"older and Sobolev inequalities and the normalization of the functions $\varphi_j^{\mathrm{D}}$, we have $\Vert \varphi_j^{\mathrm{D}} \Vert_{L^4(\mathbb{R}^3)} \lesssim \Vert \nabla \varphi_j^{\mathrm{D}} \Vert^{3/4}$. On the other hand, 
$\Vert \nabla \varphi_j^{\mathrm{D}} \Vert \leq (3\omega/2 + e_j(h^{\mathrm{D}}_{\geq R}))^{1/2} \leq (3\omega/2+ \Lambda)^{1/2}$. This implies
\begin{equation}
\int_{\mathbb{R}^3} \varrho_{\Psi_{\alpha}}(x)^2 \text{d}x \lesssim N_{\mathrm{th}}^2 (\omega+\Lambda)^{3/2}.
\label{eq:upperbound17b}
\end{equation}
Together with Eqs.~\eqref{eq:upperbound15}--\eqref{eq:upperbound16} this proves the claim.
\end{proof}

Next we analyze the numerator of the second term on the right-hand side of Eq.~\eqref{eq:upperbound12}. We compute
\begin{equation}
\nabla_i F(x_1,\ldots,x_{N_{\mathrm{th}}}) = \sum_{\substack{l = 1 \\ l \neq i}}^{N_{\mathrm{th}}} \frac{F(x_1,\ldots,x_{N_{\mathrm{th}}})}{f_b(x_l-x_i)} \nabla f_b(x_l-x_i).
\label{eq:upperbound18}
\end{equation} 
The square of this expression is given by
\begin{align}
\left( \nabla_i F \right)^2 &= \sum_{\substack{l=1 \\ l \neq i}} \frac{F^2}{f_b(x_l - x_i)^2} \left[ \nabla f_b(x_l - x_i) \right]^2 \label{eq:upperbound19} \\
&\hspace{2cm} + \sum_{\substack{k,l=1 \\ l,k \neq i \\ k \neq l}} \frac{F^2}{f_b(x_l-x_i) f_b(x_k - x_i)} \nabla f_b(x_l-x_i) \nabla f_b(x_k-x_i). \nonumber
\end{align}
These terms need to be inserted into the numerator of the second term on the right-hand side of Eq.~\eqref{eq:upperbound12} and we start with the first term on the right-hand side of the above equation. Introducing the function $\xi(x) = \left[ \nabla f_b(x) \right]^2 + \tfrac{1}{2} v_N(x) f_b(x)^2$ and noting that $0 \leq f_b \leq 1$ and $\sum_{\alpha=1}^{\infty} \lambda_{\alpha} \varrho^{(2)}_{\Psi_{\alpha}}(x,y) = \varrho^{(2)}_{\Gamma_{\geq R}^{\mathrm{D},\Lambda}}(x,y)$, we obtain
\begin{align}
&\sum_{\alpha=1}^{\infty} \lambda_{\alpha} \sum_{1 \leq i < j \leq N_{\mathrm{th}}} \int_{\mathbb{R}^{3 N_{\mathrm{th}}}} \left\lbrace \frac{2F^2}{f_b(x_i-x_j)^2} \left[ \nabla f_b(x_i - x_j) \right]^2 + v_N(x_i - x_j) F^2 \right\rbrace \Psi_{\alpha}^2 \text{d}X \label{eq:upperbound20} \\
&\hspace{8cm} \leq 2 \int_{\mathbb{R}^6} \xi(x-y) \varrho^{(2)}_{\Gamma_{\geq R}^{\mathrm{D},\Lambda}}(x,y) \text{d}(x,y). \nonumber
\end{align}
Similarly, the off-diagonal terms in Eq.~\eqref{eq:upperbound19} are bounded from above by
\begin{equation}
6 \int_{\mathbb{R}^9} \varrho^{(3)}_{\Gamma_{\geq R}^{\mathrm{D},\Lambda}}(x,y,z) \vert \nabla f_b(x-y) \nabla f_b(z-y) \vert \text{d}(x,y,z). \label{eq:upperbound21}
\end{equation}
Combining this with Eqs.~\eqref{eq:upperbound12}--\eqref{eq:upperbound13}, we obtain for an appropriately chosen constant $C>0$ and $ C (\omega+\Lambda)^{3/2} a_N b^2 N_{\mathrm{th}}^2 < 1$
\begin{equation}
\tr\left[ \left( H^{\mathrm{D}}_{\geq R} + V_{33} \right) \tilde{\Gamma}_{\geq R}^{\mathrm{D},\Lambda} \right] \leq \tr\left( H^{\mathrm{D},\Lambda}_{\geq R} \Gamma_{\geq R}^{\mathrm{D}} \right) + \frac{A}{1 - C (\omega+\Lambda)^{3/2} a_N b^2 N_{\mathrm{th}}^2} \label{eq:upperbound22} 
\end{equation}
as an upper bound for the energy of the thermal cloud, where
\begin{align}
A & = 2 \int_{\mathbb{R}^6} \xi(x-y) \varrho^{(2)}_{\Gamma_{\geq R}^{\mathrm{D},\Lambda}}(x,y)  \text{d}(x,y) \label{eq:upperbound22b} \\
&\quad  + 6 \int_{\mathbb{R}^9} \left\vert \nabla f_b(x-y) \nabla f_b(z-y) \right\vert \varrho^{(3)}_{\Gamma_{\geq R}^{\mathrm{D},\Lambda}}(x,y,z) \text{d}(x,y,z) \,.\nonumber
\end{align}
With the help of Proposition~\ref{cor:density} and Lemma~\ref{lem:densitybound}, one readily estimates the first term on the right-hand side of Eq.~\eqref{eq:upperbound22b} by
\begin{equation}
\int_{\mathbb{R}^6} \xi(x-y) \varrho^{(2)}_{\Gamma_{\geq R}^{\mathrm{D},\Lambda}}(x,y)  \text{d}(x,y) \label{eq:upperbound23}  \lesssim \frac{N_{\mathrm{th}}}{ \beta^{3/2}} \underbrace{\int_{\mathbb{R}^3}\xi(x) \text{d}x}_{\frac{4 \pi a_N}{1-\frac{a_N}{b}} + \int_{|x|>b} v_N(x) \text{d}x }  \lesssim \frac{a_N N_{\mathrm{th}}}{\beta^{3/2}}. 
\end{equation}
The terms below the curly brackets are obtained from the explicit form of $f_b$, see \cite[Appendix~C]{Themathematicsofthebosegas}. To obtain the final bound we used 
$b \geq c a_N$ for some $c>1$. Note that we assumed $e^{-\beta \Lambda/4} \lesssim \beta \omega$ and $\beta \omega \ln N \lesssim 1$ in order to be able to apply Lemma~\ref{lem:densitybound}.  

The second term on the right-hand side of Eq.~\eqref{eq:upperbound22b} can be treated with a rough bound that we derive now. Let $a_\alpha$ and $a^*_\alpha$ denote the usual creation and annihilation operators on Fock space corresponding to an eigenfunction $\varphi_\alpha^{\mathrm D}$ of $h^{\mathrm D}_{\geq R}$. An application of the Cauchy-Schwarz inequality tells us that 
\begin{align} \nonumber
\varrho^{(3)}_{\Gamma_{\geq R}^{\mathrm{D},\Lambda}}(x,y,z) &=  \sum_{\alpha_1,\alpha_2,\alpha_3}\left| \frac{1}{6}  \sum_{\sigma \in S_3} \varphi^{\mathrm{D}}_{\sigma(\alpha_1)}(x) \varphi^{\mathrm{D}}_{\sigma(\alpha_2)}(y) \varphi^{\mathrm{D}}_{\sigma(\alpha_3)}(z) \right|^2   \left\langle a_{\alpha_1}^* a_{\alpha_2}^* a_{\alpha_3}^* a_{\alpha_3} a_{\alpha_2} a_{\alpha_1}  \right\rangle_{\Gamma_{\geq R}^{\mathrm{D},\Lambda}}  \\  \label{eq:upperbound23a} 
&\leq  \sum_{\alpha_1,\alpha_2,\alpha_3} \varphi^{\mathrm{D}}_{\alpha_1}(x)^2 \varphi^{\mathrm{D}}_{\alpha_2}(y)^2 \varphi^{\mathrm{D}}_{\alpha_3}(z)^2 \left\langle a_{\alpha_1}^* a_{\alpha_2}^* a_{\alpha_3}^* a_{\alpha_3} a_{\alpha_2} a_{\alpha_1} \right\rangle_{\Gamma_{\geq R}^{\mathrm{D},\Lambda}}, 
\end{align}
where $S_3$ denotes the group of permutations of three elements. Using $abc \leq \tfrac{1}{3} (a^3+b^3+c^3)$, we obtain
\begin{equation}
\varrho^{(3)}_{\Gamma_{\geq R}^{\mathrm{D},\Lambda}}(x,y,z) \leq \frac 13\sum_{\alpha_1,\alpha_2,\alpha_3} \left( \varphi^{\mathrm{D}}_{\alpha_1}(x)^6 + \varphi^{\mathrm{D}}_{\alpha_2}(y)^6 +\varphi^{\mathrm{D}}_{\alpha_3}(z)^6 \right) \left\langle a_{\alpha_1}^* a_{\alpha_2}^* a_{\alpha_3}^* a_{\alpha_3} a_{\alpha_2} a_{\alpha_1} \right\rangle_{\Gamma_{\geq R}^{\mathrm{D},\Lambda}}.
\label{eq:upperbound23b}
\end{equation}
We insert this bound into the second term on the right-hand side Eq.~\eqref{eq:upperbound22b} and obtain 
\begin{align}
&\int_{\mathbb{R}^9} \left\vert \nabla f_b(x-y) \nabla f_b(z-y) \right\vert \varrho^{(3)}_{\Gamma_{\geq R}^{\mathrm{D},\Lambda}}(x,y,z) \text{d}(x,y,z) \label{eq:upperbound23c} \\
&\hspace{5cm} \leq N_{\mathrm{th}}^3 \left( \int_{\mathbb{R}^3} \left| \nabla f_b(x) \right| \text{d}x \right)^2 \sup_{\substack{\alpha \geq 0 : \\ e_{\alpha}(h^{\mathrm{D}}_{\geq R}) \leq \Lambda}} \left\Vert \varphi_{\alpha}^{\mathrm{D}} \right\Vert_{L^6(\mathbb{R}^3)}^6. \nonumber
\end{align}
From Sobolev's inequality 
we infer that $\left\Vert \varphi_{\alpha}^{\mathrm{D}} \right\Vert_{L^6(\mathbb{R}^3)}^2 \leq 3\omega/2 + e_{\alpha}(h^{\mathrm{D}}_{\geq R}) \leq 3\omega/2 + \Lambda$. We also have $\int_{\mathbb{R}^3} \left| \nabla f_b(x) \right| \text{d}x \leq a_N b$. Combining this with Eqs.~\eqref{eq:upperbound22}--\eqref{eq:upperbound23} we finally obtain
\begin{equation}
\tr\left[ \left( H^{\mathrm{D}}_{\geq R} + V_{33} \right) \tilde{\Gamma}_{\geq R}^{\mathrm{D},\Lambda} \right] \leq \tr\left( H^{\mathrm{D}}_{\geq R} \Gamma_{\geq R}^{\mathrm{D},\Lambda} \right) + C \left( \frac{a_N N_{\mathrm{th}}}{\beta^{3/2}} +  \Lambda^3 a_N^2 b^2 N_{\mathrm{th}}^3 \right)
\label{eq:upperbound23d}
\end{equation}
for some appropriately chosen $C>0$ as an upper bound on the  energy of the thermal cloud.

Next, we investigate the entropy of the thermal cloud. To that end, we use \cite[Lemma~2]{RobertFermigas} which we spell out here for the sake of completeness. 

\begin{lemma}
\label{lem:upperbound4}
Let $\Gamma$ be a density matrix on some Hilbert space, with eigenvalues $\lambda_{\alpha} \geq 0$, $\alpha \geq 0$. Additionally let $\lbrace P_{\alpha} \rbrace_{\alpha=0}^{\infty}$ be a family of  one-dimensional orthogonal projections (for which $P_{\alpha} P_{\alpha'} = \delta_{\alpha,\alpha'} P_{\alpha}$ need not necessarily be true) and define $\hat{\Gamma} = \sum_{\alpha} \lambda_{\alpha} P_{\alpha}$. Then
\begin{equation}
S( \hat{\Gamma} ) \geq S(\Gamma) - \ln\left( \left\Vert \textstyle{\sum_{\alpha} } P_{\alpha} \right\Vert \right).
\label{eq:upperbound24}
\end{equation} 
\end{lemma}

Since $ 0 \leq F \leq 1$ we have
\begin{equation}
\sum_{\alpha=0}^{\infty} \frac{\vert F \Psi_{\alpha} \rangle\langle F \Psi_{\alpha} \vert}{\Vert F \Psi_{\alpha} \Vert^2} \leq \left( \sup_{\alpha \geq 0} \Vert F \Psi_{\alpha} \Vert^{-2} \right). \label{eq:upperbound25}
\end{equation} 
Eq.~\eqref{eq:upperbound25} together with Lemmas~\ref{lem:denominators} and~\ref{lem:upperbound4} show that
\begin{equation}
- T S \left(\tilde{\Gamma}_{\geq R}^{\mathrm{D},\Lambda} \right) \leq -T S \left(\Gamma_{\geq R}^{\mathrm{D},\Lambda} \right) + O \left( T \Lambda^{3/2} a_N b^2 N_{\mathrm{th}}^2\right) 
\label{eq:upperbound25b}
\end{equation}
holds (as long as $(\omega+\Lambda)^{3/2} a_N b^2 N_{\mathrm{th}}^2$ is small enough). 

What remains to be done is to get rid of the Dirichlet boundary condition and the cut-off in the canonical free energy $F_{\geq R}^{\mathrm{D},\Lambda}(\beta,N,\omega) = \tr\left[ H^{\mathrm{D}}_{\geq R } \Gamma^{\mathrm{D},\Lambda}_{\geq R} \right] - T S(\Gamma^{\mathrm{D},\Lambda}_{\geq R})$ of the ideal gas. For that purpose, we will first relate it to its grand canonical analogue where explicit formulas are available. Using Corollary~\ref{cor:freeenergy}, we can bound the canonical free energy from above by
\begin{equation}
F_{\geq R}^{\mathrm{D},\Lambda}(\beta,N_{\mathrm{th}},\omega) \leq \tfrac{1}{\beta} \tr\left[ \mathds{1}\left( h^{\mathrm{D}}_{\geq R} \leq \Lambda \right) \ln\left( 1- e^{-\beta \left( h^{\mathrm{D}}_{\geq R} - \mu \right)} \right) \right] + \mu N_{\mathrm{th}} + O(T \ln N_{\mathrm{th}} )
\label{eq:upperbound26}
\end{equation}
where the chemical potential $\mu$ is chosen as in \eqref{eq:lemmamu1} such that the particle number of the grand canonical Gibbs state associated  to  $\mathds{1}( h^{\mathrm{D}}_{\geq R}\leq \Lambda) h^{\mathrm{D}}_{\geq R}$ equals $N_{\mathrm{th}}$. In order to replace the chemical potential $\mu$ by $\mu_0$ we use the convexity inequalities 
\begin{align}\nonumber
&\tfrac{1}{\beta} \tr \left[ \mathds{1} \left( h^{\mathrm{D}}_{\geq R} \leq \Lambda \right) \left\lbrace \ln\left( 1 - e^{-\beta \left( h_{\geq R}^{\mathrm{D}} - \mu \right)} \right) - \ln\left( 1 - e^{-\beta \left( h_{\geq R}^{\mathrm{D}} - \mu_0 \right) } \right) \right\rbrace \right] \\
& \leq (\mu_0 - \mu) \tr \left[ \mathds{1} \left( h^{\mathrm{D}}_{\geq R} \leq \Lambda \right) \frac{1}{e^{\beta ( h_{\geq R }^{\mathrm{D}} - \mu_0 )}-1} \right] \nonumber
\\
& \leq (\mu_0 - \mu) N_\mathrm{th} + \beta (\mu_0 - \mu)^2  \tr \left[ \mathds{1} \left( h^{\mathrm{D}}_{\geq R} \leq \Lambda \right) \frac{e^{\beta ( h_{\geq R }^{\mathrm{D}} - \mu )}}{\left( e^{\beta ( h_{\geq R }^{\mathrm{D}} - \mu )}-1\right)^2} \right].  \label{eq:upperbound28}
\end{align} 
To bound the last term from above, we can drop the projection $\mathds{1}( h^{\mathrm{D}}_{\geq R}\leq \Lambda)$. Moreover, since $\mu\lesssim \lambda \omega^2 R^2$ (as argued in the proof of Lemma~\ref{lem:densitybound}) and $h_{\geq R}^\mathrm{D} \geq \omega^2 R^2/4 - 3 \omega/2$, we have $h_{\geq R}^\mathrm{D} - \mu \geq c ( h_{\geq R}^\mathrm{D} + \omega)$ for a suitable constant $c>0$ for small  $\lambda$ and large $\omega R^2$. Since the eigenvalues of $h_{\geq R}^\mathrm{D}$ are larger than the ones of $h$,  this implies that
\begin{equation}\label{eqsq}
 \tr \left[  \frac{e^{\beta ( h_{\geq R }^{\mathrm{D}} - \mu )}}{\left( e^{\beta ( h_{\geq R }^{\mathrm{D}} - \mu )}-1\right)^2} \right] \leq  \tr \left[  \frac{e^{c \beta ( h +\omega )}}{\left( e^{c \beta ( h+\omega  )}-1\right)^2} \right] \lesssim (\beta\omega)^{-3}\,.
\end{equation}

Next, we estimate the influence of the cut-off. To do so, we use the same argumentation as the one in the proof of Lemma~\ref{lem:chemicalpotential} that leads to Eq.~\eqref{eq:lemmamu1f}. We find that 
\begin{equation}
\left\vert \tr\left[ \mathds{1}\left( h^{\mathrm{D}}_{\geq R} > \Lambda \right) \ln\left( 1- e^{-\beta \left( h^{\mathrm{D}}_{\geq R} - \mu_0 \right)} \right) \right] \right\vert \lesssim \frac{e^{-\beta \Lambda/4}}{(\beta \omega)^3} \label{eq:upperbound26c}
\end{equation}
holds. To get rid of the Dirichlet  boundary conditions at $\partial B(R)$, we use Lemma~\ref{lem:traces1} and Lemma~\ref{lem:traces2} with the choice $f(x) = -\ln(1-e^{-x})$, which gives
\begin{align} \nonumber
& \tfrac{1}{\beta} \tr\left[ \ln\left( 1- e^{-\beta \left( h^{\mathrm{D}}_{\geq R} - \mu_0 \right)} \right) \right]  \\ & \leq  \tfrac{1}{\beta} \tr\left[ \mathds{1}\left( h \geq \omega \right) \ln\left( 1- e^{-\beta \left( h - \mu_0 \right)} \right) \right] + O\left( \frac{1}{\beta^3 \omega^3 R^2} \right)  \label{eq:upperbound29} 
 + O\left( \frac{R^3}{\beta^{5/2}} \right).
\end{align}
Together with Eqs.~\eqref{eq:upperbound26}--\eqref{eq:upperbound26c}, this implies the upper bound
\begin{align}\nonumber 
F_{\geq R}^{\mathrm{D},\Lambda}(\beta,N_{\mathrm{th}},\omega) &\leq \tfrac{1}{\beta} \tr\left[ \mathds{1}\left( h \geq \omega \right) \ln\left( 1- e^{-\beta \left( h - \mu_0 \right)} \right) \right] + \mu_0 N_{\mathrm{th}}  \\ \nonumber
& \quad  + O\left( \frac{e^{-\beta \Lambda/4}}{\beta (\beta \omega)^{3}} \right) + O\left( \frac{1}{\beta^3 \omega^3 R^2} \right) + O\left( \frac{R^3}{\beta^{5/2}} \right) 
 \\ 
& \quad + O (T \ln N_\mathrm{th}) + O\left( \frac{ \beta(\mu-\mu_0)^2 }{(\beta \omega)^3 } \right) 
\,. \label{eq:upperbound30}
\end{align}

In the final step we  relate the right-hand side of Eq.~\eqref{eq:upperbound30} to the canonical free energy $F_0(\beta,N,\omega)$ of the  ideal gas. First of all, we note that we can drop the spectral projection $\mathds{1}\left( h \geq \omega \right)$ in the first term on the right-hand side of Eq.~\eqref{eq:upperbound30} at the cost of an error of the size $\beta^{-1} \ln(1-e^{\beta \mu_0}) \lesssim \beta^{-1} \ln N_0^{\mathrm{gc}} \leq \beta^{-1} \ln N$. Secondly, we add the missing term $\mu_0 N_0$ which is  $O(T)$. 
Finally, 
Corollary~\ref{cor:freeenergy} tells us that $\vert F_0(\beta,N,\omega) - F_0^{\mathrm{gc}}(\beta,N,\omega) \vert \leq T O( \ln N )$.  Together with Eq.~\eqref{eq:upperbound30}, we conclude that
\begin{align}\nonumber
F_{\geq R}^{\mathrm{D}}(\beta,N_{\mathrm{th}},\omega) & \leq  F_0(\beta,N,\omega) + O\left( \frac{1}{\beta^3 \omega^3 R^2} \right) + O\left( \frac{R^3}{\beta^{5/2}} \right) + O\left( \frac{e^{-\beta \Lambda/4}}{\beta (\beta\omega)^3} \right) \\ & \quad + O\left( \beta^{-1} \ln N \right) + O\left( \frac{ \beta(\mu-\mu_0)^2 }{(\beta \omega)^3 } \right)   \,. \label{eq:upperbound31}
\end{align}

We combine Eqs.~\eqref{eq:upperbound23d},~\eqref{eq:upperbound25b} and~\eqref{eq:upperbound31} to find the final upper bound for the contribution of the thermal cloud to the free energy \eqref{eq:upperbound8} of the trial state. It reads 
\begin{align}\nonumber 
&\tr\left[ \left( H^{\mathrm{D}}_{\geq R} + V_{33} \right) \tilde{\Gamma}_{\geq R}^{\mathrm{D},\Lambda} \right] - T S\left( \tilde{\Gamma}_{\geq R}^{\mathrm{D},\Lambda} \right) \\ \nonumber &  \leq F_0(\beta,N,\omega) + O\left( \frac{1}{\beta^3 \omega^3 R^2} \right) + O\left( \frac{R^3}{\beta^{5/2}} \right)  + O\left( \frac{e^{-\beta \Lambda/4}}{\beta (\beta \omega)^3} \right) + O\left( \beta^{-1} \ln N  \right)\\ 
&  \quad + O\left( \frac{ \beta(\mu-\mu_0)^2 }{(\beta \omega)^3 } \right)   + O\left( \frac{a_N N_{\mathrm{th}}}{\beta^{3/2}}\right) + O\left( \Lambda^3 a_N^2 b^2 N_{\mathrm{th}}^3 \right).  \label{eq:upperbound28c}
\end{align}
Recall that $\mu-\mu_0$ was estimated in Lemma~\ref{lem:chemicalpotential}. 
To obtain the result, we assumed $\omega^{-1/2} \ll R \leq \lambda \omega^{-1} \beta^{-1/2}$ with $\lambda > 0$ small enough, $e^{-\beta \Lambda/4} \lesssim \beta \omega$ and $\beta \omega \ln N \lesssim 1$ as well as that $\Lambda^{3/2} a_N b^2 N_{\mathrm{th}}^2$ is small enough.

\subsection{The condensate energy and the interaction between the condensate and the thermal cloud}
We recall that $E_{\leq  R}^{\mathrm{D}}$ denotes the ground state energy of the Hamiltonian $H_{\leq R}^{\mathrm{D}}$ in  \eqref{eq:upperbound2}. In the following  we write it as $E_{\leq  R}^{\mathrm{D}}(N_0)$ to explicitly highlight its dependence on the number of particles in the condensate. The strategy here is to use existing results for the ground state energy to relate $E_{\leq  R}^{\mathrm{D}}(N_0)$ to the GP energy $E^{\mathrm{GP}}(N_0,a_N,\omega)$ defined in Eq.~\eqref{eq:mainresult2}.

If we go through the proof of the upper bound in \cite{RobertGPderivation}, we obtain 
\begin{equation}
E^{\mathrm{D}}_{\leq R}(N_0) \leq E^{\mathrm{GP},\mathrm{D}}_{\leq R}(N_0,a_N,\omega) \left( 1 + O\left( N_0^{-2/3} \right) \right)
\label{eq:upperboundcondensate1a}
\end{equation}
Here $E^{\mathrm{GP},\mathrm{D}}_{\leq R}(N_0,a_N,\omega)$ denotes the GP energy when we minimize only over functions in $H^1_0(B(R))$, that is, over functions that vanish outside $B(R)$. It is therefore sufficient to find an upper bound for $E^{\mathrm{GP},\mathrm{D}}_{\leq R}(N_0,a_N,\omega)$ in terms of the GP energy $E^{\mathrm{GP}}(N_0,a_N,\omega)$ without additional boundary conditions. Let $j_1,j_2 \in \mathcal{C}^{\infty}(\mathbb{R}^3)$ be a partition of unity in the sense that $j_1(x)^2+j_2(x)^2 = 1$ for all $x \in \mathbb{R}^3$. We assume that $j_1$ equals one for $|x| \leq R/2$ and zero for $|x| \geq R$ and that $|\nabla j_1(x)|^2 + |\nabla j_2(x)|^2 \leq 12 R^{-2}$. The IMS localization formula (see e.g. \cite{SimonSchroedingerOps}) tells us that $h = j_1 h j_1 + j_2 h j_2 - |\nabla j_1|^2 - |\nabla j_2|^2$ where $h$ is the harmonic oscillator Hamiltonian \eqref{eq:idealbosegas1a}. For the GP energy this implies
\begin{align}
E^{\mathrm{GP}}(N_0,a_N,\omega) & \geq \inf_{ \Vert \phi \Vert^2=N_0} \left\lbrace \left\langle j_1 \phi, h j_1 \phi \right\rangle + 4 \pi a_N \int_{\mathbb{R}^3} \left\vert j_1(x) \phi(x) \right\vert^4 \text{d}x + \left( \frac{\omega^2 R^2}{4} - \frac{3 \omega}{2} \right) \left\Vert j_2 \phi \right\Vert^2 \right\rbrace \nonumber \\
&\quad\quad  - 12 R^{-2} N_0. \label{eq:upperboundcondensate1}
\end{align}
By minimizing over $j_1 \phi$ and $j_2 \phi$ separately, keeping the constraint $\Vert j_1 \phi \Vert^2 + \Vert j_2 \phi \Vert^2 = N_0$, we obtain a lower bound on the right-hand side of Eq.~\eqref{eq:upperboundcondensate1}. Since $a_N \sim N^{-1}$ and $\omega^2 R^2 \gg \omega$, one easily sees that the minimum is attained if we put all $L^2$-mass into the function $j_1 \phi$. For the energy, we therefore obtain
\begin{equation}
E^{\mathrm{GP}}(N_0,a_N,\omega) \geq E^{\mathrm{GP},\mathrm{D}}_{\leq R}(N_0,a_N,\omega) - 12 R^{-2} N_0.
\label{eq:upperboundcondensate2}
\end{equation}
This finally proves
\begin{equation}
E_{\leq  R}^{\mathrm{D}}(N_0) \leq E^{\mathrm{GP}}(N_0,a_N,\omega) + O\left(\omega N_0^{1/3} \right) + O\left(N_0 R^{-2} \right),
\label{eq:upperboundcondensate3}
\end{equation}
which is the desired bound for the first term on the right-hand side of \eqref{eq:upperbound8}.

The last term in \eqref{eq:upperbound8} to consider is the interaction between the condensate and the thermal cloud,  given by
\begin{equation}
\tr\left( V_{13} \Gamma_N \right) = \int_{B(R+\ell)^{\mathrm{c}} \times B(R)} v_N(x-y) \varrho_{\tilde{\Gamma}_{\geq R+\ell}^{\mathrm{D},\Lambda}}(x) \varrho_{\Psi^{\mathrm{D}}_{\leq R}}(y) \text{d}(x,y)\,,
\label{eq:upperboundcondensate5}
\end{equation}
where $\varrho_{\tilde{\Gamma}_{\geq R+\ell}^{\mathrm{D},\Lambda}}$ is the one-particle density of the state ${\tilde{\Gamma}_{\geq R+\ell}^{\mathrm{D},\Lambda}}$  defined in \eqref{eq:upperbound6}.  
Note that we have inserted the missing $\ell$ in $B(R+\ell)$ again (compare with the discussion in Subsection~\ref{subsec:ansatz}). 
When we use that $F \leq 1$ and apply Lemma~\ref{lem:denominators}, we find
\begin{equation}
\varrho_{\tilde{\Gamma}_{\geq R+\ell}^{\mathrm{D},\Lambda}}(x) \leq \frac{\varrho_{\Gamma_{\geq R + \ell}^{\mathrm{D},\Lambda}}(x)}{1-C (\omega+\Lambda)^{3/2} a_N b^2 N_{\mathrm{th}}^2}. 
\label{eq:upperboundcondensate6}
\end{equation}
for some $C>0$ and $C (\omega+\Lambda)^{3/2} a_N b^2 N_{\mathrm{th}}<1$. An application of Proposition~\ref{cor:density} and Lemma~\ref{lem:densitybound}  tells us that 
$\varrho_{\Gamma_{\geq R+\ell}^{\mathrm{D},\Lambda}}(x) \lesssim \beta^{-3/2}$, hence
\begin{equation}
\tr\left( V_{13} \Gamma_N \right) \lesssim \beta^{-3/2} N_0 \int_{|x|\geq \ell} v_N(x) \text{d} x = \omega^{-1/2} \beta^{-3/2} \frac{N_0}{N}  \int_{|x|\geq \ell N \omega^{1/2}} v(x) \text{d} x
\end{equation}
As already mentioned at the beginning of Section~\ref{subsec:ansatz}, we choose $\ell$ such that $\ell N\omega^{1/2}$ is larger than the radius of the hard-core part of the interaction potential. We thus conclude that 
$\tr\left( V_{13} \Gamma_N \right) \lesssim \omega^{-1/2} \beta^{-3/2}$. In combination with Eq.~\eqref{eq:upperboundcondensate3}, this yields the upper bound
\begin{equation}
E^{\mathrm{D}}_{\leq R}(N_0) + \tr\left( V_{13} \Gamma_N \right) \leq E^{\mathrm{GP}}(N_0,a_N,\omega) + O\left(\omega N_0^{1/3} \right) + O\left(N_0 R^{-2} \right) + O\left( \omega^{-1/2} \beta^{-3/2} \right)
\label{eq:upperboundcondensate7}
\end{equation}
for the condensate energy plus the interaction energy between the condensate and the thermal cloud. To obtain Eq.\ \eqref{eq:upperboundcondensate7} we assumed in addition to $\omega^{-1/2} \ll R \leq \lambda \omega^{-1} \beta^{-1/2}$ with $\lambda>0$ small enough that $e^{-\beta \Lambda/4} \lesssim \beta \omega$ and $\beta \omega \ln N \lesssim 1$ holds, as required for the application of Lemma~\ref{lem:densitybound}. We also assumed that $\Lambda^{3/2} a_N b^2 N_{\mathrm{th}}^2$ is small enough.
\subsection{The final estimate for the upper bound}

We combine the results of Eqs.~\eqref{eq:upperbound28c} and~\eqref{eq:upperboundcondensate7}. Together with an application of Lemma~\ref{lem:chemicalpotential} and Corollary~\ref{lem:particlenumbers}, which allows us to obtain a bound on $| \mu - \mu_0 |$, we obtain
\begin{align}
F(\beta,N,\omega) &\leq F_0(\beta,N,\omega) + E^{\mathrm{GP}}(N_0,a_N,\omega)  + O\left( \frac{1}{\beta^3 \omega^3 R^2} \right) + O\left( \frac{R^3}{\beta^{5/2}} \right) + O\left( \frac{e^{-\beta \Lambda/4}}{\beta^4 \omega^3} \right) \label{eq:upperboundfinal1} \\
& \quad  + O\left( \frac{\ln N}{\beta} \right) + O\left( \frac{a_N N_{\mathrm{th}}}{\beta^{3/2}}\right) + O \left( \Lambda^3 a_N^2 b^2 N_{\mathrm{th}}^3 \right) +O\left(\omega N_0^{1/3} \right) + O\left(\frac{N_0}{R^{2}} \right)  \nonumber \\
&\quad + O\left( \frac{1}{\omega^{1/2} \beta^{3/2} } \right) +O\left( \frac{1}{\beta^2 \omega^3 R^4} \right) + O \left( \frac{\omega^3 R^6}{\beta} \right) \,. \nonumber
\end{align}
We assumed $\omega^{-1/2} \ll R \leq \lambda \omega^{-1} \beta^{-1/2}$ with $\lambda>0$ small enough, $e^{-\beta \Lambda/4} \lesssim \beta \omega$, $\beta \omega \ln N \lesssim 1$, $b \geq c a_N$ for some $c>1$ as well as that $\Lambda^{3/2} a_N b^2 N_{\mathrm{th}}^2$ is small enough. To conclude the proof of the upper bound we will distinguish two cases, one where $\beta \omega \lesssim (\ln N)^{-1}$ and the other one where $(\beta \omega) \gtrsim (\ln N)^{-1}$. We start with the first one.

Let us start with the terms containing the cut-off $\Lambda$. With the choices $\beta \Lambda \sim N^{\delta}$ for $\delta >0$, $c a_N \leq b \lesssim  \omega^{-1/2} N^{-1}$ with $c>1$ and the estimate $(\beta \omega)^{-1} \lesssim N^{1/3}$, we find
\begin{equation}
	O \left( \frac{e^{-\beta \Lambda/4}}{\beta^4 \omega^3} \right) + O \left( \Lambda^3 a_N^2 b^2 N_{\mathrm{th}}^3 \right) \leq O\left( \omega N^{3 \delta} \right).
	\label{eq:upperboundfinal2}
\end{equation}
That is, these terms grow only with an arbitrarily small power of $N$. Moreover, $\Lambda^{3/2} a_N b^2 N_{\mathrm{th}}^2 \lesssim N^{3\delta/2 - 1/2} \ll 1$ if $\delta < 1/3$. 
When we use $N_{\mathrm{th}} \leq N$ and $N_0\leq N$, we see that
\begin{equation}
	O\left( \frac{\ln N}{\beta} \right) + O\left( \frac{a_N N_{\mathrm{th}}}{\beta^{3/2}}\right) + O\left(\omega N_0^{1/3} \right) + O\left( \frac{1}{\omega^{1/2} \beta^{3/2} } \right)  \leq O\left( \omega N^{1/2} \right)
\label{eq:upperboundfinal3}
\end{equation}
holds. Since $(\beta \omega)^{-3} + N_0 \lesssim N$ and $\omega^{-1/2} \ll R \leq \lambda \omega^{-1} \beta^{-1/2}$ the remaining error terms are bounded from above by
\begin{equation}
	O\left( \frac{N}{R^2} \right) + O\left( \frac{R^3}{\beta^{5/2}} \right)  \,. 
	\label{eq:upperboundfinal4}
\end{equation}
Optimization yields $R = \beta^{1/2} N^{1/5}$ which is in accordance with $R \leq \lambda \omega^{-1} \beta^{-1/2}$ as long as $\beta \omega N^{1/5} \leq \lambda$. If $\beta \omega N^{1/5} > \lambda$ we choose $R =  \lambda \omega^{-1} \beta^{-1/2}$ instead. Putting all this together, we find
\begin{equation}
F(\beta,N,\omega) - F_0(\beta,N,\omega) - E^{\mathrm{GP}}(N_0,a_N,\omega)  \lesssim \omega \begin{cases}  N^{3/5}  (\beta \omega)^{-1}  & \text{ if } \beta \omega \leq \lambda N^{-1/5} \\ N \beta \omega  & \text{ if } \lambda N^{-1/5} < \beta \omega \lesssim (\ln N)^{-1} \end{cases}
\label{eq:upperboundfinal5}
\end{equation} 
which is our final bound for the case $\beta \omega \lesssim (\ln N)^{-1}$.

To obtain a better bound for relatively large $\beta\omega$ and, in particular, to cover the case $\beta \omega \gtrsim (\ln N)^{-1}$, we proceed as follows. 
Let $E(N)$ denote the ground state energy of $H_N$ in Eq.~\eqref{eq:main1}. We use the upper bound on $E(N)$ in \cite{RobertGPderivation} and estimate 
\begin{equation}
F(\beta,N,\omega) \leq \inf_{\Vert \psi \Vert =1} \langle \Psi , H_N \Psi\rangle \leq E^{\mathrm{GP}}(N,a_N,\omega)(1+O(N^{-2/3}))
\label{eq:upperboundfinal6}
\end{equation}
Since $N = N_0 + N_\mathrm{th} \leq N_0 + O( (\beta\omega)^{-3})$ (see Remark~\ref{rem:app}), we also have $E^{\mathrm{GP}}(N,a_N,\omega) \leq E^{\mathrm{GP}}(N_0,a_N,\omega)(1+ O(N^{-1} (\beta\omega)^{-3}))$.  
Moreover, $F_0(\beta,N,\omega) \gtrsim -\omega (\beta \omega)^{-4}$, see Section~\ref{sec:freeenergyidealgas} and Appendix~\ref{sec:appendix1}. We therefore have
\begin{equation}
F(\beta,N,\omega) - F_0(\beta,N,\omega) - E^{\mathrm{GP}}(N_0,a_N,\omega) \lesssim \omega \left( N^{1/3} +( \beta \omega )^{-3} + ( \beta \omega )^{-4} \right).
\label{eq:upperboundfinal7}
\end{equation}
We use Eq.~\eqref{eq:upperboundfinal5} if $N^{-1/3} \lesssim \beta \omega \leq \lambda N^{-1/5}$ and Eq.~\eqref{eq:upperboundfinal7} for $\beta\omega > \lambda N^{-1/5}$, and arrive at
\begin{equation}
F(\beta,N,\omega) \leq F_0(\beta,N,\omega) + E^{\mathrm{GP}}(N_0,a_N,\omega) + O\left( \omega N^{14/15} \right).
\label{eq:upperboundfinal8}
\end{equation}
This completes the proof of the upper bound.

\section{Proof of the lower bound}
\label{sec:lowerbound}
As for the upper bound, the main idea of the proof of the lower bound is to make use of the two different length scales on which the condensate and the thermal cloud live. Anticipating the result that their respective particle numbers are equal to those of the ideal gas to leading order, this implies that the thermal cloud is much more dilute than the condensate and therefore does not see the interaction to the same order as the condensate does. For a more detailed discussion of these issues, see Section~\ref{sec:lengthscales}. The main technique to implement this mathematically is {\em geometric localization} in Fock space which has been introduced in \cite{DerezinskiGerard} in the context of bosonic quantum field theory. It allows, for the purpose of a lower bound,  to replace the free energy of the whole system by a sum of two free energies, one of a system localized in a ball with radius $2 R$ and another one of a system living in the complement of a ball with radius $R$. At this point the two systems are still correlated, however. The overlap of the two regions comes from the fact that we have to use smooth cut-off functions. The radius $R$ is, as in the proof of the upper bound, chosen such that $\omega^{-1/2} \ll R  \ll \omega^{-1} \beta^{-1/2}$ holds. In the following step, we minimize these two free energies separately which again results in a lower bound. For this  it is necessary to drop the restriction on the particle number and work in Fock space. The minimization procedure in the ball with radius $2 R$ yields the GP energy (as for the upper bound we use the known result at zero temperature here) plus lower order corrections coming from the entropy. In the complement of the ball with radius $R$ we drop the interaction (it is positive by assumption) and obtain the free energy of an ideal gas. 
Throughout the proof of the lower bound we shall assume 
that $(\beta\omega)^{-1} \lesssim N^{1/3}$ (i.e., $T\lesssim T_c$)  as well as that $R \omega^{1/2}$ is large enough.

The case where $N_0 = o(N)$ and/or $N_0 a_N =o(1)$ is quite simple for the lower bound so we discuss it first. On the one hand, the GP energy is of order $ o(\omega N)$ in this case. On the other hand, our interaction potential is nonnegative which allows us to drop it to obtain a lower bound on the free energy, i.e., $F(\beta,N,\omega)\geq F_0(\beta,N,\omega)$, which is sufficient in this case.  
The case where $N_0 \sim N$  is considerably more difficult and will be treated in the remaining part of this Section. 

Our first task is to replace the free energy of a given state $\Gamma_N$ by the sum of the free energies of two localized versions of $\Gamma_N$ and to show that this yields a lower bound. Let $j_1,j_2 \in \mathcal{C}^{\infty}(\mathbb{R}^3)$ be a partition of unity in the sense that $j_1(x)^2+j_2(x)^2=1$ for all $x \in \mathbb{R}^3$ and choose $j_1$ such that it equals one for $x \in B(R)$ and zero for $x \in  B(2R)^{\mathrm{c}}$. We can also assume that $|\nabla j_1(x) |^2 + |\nabla j_2(x)|^2 \leq 3 R^{-2}$ holds. By $\gamma^{(k)}_{\Gamma}$ we denote the $k$-particle reduced density matrix of a state $\Gamma$ on the bosonic Fock space $\mathcal{F}(\mathcal{H})$. It is defined via the integral kernel
\begin{equation}
\gamma^{(k)}_{\Gamma}(y_1,\ldots,y_k;x_1,\ldots,x_k) = \tr\left[ a_{x_1}^* \cdots a_{x_k}^* a_{y_k} \cdots a_{y_1} \Gamma \right]. \label{eq:lowerbound3b}
\end{equation}
The following well-known Lemma, see  \cite{DerezinskiGerard,HLSThemodynamicLimit}, concerns  localizations of a given state. Because its proof is instructive we sketch it here.
\begin{lemma}
\label{lem:localization}
Let $\Gamma$ be a state on $\mathcal{F}(\mathcal{H})$ with $k$-particle density matrices $\gamma_{\Gamma}^{(k)}$, $k\geq 1$ and $j_1$, $j_2$ as above. Then there exist unique states $\Gamma_{j_1}$ and $\Gamma_{j_2}$ on $\mathcal{F}(\mathcal{H})$ with $k$-particle density matrices $j_1^{\otimes k} \gamma_{\Gamma}^{(k)} j_1^{\otimes k}$ and $j_2^{\otimes k} \gamma_{\Gamma}^{(k)} j_2^{\otimes k}$, respectively. Moreover, the entropies of these states are related by
\begin{equation}
S(\Gamma) \leq S(\Gamma_{j_1}) + S(\Gamma_{j_2}).
\label{eq:lowerbound1a}
\end{equation} 
\end{lemma}
\begin{proof}
Define the map 
\begin{equation}
J : \mathcal{H} \rightarrow \mathcal{H} \oplus \mathcal{H}, \quad \psi \mapsto (j_1 \psi, j_2 \psi),
\label{eq:lowerbound1}
\end{equation}
where $\mathcal{H} = L^2(\mathbb{R}^3)$. That is, $J$ splits the wavefunction into two parts, one supported in $B(2R)$ and one supported in $B(R)^{\mathrm{c}}$. Let $A$ be an operator on $\mathcal{H}$. We denote by $\Upsilon(A)$ the second quantization of $A$ acting on $\mathcal{F}(\mathcal{H})$. On basis vectors its action is defined by
\begin{equation}
\Upsilon(A) a^*(\psi_1) \cdots a^*(\psi_n) \vert \Omega \rangle = a^*(A \psi_1) \cdots a^*(A \psi_n) \vert \Omega \rangle, 
\label{eq:lowerbound2}
\end{equation}
where $\Omega$ denotes the Fock space vacuum. It extends to all  Fock space vectors by linearity. Let $U$ be the unitary map that identifies $\mathcal{F}(\mathcal{H} \oplus \mathcal{H})$ and $\mathcal{F}(\mathcal{H}) \otimes \mathcal{F}(\mathcal{H})$ given by
\begin{equation}
U  a^*(\psi_1) \cdots a^*(\psi_n) a^*(\varphi_1) \cdots a^*(\varphi_n)  \vert \Omega \rangle = c^*(\psi_1) \cdots c^*(\psi_n) d^*(\varphi_1) \cdots d^*(\varphi_n)  \vert \Omega_1 \otimes \Omega_2 \rangle.
\label{eq:lowerbound3}
\end{equation}
Here $\lbrace \psi_i \rbrace_{i=0}^{\infty}$ denotes a basis for the first copy of $\mathcal{H}$ and $\lbrace \varphi_i \rbrace_{0=1}^{\infty}$ is a basis for the second copy. By $c^*$/$d^*$ we denote the creation operator acting on the first/second factor of $\mathcal{F}(\mathcal{H}) \otimes \mathcal{F}(\mathcal{H})$ and $\Omega_{1}$/$\Omega_2$ is the vacuum in the first/second factor. Having those two operators at hand, we can define the $J$-extension $\Gamma_J$ of $\Gamma$ by
\begin{equation}
\Gamma_J = U \Upsilon(J) \Gamma \Upsilon(J)^* U^*.
\label{eq:lowerbound3c}
\end{equation}
It follows that the states
\begin{equation}
\Gamma_{j_1} = \tr_2 \left( \Gamma_J \right) \quad \text{ and } \quad \Gamma_{j_2} = \tr_1 \left( \Gamma_J \right)
\label{eq:lowerbound3d}
\end{equation}
have the desired property. Here $\tr_{1,2}$ denotes the partial trace over the first/second factor of $\mathcal{F}(\mathcal{H}) \otimes \mathcal{F}(\mathcal{H})$. The uniqueness part follows from the fact that states are uniquely determined by their reduced density matrices, which we are not going to discuss here. The proof of Eq.~\eqref{eq:lowerbound1a} follows from Eq.~\eqref{eq:lowerbound3d} and the subadditivity of the entropy.
\end{proof}
Before we continue, let us introduce some notation. For an operator $A$ on $\mathcal{H}$ let the operator $\text{d}\Upsilon(A)$ on $\mathcal{F}(\mathcal{H})$ be given by $\text{d}\Upsilon(A)|_{\mathcal{F}_N} = \sum_{i=1}^{{N}} A_i$ where $A_i$ stands for $A$ acting on the $i$-th particle in the the Fock space sector  $\mathcal{F}_N$ with $N$ particles. We also define 
\begin{equation}
H_{\leq 2 R}^{\mathrm{D}} = \text{d}\Upsilon\left(h_{\leq 2 R}^{\mathrm{D}} \right) + \bigoplus_{N \geq 2} \sum_{1 \leq i < j \leq N}  v_N(x_i-x_j)
\label{eq:lowerbound7b}
\end{equation}
on $\mathcal{F}(L^2(B(2R))$, where $h_{\leq 2 R}^{\mathrm{D}}$ denotes the operator \eqref{eq:idealbosegas1a} restricted to $B(2R)$ with Dirichlet boundary conditions.  By $\hat{N}$ we denote the particle number operator on Fock space.

Let $\Gamma_N \in \mathcal{S}_N$ be an $N$-particle state. Using the IMS localization formula (see e.g. \cite{SimonSchroedingerOps}) and the fact that $v(x) \geq 0$ for all $x \in \mathbb{R}^3$, we find
\begin{equation}\label{eq:lowerbound6}
\tr\left( H_N \Gamma_N \right)  \geq \tr\left( h j_1 \gamma_{\Gamma_N} j_1 \right) + \tr\left( h j_2 \gamma_{\Gamma_N} j_2 \right) + \tr\left( V j_1^{\otimes 2} \gamma_{\Gamma_N}^{(2)} j_1^{\otimes 2} \right)  \\
- 3 N R^{-2}. 
\end{equation}
Lemma~\ref{lem:localization} tells us that we can bound the entropy of $\Gamma_N$ by the ones of $\Gamma_{j_1}$ and $\Gamma_{j_2}$, the localized states related to $\Gamma_N$. It also allows us to write the energies in Eq.~\eqref{eq:lowerbound6} in terms of $\Gamma_{j_1}$ and $\Gamma_{j_2}$. This implies
\begin{equation}
\tr\left( H_N \Gamma_N \right) - TS(\Gamma_N) \geq \tr \left( H_{\leq 2 R}^{\mathrm{D}} \Gamma_{j_1} \right) - T S\left( \Gamma_{j_1} \right) + \tr\left( \text{d}\Upsilon\left( h_{\geq R}^{\mathrm{D}} \right) \Gamma_{j_2} \right) - T S\left( \Gamma_{j_2} \right) - 3 N R^{-2}
\label{eq:lowerbound6a}
\end{equation}
with $h_{\geq R}^{\mathrm{D}}$ the operator $h$ in \eqref{eq:idealbosegas1a} restricted to $B(R)^c$ with Dirichlet boundary conditions. 
The states $\Gamma_{j_1}$ and $\Gamma_{j_2}$ are related via the fact that they both are constructed by localizing the state $\Gamma_N$. In the next step we will minimize over each of them separately which results in a lower bound. We will also drop the restriction on the particle number. To do so, we have to introduce a chemical potential (and subtract it again). 

The intuition from Section~\ref{sec:lengthscales} tells us that only the condensate is affected by the interaction to leading order. Since the interaction energy per particle inside the condensate can be expected to be of order $\omega$,  while $-\mu_0 \sim \beta^{-4} \omega^{-3}$ is much smaller when $N_0\sim N$ (see Section~\ref{sec:freeenergyidealgas}), the chemical potential $\mu$ for the interacting system will be of order $\omega$, too. In fact, in the GP limit the GP chemical potential $\mu^\mathrm {GP} = \mu^{\mathrm{GP}}(N_0, a_N,\omega )$, defined by
\begin{equation}
\mu^{\mathrm{GP}}(N_0,a_N, \omega) = \frac{\text{d}E^{\mathrm{GP}}}{\text{d}N}(N_0,a_N,\omega) = \frac{E^{\mathrm{GP}}(N_0,a_N,\omega)}{N_0} + \frac{4 \pi a_N}{N_0} \int_{\mathbb{R}^3} \left\vert \phi^{\mathrm{GP}}_{N_0,a_N}(x) \right\vert^4 \text{d}x,
\label{eq:lowerbound6b}
\end{equation}
is the correct choice for $\mu$. We note that because of scaling (see Eqs.~\eqref{eq:scaling} and~\eqref{eq:GPscaling}) $\mu^{\mathrm{GP}} = \omega \mu^\mathrm{GP}(1,a_v,1)$. In particular, $\mu^\mathrm{GP} \sim \omega$ for fixed $a_v> 0$. The chemical potential $\mu^{\mathrm{GP}}$ necessarily also appears in the thermal cloud, but as we will see below, this does not affect its free energy at the level of accuracy we are interested in.

Using the explicit form of the one-particle density matrices of $\Gamma_{j_1}$ and $\Gamma_{j_2}$ (see Lemma~\ref{lem:localization}), we check that $\tr ( \gamma_{\Gamma_{j_1}} + \gamma_{\Gamma_{j_2}} ) = N$ 
holds. Together with Eq.~\eqref{eq:lowerbound6a}, this implies
\begin{align}
\tr\left( H_N \Gamma_N \right) - TS(\Gamma_N) &\geq \tr \left[ \left( H_{\leq 2 R}^{\mathrm{D}} -\mu^{\mathrm{GP}} \hat{N} \right) \Gamma_{j_1} \right] - T S\left( \Gamma_{j_1} \right)  - 3 N R^{-2} \label{eq:lowerbound7} \\
&\quad + \tr\left[ \left( \text{d}\Upsilon\left( h_{\geq R}^{\mathrm{D}} \right) - \mu^{\mathrm{GP}} \hat{N} \right) \Gamma_{j_2} \right] - T S\left( \Gamma_{j_2} \right) + \mu^{\mathrm{GP}} N. \nonumber
\end{align}
From here on we estimate the two contributions on the right-hand side of Eq.~\eqref{eq:lowerbound7} separately,  and start with  the first line.
\subsection{The condensate}
\label{sec:lowerboundcondensate}
For the minimization problem inside $B(2 R)$ we take a small amount of the kinetic energy to control the entropy, which results in a lower order contribution. The minimization problem for the remaining part of the energy  then follows from the results in \cite{RobertGPderivation} for the ground state energy. Let $0<\epsilon <1$. From the positivity of the interaction potential $v$, we conclude that 
\begin{align}
\tr \left[ \left( H_{\leq 2 R}^{\mathrm{D}} -\mu^{\mathrm{GP}} \hat{N} \right) \Gamma_{j_1} \right] - T S\left( \Gamma_{j_1} \right) &\geq (1-\epsilon) \tr \left[ \left( H_{\leq 2 R}^{\mathrm{D}} -\mu^{\mathrm{GP}} \hat{N} \right) \Gamma_{j_1} \right] \label{eq:lowerbound7a} \\
&\quad + \epsilon \tr\left[ \left( \text{d}\Upsilon\left( h_{\leq 2R}^{\mathrm{D}} \right) - \mu^{\mathrm{GP}}\hat{N} \right) \Gamma_{j_1} \right]  - T S\left( \Gamma_{j_1} \right). \nonumber 
\end{align}
We will later choose $\epsilon$ such that $\epsilon \ll 1$ holds.

We start by considering the first term on the right-hand side of Eq.~\eqref{eq:lowerbound7a}.
We denote by $E^{\mathrm{D}}_{\leq 2 R}(M)$ the ground state energy of the Hamiltonian $H_{\leq 2R}^{\mathrm{D}}$ \eqref{eq:lowerbound7b} when restricted to the $M$-particle sector of the Fock space. Also, let $E(N)$ be the ground state energy of the Hamiltonian \eqref{eq:main1}. By the variational principle for the energy, we have $E^{\mathrm{D}}_{\leq 2 R}(M) \geq E(M)$ for all $M \geq 0$ which allows us to get rid of the Dirichlet boundary conditions at $\partial B(2R)$. In particular, it implies
\begin{equation}
\tr \left[ \left( H_{\leq 2 R}^{\mathrm{D}} -\mu^{\mathrm{GP}}\right)\Gamma_{j_1} \right] \geq \inf_{0 \leq M \leq N} \left\lbrace E(M) - \mu^{\mathrm{GP}}M \right\rbrace.
\label{eq:lowerbound7b2}
\end{equation}
To bound the infimum on the right-hand side of the above equation, we distinguish two different regimes for the particle number $M$. This is necessary because we cannot relate $E(M)$ to the GP energy \eqref{eq:mainresult2} if $M$  is too small.

We first consider the case where $M \ll N_0$. Choose $\delta > 0$ such that $\delta \ll 1$ and still $\delta N_0 \gg 1$ and assume $M \leq \delta N_0$. When we drop the positive energy $E(M)$, we obtain the lower bound
\begin{equation}
E(M) - \mu^{\mathrm{GP}} M \geq -\mu^{\mathrm{GP}} \delta N_0.
\label{eq:lowerbound7d}
\end{equation}
Since the term on the right-hand side is of order $ o(\omega N_0)$ it is not important for the minimization  on the right-hand side of Eq.~\eqref{eq:lowerbound7b2}. 
For $\delta N_0 < M \leq N$ we apply \cite[Thm.~IV.1]{RobertGPderivation} and obtain $E(M) \geq E^{\mathrm{GP}}(M,a_N,\omega) (1- o(1))$. Together with Eq.~\eqref{eq:lowerbound7b2}, Eq.~\eqref{eq:lowerbound7d} and the convexity of the map $N_0 \mapsto E^{\mathrm{GP}}(N_0,a_N,\omega)$, this proves the lower bound
\begin{align}
\tr\left( \left( H^{\mathrm{D}}_{\leq 2 R} - \mu^{\mathrm{GP}}\hat{N} \right) \Gamma_{j_1} \right) &\geq \min_{\delta N_0 \leq M \leq N} \left(E^{\mathrm{GP}}(M,a_N,\omega) - \mu^{\mathrm{GP}} M \right) - o(\omega N) \nonumber \\
&\geq E^{\mathrm{GP}}(N_0,a_N,\omega) - \mu^{\mathrm{GP}} N_0 - o(\omega N) \label{eq:lowerbound8}
\end{align}
for the energy inside $B(2R)$. 

\begin{remark}
Let us make a short remark concerning the grand-canonical lower bound (compare with Remark 9 in Section~\ref{sec:remarks}) in which case one cannot restrict attention to $M \leq N$: In this case we can use the superadditivity of the energy $E(M)$ and \cite[Thm.~IV.1]{RobertGPderivation} to observe that for $\eta \in \mathbb{N}$ with $\eta\sim N_0$ we have $E(M) \geq M \eta^{-1} E(\eta ) \geq M \eta^{-1} E^{\mathrm{GP}}(\eta ,a_N,\omega)(1-o(1))$. Choosing $\eta = O(N_0)$ large enough such that $E^{\mathrm{GP}}(\eta ,a_N,\omega) - \mu^{\mathrm{GP}} \eta$ is larger than a constant times $\omega \eta$, one checks that values of $M$ with $M>\eta $ are not relevant for the computation of the minimum in Eq.~\eqref{eq:lowerbound8}. The rest of the argument remains unchanged. 
\end{remark}

Next, we estimate the contribution coming from the entropy in the region $B(2 R)$. We use $\tr[\hat{N} \Gamma_{j_1} ] \leq N$ 
to see that the term in the second line of Eq.~\eqref{eq:lowerbound7a} is bounded from below by
\begin{equation}
\epsilon \tr\left[ \left( \text{d}\Upsilon\left(h_{\leq 2 R}^{\mathrm{D}} \right) - \mu^{\mathrm{GP}} \hat{N} \right] \Gamma_{j_1} \right) - TS(\Gamma_{j_1}) \geq \tfrac{1}{\beta} \tr \left[ \ln \left( 1 - e^{-\beta \epsilon \left( h_{\leq 2 R}^{\mathrm{D}} + \frac 32\omega \right)} \right) \right] - \epsilon N \left( \tfrac 32 \omega + \mu^{\mathrm{GP}}  \right). \label{eq:lowerbound9}
\end{equation}
Note that we have added and subtracted a chemical potential of size $-\frac 32\omega$. To bound the first term on the right-hand side, we note that the  $h^{\mathrm{D}}_{\leq 2R}$ is bounded from below by $-\Delta^{\mathrm{D}}_{\leq 2R} - \frac 32 \omega$. 
Using the Weyl asymptotics  \cite[Satz~XI]{weyl}  of the eigenvalues of $-\Delta^{\mathrm{D}}_{\leq 2R}$, we see that there exists a constant $C>0$ such that the $\alpha$th eigenvalue of this operator satisfies $e_{\alpha}(-\Delta^{\mathrm{D}}_{\leq R}) \geq C (\alpha+1)^{2/3}/R^2$. This allows us to estimate 
\begin{equation}
 \tr \left[ \ln \left( 1 - e^{-\beta \epsilon \left( h_{\leq 2 R}^{\mathrm{D}} +\frac 32 \omega \right)} \right) \right]   \geq  \tr \left[ \ln \left( 1 - e^{-\beta \epsilon \left(  -\Delta^{\mathrm{D}}_{\leq 2R}  \right)} \right) \right]  \geq  
 \sum_{\alpha\geq 1} \ln \left( 1 - e^{- C \beta \epsilon  \alpha^{2/3} R^{-2} } \right). \label{eq:lemupperbound17b}
\end{equation}
The last sum  can be interpreted a Riemann sum approximating the corresponding integral, and one readily checks that it is bounded below by $-O(R^3 \beta^{-3/2} \epsilon^{-3/2}$). That is,
\begin{equation}
\epsilon \tr\left[ \left( \text{d}\Upsilon\left(h_{\leq 2 R}^{\mathrm{D}} \right) - \mu^{\mathrm{GP}} \hat{N} \right] \Gamma_{j_1} \right) - TS(\Gamma_{j_1}) \gtrsim - \frac{R^3}{\beta^{5/2}\epsilon^{3/2}} - \epsilon N \omega\,. \label{eq:lowerbound9as}
\end{equation}
Together with Eqs.~\eqref{eq:lowerbound7a},~\eqref{eq:lowerbound8} and~\eqref{eq:lowerbound9} this yields
\begin{align}\nonumber
\tr \left[ \left( H_{\leq 2 R}^{\mathrm{D}} -\mu^{\mathrm{GP}} \hat{N} \right) \Gamma_{j_1} \right] - T S\left( \Gamma_{j_1} \right) &\geq (1-\epsilon) \left[ E^{\mathrm{GP}}(N_0,a_N,\omega) - \mu^{\mathrm{GP}} N_0 \right]  \\
& \quad - o(\omega N) - O\left( \omega \epsilon N \right) - O\left( \frac{R^3}{\epsilon^{3/2} \beta^{5/2}} \right)  \label{eq:lowerbound9a}
\end{align} 
as a lower bound for the free energy inside $B(2 R)$.
\subsection{The thermal cloud}
Next we consider the  terms in the second line on the right-hand side of Eq.~\eqref{eq:lowerbound7}. Explicit minimization shows that 
\begin{equation}
\tr\left( \left( \text{d}\Upsilon(h_{\geq R}^{\mathrm{D}}) - \mu^{\mathrm{GP}} \hat{N} \right) \Gamma_{j_2} \right) - TS(\Gamma_{j_2}) \geq \tfrac{1}{\beta} \tr \left[ \ln\left(1- e^{-\beta \left( h_{\geq R}^{\mathrm{D}} - \mu^{\mathrm{GP}} \right)} \right) \right]. \label{eq:lowerbound10} 
\end{equation}
To relate the right-hand side of Eq.~\eqref{eq:lowerbound10} to $F_0(\beta,N,\omega)$, we first need to replace $\mu^{\mathrm{GP}}$ by $\mu_0$, the chemical potential of the ideal gas leading to an expected number of $N$ particles. After the chemical potential has been replaced, we have to get rid of the Dirichlet boundary conditions in the formula for the grand canonical free energy, and replace it by its canonical version.

To replace $\mu^\mathrm{GP}$ by $\mu_0$, we use convexity to bound
\begin{align}\nonumber
 &\tr\left[ \ln\left( 1 - e^{-\beta \left( h_{\geq R}^{\mathrm{D}}- \mu^{\mathrm{GP}} \right)} \right) \right] - \tr\left[ \ln\left( 1 - e^{-\beta \left( h_{\geq R}^{\mathrm{D}}- \mu_0 \right)} \right) \right]  \\ \nonumber& \geq \beta(\mu_0 - \mu^\mathrm{GP} )   \tr \left[ \frac 1 {e^{\beta ( h_{\geq R}^{\mathrm{D}} - \mu^\mathrm{GP} ) }-1} \right] \\  & \geq
 \beta(\mu_0 - \mu^\mathrm{GP} )   \tr \left[ \frac 1 {e^{\beta ( h_{\geq R}^{\mathrm{D}} - \mu_0 ) }-1} \right] 
  - \beta^2 (\mu_0-\mu^\mathrm{GP}) ^2 \tr\left[ \frac{ e^{\beta (h_{\geq R}^\mathrm{D} -  \mu^{\mathrm{GP}}) } }{\left(e^{\beta ( h_{\geq R}^{\mathrm{D}} - \mu^\mathrm{GP} )}-1\right)^2 } \right] \,. \label{eq:lowerbound12} 
\end{align}
By arguing as in Eq.~\eqref{eqsq}, we see that the last trace is bounded by $O( (\beta\omega)^{-3})$. 
It remains to get rid of the Dirichlet boundary conditions in the various terms in Eq.~\eqref{eq:lowerbound12}.
To that end, we use  that the $\alpha$-th eigenvalue of $h_{\geq R}^{\mathrm{D}}$ is bounded from below by $e_{\alpha}(h_{\geq R}^{\mathrm{D}}) \geq e_{\alpha}(h)$. For the second term on the left-hand side of \eqref{eq:lowerbound12}, this implies
\begin{equation}
\tr \left[ \ln\left(1- e^{-\beta (h_{\geq R}^{\mathrm{D}} - \mu_0)} \right) \right] \geq \tr\left[ \ln \left( 1 - e^{-\beta \left( h - \mu_0 \right)} \right) \right]. \label{eq:lowerbound11}
\end{equation}  
Note that the above bound is rough in the sense that the right-hand side includes the grand canonical potential of the condensate. The latter is negligible in the limit considered, however. The first term on the right-hand side of Eq.~\eqref{eq:lowerbound12} is proportional to the particle number and hence we have to be more careful. We use $e_{\alpha}(h_{\geq R}^{\mathrm{D}}) \geq e_{\alpha}(h)$ for $\alpha \geq 1$ and $e_{0}(h_{\geq R}^{\mathrm{D}}) \geq \frac{\omega^2 R^2}{4} - \frac{3 \omega }{2}$, which avoids adding the expected number of particles in the condensate, and find
\begin{equation}
\tr\left[ \frac{1}{e^{\beta ( h_{\geq R}^{\mathrm{D}} - \mu_0 )}-1} \right] \leq \tr\left[ \mathds{1}\left( h \geq \omega \right) \frac{1}{e^{\beta \left( h - \mu_0 \right)}-1} \right] + O\left( \beta^{-1}\omega^{-2} R^{-2} \right).
\label{eq:lowerbound11b}
\end{equation}
Note that the first term on the right-hand side of Eq.~\eqref{eq:lowerbound11b} is $N_{\mathrm{th}}^{\mathrm{gc}}$, the expected number of particles in the grand canonical thermal cloud. Combining Eqs.~\eqref{eq:lowerbound10}--\eqref{eq:lowerbound11b} and using that $|\mu_0-\mu^\mathrm{GP}|\lesssim \omega$ 
 we find
\begin{align}
&\tr\left( \left( \text{d}\Upsilon(h_{\geq R}^{\mathrm{D}}) - \mu^{\mathrm{GP}} \hat{N} \right) \Gamma_{j_2} \right) - TS(\Gamma_{j_2}) \geq \tfrac{1}{\beta} \tr\left[ \ln \left( 1 - e^{-\beta \left( h - \mu_0 \right)} \right) \right]  - \left( \mu^{\mathrm{GP}} - \mu_0 \right) N_{\mathrm{th}}^{\mathrm{gc}}  \label{eq:lowerbound15} \\
&\hspace{8cm}  - O( \beta^{-2} \omega^{-1} ) - O (\beta^{-1} \omega^{-1} R^{-2})  \,. \nonumber
\end{align}

It remains to replace the grand canonical free energy by the canonical one, and $N_{\mathrm{th}}^{\mathrm{gc}}$ by $N_{\mathrm{th}}$. 
Corollary~\ref{cor:freeenergy} tells us that the difference of the canonical and the grand canonical free energy is at most of order $- T \ln N $. Moreover, $| N_{\mathrm{th}} -  N_{\mathrm{th}}^{\mathrm{gc}} | \lesssim (\beta \omega)^{-3/2} (\ln N)^{1/2} + (\beta \omega)^{-1} \ln N$ by Corollary~\ref{lem:particlenumbers}. We therefore have
\begin{align}
\tr\left( \left( \text{d}\Upsilon(h_{\geq R}^{\mathrm{D}}) - \mu^{\mathrm{GP}} \hat{N} \right) \Gamma_{j_2} \right) - TS(\Gamma_{j_2}) &\geq  F_0(\beta,N,\omega) - \mu^{\mathrm{GP}} N_{\mathrm{th}} - O\left( \frac{1}{\beta^2 \omega} \right) - O\left( \frac 1{\beta \omega R^2}\right)  \nonumber  \\
& \quad  - O\left( \frac{(\ln N)^{1/2}}{\omega^{1/2} \beta^{3/2}} \right) - O\left( \frac{\ln N}{\beta} \right) \label{eq:lowerbound16} 
\end{align}
as a lower bound for the free energy in region $B(R)^{\mathrm{c}}$. Note that we have added the additional negative term $\mu_0 N_0^\textrm{gc}$ on the right hand side.

\subsection{The final estimate for the lower bound}
\label{sec:finalestimatelowerbound}
We combine the results from Eqs.~\eqref{eq:lowerbound7},~\eqref{eq:lowerbound9a} and~\eqref{eq:lowerbound16} and the fact that $E^{\mathrm{GP}}(N_0,a_N,\omega) - \mu^{\mathrm{GP}} N_0 = O(\omega N)$ to find
\begin{align}
\tr(H_N \Gamma_N) - TS(\Gamma_N) &\geq  F_0(\beta,N,\omega)  + E^{\mathrm{GP}}(N_0,a_N,\omega) - o(\omega N) \nonumber \\
&\quad  - O(\omega \epsilon N) - 3 N R^{-2}   - O\left( \frac{R^3}{\epsilon^{3/2} \beta^{5/2} } \right) \,.  \label{eq:lowerbound17} 
\end{align}
To obtain the result we assumed that $R\omega^{1/2}$ is large enough,  and used that $(\beta\omega)^{-1} \lesssim N^{1/3}$ to dominate some of the error terms by others.  The optimal choice of the parameters $R$ and $\epsilon$ turns out to be $R \sim \omega^{-1/2} N^{1/8} (\beta\omega)^{5/16}$ and $\epsilon \sim N^{-1/4} (\beta\omega)^{-5/8}$. The three terms on the second line of the right hand side of Eq.~\eqref{eq:lowerbound17} are thus bounded by $\omega N^{23/24}$.  In particular, 
\begin{equation}
F(\beta,N,\omega) \geq F_0(\beta,N,\omega) + E^{\mathrm{GP}}(N_0,a_N,\omega) - o( \omega N ) \,. 
\label{eq:lowerbound21}
\end{equation}
This completes the proof of the lower bound.

\section{Proof of the asymptotics of the one-particle density matrix}
\label{sec:densitymatrix}
In the following discussion, we assume that $N_0 \sim N$ 
and $N a_N \geq \epsilon \omega^{-1/2}$, i.e. $a_v \geq \epsilon$ for some $\epsilon>0$ holds. The last assumption is necessary keeping in mind the second remark in Section~\ref{sec:remarks}. The case where one of these conditions is not fulfilled will be taken care of at the end. 

Assume we are given a sequence of states $\Gamma_N$ with reduced one-particle density matrices $\gamma_N$ such that
\begin{equation}
\tr\left[ H_N \Gamma_N \right] - TS(\Gamma_N) = F_0(\beta,N,\omega) + E^{\mathrm{GP}}(N_0,a_N,\omega) + o(\omega N)
\label{eq:densitymatrix1}
\end{equation}
as $N$ tends to infinity. Choose two functions $j_1$ and $j_2$ as in the proof of the lower bound, satisfying $j_1(x)^2 + j_2(x)^2 = 1$ for all $x \in \mathbb{R}^3$, $j_1(x) = 1$ for $x \in B(R)$ and $j_1(x) = 0$ for $x \in B(2R)^c$, and also  $|\nabla j_1(x) |^2 + |\nabla j_2(x)|^2 \leq 3 R^{-2}$. Using  Eqs.~\eqref{eq:lowerbound7},~\eqref{eq:lowerbound7a},~\eqref{eq:lowerbound9a} and~\eqref{eq:lowerbound16} and the choice of parameters from Section~\ref{sec:finalestimatelowerbound}, we see that
\begin{equation}
o(\omega N)  \geq \tr \left[ \left( H_{\leq 2 R}^{\mathrm{D}} - \mu^{\mathrm{GP}} \hat{N} \right) \Gamma_{N,j_1} \right] - E^{\mathrm{GP}}(N_0,a_N,\omega) + \mu^{\mathrm{GP}}N_0
\label{eq:densitymatrix2a}
\end{equation}
holds. The operator $H_{\leq 2 R}^{\mathrm{D}}$ was defined in Eq.~\eqref{eq:lowerbound7b}, $\mu^{\mathrm{GP}}$ is given by \eqref{eq:lowerbound6b} and the state $\Gamma_{N,j_1}$ is related to $\Gamma_N$ in the way described in the proof of Lemma~\ref{lem:localization}. Moreover, Eqs.~\eqref{eq:lowerbound7} and~\eqref{eq:lowerbound9a} 
together with the fact that $\mu_0 N_0< 0$ tell us that
\begin{align}
o(\omega N)  &\geq \tr\left[ \left( \text{d}\Upsilon\left( h_{\geq R}^{\mathrm{D}} \right) - \mu^{\mathrm{GP}} \hat{N} \right) \Gamma_{N,j_2} \right] - T S\left( \Gamma_{N,j_2} \right) \label{eq:densitymatrix2b} \\
&\quad - \tfrac{1}{\beta}\tr\left[ \ln\left( 1-e^{-\beta(h-\mu_0)} \right) \right]  + \left( \mu^{\mathrm{GP}} - \mu_0 \right) N_{\mathrm{th}}\,, \nonumber
\end{align}
where also  $\Gamma_{N,j_2}$ is related to $\Gamma_N$ in the way described in the proof of Lemma~\ref{lem:localization}. Note that we have the grand canonical free energy on the right-hand side of Eq.~\eqref{eq:densitymatrix2b} instead of the canonical free energy, which is allowed by Corollary~\ref{cor:freeenergy} in the Appendix. 

Eqs.~\eqref{eq:densitymatrix2a} and~\eqref{eq:densitymatrix2b} will be used to deduce the desired bounds on $\gamma_N$. 
As a first step we will derive asymptotic expressions for the one-particle density matrices of $\Gamma_{N,j_1}$ and $\Gamma_{N,j_2}$, that is, for $j_1 \gamma_N j_1$ and $j_2 \gamma_N j_2$, respectively. Afterwards, we consider the \lq\lq off-diagonal\rq\rq\ contribution coming from $j_1 \gamma_N j_2$ and $j_2 \gamma_N j_1$ and show that their trace norm is of order $o(N)$. As one would expect, $j_1 \gamma_N j_1$ turns out to be close to $N_0 \vert \phi^{\mathrm{GP}}_{1,N_0 a_N} \rangle\langle \phi^{\mathrm{GP}}_{1,N_0 a_N} \vert$ and $j_2 \gamma_N j_2$ is close to $\gamma_{N,0} - N_0 \vert \varphi_0 \rangle\langle \varphi_0 \vert$. 
\subsection{The bound for $j_1 \gamma_N j_1$}
To derive a bound for $j_1 \gamma_N j_1$, we make use of existing results \cite{LiSei2002, LiSei2006, largecoulombsystems} on the convergence of the one-particle density matrix of approximate minimizers of the ground state energy functional to the projection onto the GP minimizer. The main difficulty to overcome is that the particle number of the state $\Gamma_{N,j_1}$ may fluctuate, that is, it is a state on the full Fock space.

As in Section~\ref{sec:lowerboundcondensate} we choose $0 < \delta \ll 1$ such that still $\delta N_0 \gg 1$. Let $P_M$ be the projection onto  the Fock space sector with $M$ particles. Keeping in mind the normalization  $\tr[ \Gamma_{N,j_1} ] = 1$, Eq.~\eqref{eq:densitymatrix2a} can be written as
\begin{equation}
o(\omega N) \geq \sum_{M=0}^{N} \tr\left[ \left\lbrace \left(H^{\mathrm{D}}_{\leq 2R} - \mu^{\mathrm{GP}} M \right) - \left( E^{\mathrm{GP}}(N_0,a_N,\omega) - \mu^{\mathrm{GP}} N_0 \right) \right\rbrace P_M \Gamma_{N,j_1} P_M \right] . \label{eq:aprioricondensate1}
\end{equation}   
Let us again distinguish two cases: For $0 \leq M \leq \delta N_0$ we drop $H^{\mathrm{D}}_{\leq 2R}$ to obtain a lower bound and use $\mu^{\mathrm{GP}} \delta N_0 \ll \omega N_0$ to show that there is a constant $C_1 > 0$ such that the expression in the curly brackets in the above equation is bounded from below by $C_1 \omega N_0$. In the case where $\delta N_0 < M \leq N$, we invoke \cite[Eq.~(101) and Lemma~4]{largecoulombsystems} to see that there exists a constant $C_2 >0$ such that the summand in Eq.~\eqref{eq:aprioricondensate1} is bounded from below by 
\begin{align}\nonumber
& \tr\left[ P_M \Gamma_{N,j_1} P_M \right] \Big\lbrace E^{\mathrm{GP}}(M,a_N,\omega) - E^{\mathrm{GP}}(N_0,a_N,\omega) - \mu^{\mathrm{GP}} (M-N_0)  \\
&\hspace{4.5cm}+ \frac{\omega C_2}M  \left\Vert \gamma_{P_M \Gamma_{N,j_1} P_M} - M P_{Ma_N}^{\mathrm{GP}} \right\Vert_1^2  - o(\omega N) \Big\rbrace.  \label{eq:aprioricondensate2}
\end{align}
Here $\gamma_{P_M \Gamma_{N,j_1} P_M}$ denotes the one-particle density matrix of the state $P_M \Gamma_{N,j_1} P_M (\tr[P_M \Gamma_{N,j_1} P_M])^{-1}$. It is normalized to have  $\tr[ \gamma_{P_M \Gamma_{j_1} P_M} ] = M$. By $P_{M a_N}^{\mathrm{GP}}$ we denote the projection onto the GP minimizer $\phi^{\mathrm{GP}}_{1,M a_N}$. 
This estimate is in fact uniform in $M a_N$ for $M$ in the range we consider.

We shall use the strict convexity of $M \mapsto E^{\mathrm{GP}}(M,a_N,\omega)$ in order to obtain a lower bound on the first three terms in the above parentheses that is strictly positive for $M \neq N_0$. 
Using that $\mu^{\mathrm{GP}}  = \tfrac{\text{d}}{\text{d}N} E^{\mathrm{GP}}(N_0,a_N,\omega)$, as well as the convexity of $\rho\mapsto \int |\nabla\sqrt{\rho}|^2$, we deduce that 
\begin{equation}
E^{\mathrm{GP}}(M,a_N,\omega) - E^{\mathrm{GP}}(N_0,a_N,\omega) - \mu^{\mathrm{GP}}(M - N_0) \geq  4\pi a_N \int_{\mathbb{R}^3} \left( \phi^\mathrm{GP}_{M,a_N}(x)^2 - \phi^\mathrm{GP}_{N_0,a_N}(x)^2 \right)^2 \text{d}x    \,.  \label{eq:aprioricondensate4}
\end{equation}
For a lower bound, we pick $s>0$ and $t \in \mathbb{R}$ and estimate
\begin{align}\nonumber
& \int_{\mathbb{R}^3} \left( \phi^\mathrm{GP}_{M,a_N}(x)^2 - \phi^\mathrm{GP}_{N_0,a_N}(x)^2 \right)^2 \text{d}x  \geq    \int_{|x|<s} \left( \phi^\mathrm{GP}_{M,a_N}(x)^2 - \phi^\mathrm{GP}_{N_0,a_N}(x)^2 \right)^2 \text{d}x \\ & \geq 2 t  \int_{|x|<s} \left( \phi^\mathrm{GP}_{M,a_N}(x)^2 - \phi^\mathrm{GP}_{N_0,a_N}(x)^2 \right) \text{d}x  - \frac {4\pi}{3} t^2 s^3\,.  \label{eq:aprioricondensate40}
\end{align}
Since
\begin{equation}\label{eq:aprioricondensate41}
  \int_{|x|\geq s}  \phi^\mathrm{GP}_{M,a_N}(x)^2  \text{d}x \leq   s^{-2} \int_{\mathbb{R}^3}  \phi^\mathrm{GP}_{M,a_N}(x)^2  |x|^2 \text{d}x \leq 4 s^{-2} \omega^{-2} \left( E^\mathrm{GP}(M,a_N,\omega) + \tfrac 32 M \omega \right) 
\end{equation}
this implies
\begin{align}\nonumber 
& \int_{\mathbb{R}^3} \left( \phi^\mathrm{GP}_{M,a_N}(x)^2 - \phi^\mathrm{GP}_{N_0,a_N}(x)^2 \right)^2 \text{d}x  \\ & \geq  2 t (M-N_0) - 8 |t| s^{-2} \omega^{-2} \left( E^\mathrm{GP}(M,a_N,\omega) + E^\mathrm{GP}(N_0,a_N,\omega) +  \tfrac 32 (M+N_0) \omega\right) - \frac {4\pi} 3 t^2 s^3\,.\label{eq:aprioricondensate42}
 \end{align}
After optimizing over $s$ and $t$, we thus obtain the lower bound 
\begin{equation}\label{eq:aprioricondensate43}
 \frac 1{14\pi} \left( \frac 37\right)^{5/2} \omega^3 \frac{ |M-N_0|^{7/2}}{\left( E^\mathrm{GP}(M,a_N,\omega) + E^\mathrm{GP}(N_0,a_N,\omega) +  \tfrac 32 (M+N_0) \omega \right)^{3/2}}\,.
\end{equation}
In particular, with \eqref{eq:aprioricondensate4} we conclude that 
\begin{equation}
E^{\mathrm{GP}}(M,a_N,\omega) - E^{\mathrm{GP}}(N_0,a_N,\omega) - \mu^{\mathrm{GP}}(M - N_0) 
\gtrsim  \omega^{3/2} a_N  \frac{ |M-N_0|^{7/2} }{(M+N_0)^{3/2}} \,.  \label{eq:aprioricondensate4a}
\end{equation}

Putting all this together, we obtain
\begin{align}
o(N) & \geq  N_0 \sum_{0 \leq M \leq \delta N_0} \tr\left[ P_M \Gamma_{N,j_1} P_M \right]   \label{eq:aprioricondensate3} \\
& \quad +  \sum_{\delta N_0 < M \leq N} \tr\left[ P_M \Gamma_{N,j_1} P_M \right] \left( \omega^{1/2} a_N  \frac{ |M-N_0|^{7/2} }{(M+N_0)^{3/2}}  + \frac{1}M  \left\Vert \gamma_{P_M \Gamma_{N,j_1} P_M} - M P^{\mathrm{GP}}_{Ma_N} \right\Vert_1^2 \right). \nonumber
\end{align}
Next, we write $j_1 \gamma_N j_1  = \sum_{M=0}^{N} \tr \left[ P_M \Gamma_{N,j_1} P_M \right] \gamma_{P_M \Gamma_{N,j_1} P_M}$ and estimate the trace norm difference of $j_1 \gamma_N j_1$ and $N_0 P^{\mathrm{GP}}_{N_0 a_N}$ in a first step by
\begin{align}
\left\Vert j_1 \gamma_N j_1 - N_0 P^{\mathrm{GP}}_{N_0 a_N} \right\Vert_1 &\leq N_0(1+ \delta) \sum_{0 \leq M \leq \delta N_0} \tr\left[ P_M \Gamma_{N,j_1} P_M \right]\label{eq:aprioricondensate5} \\
&\quad + \sum_{\delta N_0 < M \leq N} \tr\left[ P_M \Gamma_{N,j_1} P_M \right] \left\Vert  \gamma_{P_M \Gamma_{N,j_1} P_M} - M P^{\mathrm{GP}}_{M a_N} \right\Vert_1 \nonumber \\
&\quad + \left\Vert \sum_{\delta N_0 < M \leq N} \tr\left[ P_M \Gamma_{N,j_1} P_M \right] \left( M P^{\mathrm{GP}}_{M a_N} - N_0 P^{\mathrm{GP}}_{N_0 a_N} \right) \right\Vert_1. \nonumber
\end{align}
Eq.~\eqref{eq:aprioricondensate3} tells us that the first two terms on the right-hand side of the above equation are of order $o(N)$. For the term in the last line, we insert $N_0 P^{\mathrm{GP}}_{M a_N} - N_0 P^{\mathrm{GP}}_{M a_N}$ in the obvious place to see that it is bounded from above by 
\begin{align}
&\left\Vert \sum_{\delta N_0 < M \leq N} \tr\left[ P_M \Gamma_{N,j_1} P_M \right] \left( M P^{\mathrm{GP}}_{M a_N} - N_0 P^{\mathrm{GP}}_{N_0 a_N} \right) \right\Vert_1 \label{eq:aprioricondensate6} \\
&\hspace{4cm} \leq \sum_{\delta N_0 < M \leq N} \tr\left[ P_M \Gamma_{N,j_1} P_M \right] \left\lbrace \vert M- N_0 \vert +  N_0 \left\Vert P^{\mathrm{GP}}_{M a_N} - P^{\mathrm{GP}}_{N_0 a_N} \right\Vert_1 \right\rbrace. \nonumber
\end{align} 
To bound the right-hand side of Eq.~\eqref{eq:aprioricondensate6}, we choose $0< \kappa < 1$ and split the sum into two parts, one where $| M-N_0 | \leq \kappa N_0$ and another one where $| M - N_0 | > \kappa N_0$. We claim that there exists a function $f:\mathbb{R}_+ \to \mathbb{R}_+$ with $f(x) \to 0$ for $x \to 0$ such that the first part of the sum is bounded by $(\kappa + f(\kappa)) N_0$. This follows from the continuity of  the map $M \mapsto \phi^{\mathrm{GP}}_{1,M a_N}$ in $L^2(\mathbb{R}^3)$, which can easily be deduced from the uniqueness of the minimizer of the GP functional. To estimate the contribution to the sum of the terms with $| M - N_0 | > \kappa N_0$, we write
\begin{align}\nonumber
&\sum_{ \substack{\delta N_0 < M \leq N : \\ |M-N_0| > \kappa N_0}}  \tr\left[ P_M \Gamma_{N,j_1} P_M \right]  \left\lbrace \vert M- N_0 \vert +  N_0 \left\Vert P^{\mathrm{GP}}_{M a_N} - P^{\mathrm{GP}}_{N_0 a_N} \right\Vert_1 \right\rbrace  \\
&\hspace{1cm}\leq \sum_{ \substack{\delta N_0 < M \leq N : \\ |M-N_0| > \kappa N_0}}  \tr\left[ P_M \Gamma_{N,j_1} P_M \right] \left\{ \frac{ \vert M- N_0 \vert^{7/2} }{\kappa^{5/2} N_0^{5/2}} + 2 \frac{ \vert M- N_0 \vert^{7/2} }{\kappa^{7/2} N_0^{5/2}} \right\}\,. \label{eq:aprioricondensate6c}
\end{align}
Together with Eq.~\eqref{eq:aprioricondensate3}, this implies for $\kappa \lesssim 1$ that
\begin{equation}
\sum_{ \substack{\delta N_0 < M \leq N : \\ |M-N_0| > \kappa N_0}}  \tr\left[ P_M \Gamma_{N,j_1} P_M \right]  \left\lbrace \vert M- N_0 \vert +  N_0 \left\Vert P^{\mathrm{GP}}_{M a_N} - P^{\mathrm{GP}}_{N_0 a_N} \right\Vert \right\rbrace \label{eq:aprioricondensate6d} \\
 \leq \frac{o(N)}{\omega^{1/2} \kappa^{7/2} a_N N_0 } \,.
\end{equation}
Choosing $\kappa \ll 1$ and $\delta$ appropriately, we see that the right-hand side of Eq.~\eqref{eq:aprioricondensate6d} as well as the part of the sum in Eq.~\eqref{eq:aprioricondensate6} where $\vert M - N_0 \vert \leq \kappa N_0$ is of the order $o(N)$. Together with Eq.~\eqref{eq:aprioricondensate5}, this proves
\begin{equation}
\left\Vert j_1 \gamma_N j_1 - N_0 P^{\mathrm{GP}}_{N_0 a_N} \right\Vert_1 \leq o(N).
\label{eq:aprioricondensate6e}
\end{equation}

\subsection{The bound for $j_2 \gamma_N j_2$}
\label{sec:aprioriboundj2gammaNj2}
The main ingredient to derive a bound for $j_2 \gamma_N j_2$ is a {\em novel coercivity estimate} for the bosonic relative entropy that we prove in Lemma~\ref{lem:relativeentropy} below. Using this estimate, we shall show that $j_2 \gamma_N j_2$ is close to $\gamma_0^{\mathrm{gc}} - N_0^\mathrm{gc} \vert \varphi_0 \rangle\langle \varphi_0 \vert$ in trace norm, where 
\begin{equation}
\gamma_0^{\mathrm{gc}} = \frac{1}{e^{\beta (h-\mu_0)}-1}
\label{eq:aprioricondensate6f}
\end{equation}
denotes the grand canonical analogue of $\gamma_{N,0}$. This part of the proof is motivated by a related analysis for the one-particle density matrix of a dilute Fermi gas in \cite{RobertFermigas}. 

For positive trace-class operators $\gamma$ define 
\begin{equation}
s(\gamma)  = - \tr \sigma(\gamma), \quad \ \text{with} \quad \sigma(x) = x \ln(x) - (1+x) \ln(1+x).
\label{eq:densitymatrix5b}
\end{equation}
 We have \cite[2.5.14.5]{Thirring_4} 
\begin{equation}
S(\Gamma_{j_2}) \leq s(j_2 \gamma_N j_2)\,.
\label{eq:densitymatrix3a}
\end{equation}
Since 
\begin{equation}
\tr\left[ \left( \text{d}\Upsilon\left( h_{\geq R}^{\mathrm{D}} \right) - \mu^{\mathrm{GP}} \hat{N} \right) \Gamma_{N,j_2} \right]  = \tr\left[ \left( h - \mu^{\mathrm{GP}}\right) j_2 \gamma_N j_2 \right] 
\end{equation}
we conclude that
\begin{align}
&\tr\left[ \left( \text{d}\Upsilon\left( h_{\geq R}^{\mathrm{D}} \right) - \mu^{\mathrm{GP}} \hat{N} \right) \Gamma_{N,j_2} \right] - T S\left( \Gamma_{N,j_2} \right) \label{eq:densitymatrix3b} \\
&\hspace{5cm} \geq \tr\left[ \left( h - \mu^{\mathrm{GP}}\right) j_2 \gamma_N j_2 \right] - T s\left( j_2 \gamma_N j_2 \right) \nonumber
\end{align}
holds. Let us define 
\begin{equation}
\nu(x)=\max\lbrace x, \mu \rbrace \quad \text{ with } \quad c \mu^{\mathrm{GP}} < \mu \lesssim \omega \text{ for some } c>1.
\label{eq:densitymatrixdefofg}
\end{equation}
For what follows, it will be convenient to replace the Hamiltonian $h$ by $\nu(h)$ on the right-hand side of Eq.~\eqref{eq:densitymatrix3b}. We  write $h=\sum_{\alpha=0}^{\infty} e_{\alpha}(h) \vert \varphi_{\alpha} \rangle\langle \varphi_{\alpha} \vert$ and choose $\alpha_0$ to be the largest integer such that $e_{\alpha_0}(h) < \mu$. Using  $\alpha_0 = O(1)$ and $\mu \lesssim \omega$ we can estimate
\begin{equation}
\tr[ (\nu(h) - h) j_2 \gamma_N j_2 ] \leq \mu \sum_{\alpha=0}^{\alpha_0} \langle \varphi_{\alpha},  j_2 \gamma_N j_2 \varphi_{\alpha} \rangle \lesssim \omega N \sum_{\alpha=0}^{\alpha_0} \left\Vert j_2 \varphi_{\alpha} \right\Vert^2 \lesssim \omega N e^{-C \omega R^2}\,,
\label{eq:densitymatrix3c}
\end{equation}
where $C>0$ is some appropriately chosen constant. To obtain the last inequality on the right-hand side of Eq.~\eqref{eq:densitymatrix3c}, we used the decay of the eigenfunctions of $h$ and the fact that the support of $j_2$ is given by $B(R)^{\mathrm{c}}$. Together with Eq.~\eqref{eq:densitymatrix2b}, Eq.~\eqref{eq:densitymatrix3b} and $\omega^{1/2} R \gg 1$, this implies
\begin{align}
o(\omega N)  &\geq  \tr\left[ \left( \nu( h ) - \mu^{\mathrm{GP}}\right) j_2 \gamma_N j_2 \right] - T s\left( j_2 \gamma_N j_2 \right) \label{eq:densitymatrix3d} \\
& \quad - \tfrac{1}{\beta}\tr\left[ \ln\left( 1-e^{-\beta(h-\mu_0)} \right) \right]  + \left( \mu^{\mathrm{GP}} - \mu_0 \right) N_{\mathrm{th}}. \nonumber
\end{align}
To be able to  compare the expressions in the first and in the second line on the right-hand side of Eq.~\eqref{eq:densitymatrix3d}, we will replace $h$ by $\nu(h)$ and afterwards $\mu_0$ by $\mu^{\mathrm{GP}}$ in the first term in the second line. In fact, since $\nu(h)\geq h$, 
\begin{equation}
 \tr \left[ \ln\left( 1 - e^{-\beta \left( h - \mu_0 \right)} \right)  \right] \leq \tr \left[  \ln\left( 1 - e^{-\beta \left( \nu(h) - \mu_0 \right)} \right)  \right] \label{eq:densitymatrix3e_o}
\end{equation}
Moreover, using that $\beta \mu_0 = O(N_0^{-1})$ and thus $\vert \mu_0 - \mu^{\mathrm{GP}} \vert \lesssim \omega$, we can proceed similarly to Eq.~\eqref{eq:lowerbound12} to obtain
\begin{equation}
 \frac 1 \beta \tr \left[ \ln\left( 1 - e^{-\beta \left( \nu(h) - \mu_0 \right)} \right) - \ln\left( 1 - e^{-\beta \left( \nu( h) - \mu^\mathrm{GP} \right)} \right) \right] \leq \left( \mu_0 - \mu^{\mathrm{GP}} \right) N_{\mathrm{th}}^{\mathrm{gc}} + O\left( \frac{1}{\beta^2 \omega} \right).
\label{eq:densitymatrix3e}
\end{equation}
Corollary~\ref{lem:particlenumbers} tells us that $\vert N_{\mathrm{th}} - N_{\mathrm{th}}^{\mathrm{gc}} \vert \lesssim (\beta \omega)^{-3/2} (\ln N)^{1/2} + (\beta \omega)^{-1} \ln N$. Together with Eqs.~\eqref{eq:densitymatrix3d}--\eqref{eq:densitymatrix3e}, this shows that
\begin{equation}
o(\omega N) \geq \tr[ (\nu(h)-\mu^{\mathrm{GP}}) j_2 \gamma_N j_2 ] - Ts(j_2 \gamma_N j_2) - \tfrac{1}{\beta} \tr\left[ \ln\left( 1-e^{-\beta( \nu(h)-\mu^{\mathrm{GP}})} \right) \right] \label{eq:densitymatrix5}
\end{equation}
holds.

The right-hand side of Eq.~\eqref{eq:densitymatrix5} can be written in terms of the relative entropy, which is defined as follows. 
For two nonnegative operators $\gamma,\gamma_0$ with finite trace, the bosonic relative entropy of $\gamma$ with respect to $\gamma_0$ is given by
\begin{equation}
\mathcal{S}(\gamma,\gamma_0) = \tr \left( \sigma(\gamma) - \sigma(\gamma_0) - \sigma'(\gamma_0)(\gamma-\gamma_0)  \right).
\label{eq:relativeentropy1}
\end{equation}
with $\sigma$ defined in Eq.~\eqref{eq:densitymatrix5b}. 
For matrices it is well-defined as long as $\gamma_0>0$. In case $\gamma_0$ has a nontrivial kernel and $\gamma - \gamma_0 \neq 0$ on $\text{ker}(\gamma_0)$, one defines $\mathcal{S}(\gamma,\gamma_0) = \infty$. If $\gamma - \gamma_0 = 0$ on $\text{ker}(\gamma_0)$ the trace is by definition taken on the complement of that subspace. In the case  of trace-class operators with $\gamma_0$ strictly positive, one can equivalently define
\begin{equation}\label{def:cS}
\mathcal{S}(\gamma,\gamma_0)  = \sum_{i,j}  \left| \left\langle \psi_i | \varphi_j \right\rangle \right|^2 S(\gamma_i,\nu_j)
\end{equation}
where $S(x,y) = \sigma(x) - \sigma(y) - \sigma'(y)(x-y) \geq 0$, and $\{ \lambda_i, \psi_i\}$ respectively $\{\nu_j,\varphi_j\}$ are the eigenvalues and eigenfunctions of $\gamma$ and $\gamma_0$, respectively. The definition \eqref{def:cS} will be most convenient for our purpose. We note that $\mathcal{S}$ can be defined more generally even for non-compact operators by  approximating the operators by matrices and  taking limits, see  \cite{LewinSabin2014, DHS2015}. 

Denote by
\begin{equation}
\gamma_{\nu,0} = \left(e^{\beta\left(\nu(h)-\mu^{\mathrm{GP}} \right)}-1 \right)^{-1} 
\label{eq:densitymatrix3f}
\end{equation}
the one-particle density matrix of the Gibbs state related to the grand canonical potential $\tfrac{1}{\beta} \tr[ \ln( 1-e^{-\beta( \nu(h)-\mu^{\mathrm{GP}})} ) ]$. A simple computation shows that Eq.~\eqref{eq:densitymatrix5} can equivalently be written as
\begin{equation}
o(\omega N) \geq \tfrac{1}{\beta} \mathcal{S}\left( j_2 \gamma_N j_2,  \gamma_{\nu,0} \right).
\label{eq:densitymatrix3g}
\end{equation}
In order to get quantitative information out of Eq.~\eqref{eq:densitymatrix3g}, we need the following Lemma:
\begin{lemma}
\label{lem:relativeentropy}
There exists a constant $C>0$ such that for any two nonnegative trace-class operators $\gamma,\gamma_0$ we have
\begin{equation}
\mathcal{S}(\gamma,\gamma_0) \geq C \tr\left( \frac{1}{1+\gamma_0} \left( \frac{\gamma}{\sqrt{1+\gamma}} - \frac{\gamma_0}{\sqrt{1+\gamma_0}} \right)^2  \right)
\label{eq:relativeentropy2}
\end{equation}
and
\begin{equation}
\mathcal{S}(\gamma,\gamma_0) \geq C \frac{\left[ \tr\left(\gamma - \gamma_0 \right) \right]^2}{\tr\left( \left( \gamma+\gamma_0 \right)\left(1+\gamma_0 \right) \right)}.
\label{eq:relativeentropy3}
\end{equation}
\end{lemma}
\begin{remark}
With Eqs.~\eqref{eq:relativeentropy2} and~\eqref{eq:relativeentropy3} one can show that for fixed $\gamma_0$ convergence of the relative entropy to zero implies convergence of $\gamma$ to $\gamma_0$ in trace norm. A quantitative estimate on their trace norm distance will in fact be given below, see Eq.~\eqref{eq:densitymatrix9} et seq.  
\end{remark}
\begin{proof}
We start with the proof of the second inequality. With $S(x,y) = \sigma(x) - \sigma(y) - \sigma'(y)(x-y)$  we wish to show that
\begin{equation}
S(x,y) \geq C \frac{(x-y)^2}{(x+y)(1+y)}
\label{eq:relativeentropy4}
\end{equation}
holds for all numbers $x,y \in \mathbb{R}_+$. To that end, we will consider several cases and start with the one where $x \geq y$ and $y \geq 1$. We claim that
\begin{equation}
S(x,y) = \int_y^x \frac{(x-s)}{s(s+1)} \text{d}s \geq \frac{1}{2} \left( -1 + \frac{x}{y} - \ln\left( \frac{x}{y} \right) \right)
\label{eq:relativeentropy5}
\end{equation}
which follows from $s(s+1) \leq 2 s^{2}$ for $s \geq 1$. One also checks that $z \geq 1$ implies $-1+z-\ln(z) \geq (8/9) (z-1)^2/(z+1)$. Together with Eq.~\eqref{eq:relativeentropy5}, this gives
\begin{equation}
S(x,y) \geq \frac{4}{9} \frac{(x-y)^2}{(x+y)(1+y)}.
\label{eq:relativeentropy6}
\end{equation}
Next, consider the case $x \geq 2$ and $y \leq 1$. Here
\begin{equation}
\int_y^x \frac{(x-s)}{s(s+1)} \text{d}s \geq \frac{1}{2} \int_{1}^{x} \frac{(x-s)}{s^2} \text{d}s = \frac{1}{2} \left( -1 + x - \ln(x) \right).
\label{eq:relativeentropy7}
\end{equation}
We use the same inequality as above to obtain a lower bound for the right-hand of Eq.~\eqref{eq:relativeentropy7} and find
\begin{equation}
S(x,y) \geq \frac{4}{9} \frac{(x-1)^2}{1+x} \geq \frac{1}{9} \frac{(x-y)^2}{1+x} \geq \frac{2}{27} \frac{(x-y)^2}{(x+y)(1+y)}.  
\label{eq:relativeentropy8}
\end{equation}

For $x \geq y$ and $x \leq 2$, we again consider the integral representation of $S(x,y)$ from above and replace $s$ by $x+y$ in the denominator of the function under the integral sign. This yields a lower bound of the form
\begin{equation}
S(x,y) \geq \frac{1}{2} \frac{(x-y)^2}{(x+y)(1+x+y)} \geq \frac{1}{6} \frac{(x-y)^2}{(x+y)(1+y)}.
\label{eq:relativeentropy9}
\end{equation}
It remains to consider the case $y \geq x$. Here we argue in the same way as in the previous step:
\begin{equation}
S(x,y) \geq \frac{1}{2} \frac{(x-y)^2}{(x+y)(1+x+y)} \geq \frac{1}{4} \frac{(x-y)^2}{(x+y)(1+y)}
\label{eq:relativeentropy10}
\end{equation}
This proves the claimed bound \eqref{eq:relativeentropy4} for $S(x,y)$ with $C\geq 2/27$. 

In order to deduce Eq.~\eqref{eq:relativeentropy3} from Eq.~\eqref{eq:relativeentropy4}, we follow the argument in the proof of \cite[Lemma~7]{RobertFermigas}. For $(x,y) \in \mathbb{R}^2$ with $y > 0$, the map $(x,y) \mapsto \tfrac{x^2}{y}$ is jointly convex. In fact, $\tfrac{x^2}{y} = \sup_{\lambda \in \mathbb{R}} \left\lbrace 2 \lambda x - \lambda^2 y \right\rbrace$. Using this representation and Klein's inequality (see e.g.~\cite[2.1.4(5)]{Thirring_4}), which simply amounts to plugging the lower bound for $S$ into \eqref{def:cS}, we find that
\begin{equation}
\mathcal{S}\left( \gamma, \gamma_0 \right) \geq \frac{4}{27} \lambda  \tr \left[ \gamma - \gamma_0 \right] - \frac{2}{27} \lambda^2 \tr \left[ (\gamma - \gamma_0)(1+\gamma_0) \right]
\label{eq:relativeentropy11}
\end{equation} 
which holds for all $\lambda \in \mathbb{R}$. By optimizing over $\lambda$ we obtain Eq.~\eqref{eq:relativeentropy3}.

It remains to show that Eq.~\eqref{eq:relativeentropy2} holds. Let $f(x) = x/\sqrt{1+x}$ and note that $f'(x) \leq x^{-1/2}$ implies
\begin{equation}
\left( f(x) - f(y) \right)^2 \leq 4 \left( \sqrt{x} - \sqrt{y} \right)^2.
\label{eq:relativeentropy12}
\end{equation}
Since $\left( \sqrt{z} - 1 \right)^2 \leq \tfrac{(z-1)^2}{z+1}$ for $z \geq 0$ we conclude that
\begin{equation}
\left( \frac{x}{\sqrt{1+x}} - \frac{y}{\sqrt{1+y}} \right)^2 \leq 4 \frac{(x-y)^2}{x+y}.
\label{eq:relativeentropy13}
\end{equation}
Eq.~\eqref{eq:relativeentropy13} together with Eq.~\eqref{eq:relativeentropy4} and Klein's inequality prove the claim. 
\end{proof}
As in the above Lemma, let $\gamma$ and $\gamma_0$ be nonnegative trace-class operators and choose $\lambda > 0$. We have
\begin{equation}
\left\Vert \gamma - \gamma_0 \right\Vert_1 \leq \tr\left[ \mathds{1}(\gamma > \lambda) \gamma \right] + \left\Vert \mathds{1}(\gamma \leq \lambda) \gamma - \gamma_0 \right\Vert_1. 
\label{eq:densitymatrix6}
\end{equation}
Let $P$ be some orthogonal projection, $Q = 1-P$ and denote 
\beq
\tilde{\gamma}=\mathds{1}(\gamma \leq \lambda) \gamma \,.
\eeq
The following argumentation follows closely the related argument in \cite[Eqs.~(4.33)--(4.35)]{RobertFermigas}. We estimate
\begin{align}
\left\Vert \tilde{\gamma} - \gamma_0 \right\Vert_1 &\leq \left\Vert \left( \tilde{\gamma} - \gamma_0 \right) P \right\Vert_1 + \left\Vert \tilde{\gamma} Q \right\Vert_1 + \left\Vert \gamma_0 Q \right\Vert_1 \label{eq:densitymatrix7} \\
&\leq \left\Vert \left( \tilde{\gamma} - \gamma_0 \right) P \right\Vert_1 + \left\Vert \tilde{\gamma} \right\Vert_1^{1/2} \left\Vert Q \tilde{\gamma} Q \right\Vert_1^{1/2} + \left\Vert \gamma_0 \right\Vert_1^{1/2} \left\Vert Q \gamma_0 Q \right\Vert_1^{1/2}. \nonumber
\end{align}
For the trace norm of $\tilde{\gamma}$, we have $\left\Vert \tilde{\gamma} \right\Vert_1 \leq \left\Vert \gamma_0 \right\Vert_1 + \left| \tr\left[ \tilde{\gamma} - \gamma_0 \right] \right|$. Together with
\begin{align}
\left\Vert Q \tilde{\gamma} Q \right\Vert_1 &= \tr\left( \tilde{\gamma} Q \right) = \tr\left[ \gamma_0 Q + (\tilde{\gamma}-\gamma_0) - (\tilde{\gamma} - \gamma_0) P \right] \label{eq:densitymatrix8} \\
&\leq \left\Vert Q \gamma_0 Q \right\Vert_1 + \left| \tr[\tilde{\gamma} - \gamma_0 ] \right| + \left\Vert (\tilde{\gamma}-\gamma_0) P \right\Vert_1, \nonumber
\end{align}
this implies
\begin{align}
\left\Vert \gamma - \gamma_0 \right\Vert_1 &\leq \tr\left[ \mathds{1}(\gamma > \lambda) \gamma \right]  + \left\Vert \left( \tilde{\gamma} - \gamma_0 \right) P \right\Vert_1  \label{eq:densitymatrix9} \\
& \quad + 2 \left( \left\Vert \gamma_0 \right\Vert_1 + \left| \tr\left[ \tilde{\gamma} - \gamma_0 \right] \right| \right)^{1/2} \left( \left\Vert Q \gamma_0 Q \right\Vert_1 + \left| \tr[\tilde{\gamma} - \gamma_0 ] \right| + \left\Vert (\tilde{\gamma}-\gamma_0) P \right\Vert_1 \right)^{1/2}. \nonumber 
\end{align}
We shall apply this to   $\gamma = j_2 \gamma_N j_2$ and $\gamma_0=\gamma_{\nu,0}$ in Eq.~\eqref{eq:densitymatrix3f}, with $\lambda = 4 \Vert \gamma_{\nu,0} \Vert$. Let us note that $\Vert \gamma_{\nu,0} \Vert = O((\beta \omega)^{-1})$ which follows from the explicit form of the eigenvalues of $h$ and the definition of $\nu$, see Eq.~\eqref{eq:densitymatrixdefofg}. We also define 
\beq 
f(x) = \frac x {\sqrt{1+x}} \,. 
\eeq
To keep the notation simple, we will still write $\gamma$ instead of $j_2 \gamma_N j_2$ in the following discussion.

Let us start with the first term on the right-hand side of $\eqref{eq:densitymatrix9}$. We want to show that it is of the order $o(N)$. To do so, we bound
\begin{equation}\label{eq:densitymatrix10}
\tr[ \gamma \mathds{1}(\gamma > \lambda) ] \leq \frac{1+\lambda}{\lambda} \tr\left[ f(\gamma)^2 \mathds{1}(\gamma > \lambda) \right]  =\frac{1+\lambda}{\lambda}   \sum_{e_i(\gamma)>\lambda}   f(e_i(\gamma))^2
\end{equation}
where $e_i(\gamma)$ are the eigenvalues of $\gamma$. Since $f(t)\geq 2 f(t/4)$, we  have $f(e_i(\gamma)) \leq 2( f(e_i(\gamma)) - f(\lambda/4))$ for $e_i(\gamma)>\lambda$. Since $f(\gamma_{\nu,0}) \leq f(\lambda/4)$,  we further have  (denoting by $\psi_i$  the eigenfunctions of $\gamma$)
 \begin{align} 
\tr[ \gamma \mathds{1}(\gamma > \lambda) ] &\leq  4 \frac{1+\lambda}{\lambda}  \sum_{\gamma_i>\lambda}   \langle \psi_i | (f(\gamma) - f(\gamma_{\nu,0})) | \psi_i \rangle^2  \label{eq:densitymatrix101} \\ 
& \leq   4 \frac{1+\lambda}{\lambda}  \sum_{\gamma_i>\lambda}   \langle \psi_i | (f(\gamma) - f(\gamma_{\nu,0}))^2 | \psi_i \rangle  = 4 \frac{1+\lambda}{\lambda} \tr\left[ \mathds{1}(\gamma > \lambda) \left( f(\gamma) - f(\gamma_{\nu,0}) \right)^2 \right]. \nonumber
\end{align}
We know from Lemma~\ref{lem:relativeentropy} and Eq.~\eqref{eq:densitymatrix3g} that
\begin{equation}
\tr\left[ \frac{1}{1+ \gamma_{\nu,0} } \left( f(\gamma) - f(\gamma_{\nu,0}) \right)^2 \right] \leq o(\beta \omega N). \label{eq:densitymatrix11}
\end{equation}
Together with $\mathds{1}(\gamma > \lambda) \leq 1$ and $\Vert \gamma_{\nu,0} \Vert \lesssim (\beta \omega)^{-1}$, Eqs.~\eqref{eq:densitymatrix10}--\eqref{eq:densitymatrix11} imply that
\begin{equation}
\tr[ \gamma \mathds{1}(\gamma > \lambda) ] \leq o(N) \,. \label{eq:densitymatrix12}
\end{equation}

Next, consider the second term on the right-hand side of Eq.~\eqref{eq:densitymatrix9}. Using the Cauchy-Schwarz inequality and the cyclicity of the trace, we find
\begin{align}
\left\Vert ( \tilde{\gamma} - \gamma_{\nu,0} ) P \right\Vert_1 \leq \left( \tr\left[ \frac{1}{1+\gamma_{\nu,0}}(\tilde{\gamma} - \gamma_{\nu,0})^2 \right] \right)^{1/2} \left( \tr\left[ (1+\gamma_{\nu,0}) P \right] \right)^{1/2}. \label{eq:densitymatrix13}
\end{align}
We write $f(x) - f(y) =(x-y) \int_0^1 f'(y+t(x-y)) \text{d}t $ and use the lower bound $f'(x) \geq \frac{1}{2} (1+x)^{-1/2}$ to show that
\begin{equation}
\left( f(x) - f(y) \right)^2 \geq \frac{1}{4 (1+\max\lbrace x,y\rbrace)} (x-y)^2 \geq  \frac{1}{4 (1+\lambda)} (x-y)^2
\label{eq:densitymatrix13b}
\end{equation} 
for $0\leq x,y\leq \lambda$. The fact that $\max \lbrace \| \gamma_{\nu,0}\| , \| \tilde{\gamma}\| \rbrace \leq \lambda$ and an application of Klein's inequality therefore gives
\begin{equation}
\tr\left[ \frac{1}{1+\gamma_{\nu,0}}(\tilde{\gamma} - \gamma_{\nu,0})^2 \right] \leq 4 \tr \left[ \frac{1+\lambda}{1+\gamma_{\nu,0}} \left( f(\tilde{\gamma}) - f (\gamma_{\nu,0}) \right)^2 \right]. \label{eq:densitymatrix14}
\end{equation}
Denote by $\lbrace \psi_{j} \rbrace_{j=0}^{\infty}$ the eigenbasis of $\gamma$ and by $\lbrace \varphi_{i} \rbrace_{i=0}^{\infty}$ the eigenbasis of $\gamma_{\nu,0}$. In order to replace $\tilde{\gamma}$ by $\gamma$, we write
\begin{align}\nonumber
& \tr \left[ \frac{1}{1+\gamma_{\nu,0}} \left( f(\tilde{\gamma}) - f (\gamma_{\nu,0}) \right)^2 \right]  \\
&= \sum_{i,j = 0}^{\infty} \frac{1}{1+e_{i}(\gamma_{\nu,0})} \left\vert \left\langle \psi_{i} \vert  \varphi_{j} \right\rangle \right\vert^2 \left( f(\mathds{1}(e_{j}(\gamma) \leq \lambda) e_{j}(\gamma)) - f(e_{i}(\gamma_{\nu,0})) \right)^2. \label{eq:densitymatrix15}
\end{align}
Since $\lambda = 4 \Vert \gamma_{\nu,0} \Vert$, we have $f(\lambda) \geq 2 f(\|\gamma_{\nu,0}\|)$. It follows that $f(e_j(\gamma)) \geq 2 f(e_i(\gamma_{\nu,0}))$ for all $i$ and $j$ such that $e_j(\gamma)> \lambda$. Hence we can replace $f(\mathds{1}(e_{j}(\gamma) \leq \lambda) e_{j}(\gamma))$ by $f(e_{j}(\gamma))$ in Eq.~\eqref{eq:densitymatrix15} to obtain an upper bound.  Combining this upper bound with Eqs.~\eqref{eq:densitymatrix14},~\eqref{eq:densitymatrix13} and~\eqref{eq:densitymatrix11}, we find that 
\begin{equation}
\left\Vert ( \tilde{\gamma} - \gamma_{\nu,0} ) P \right\Vert_1 \leq o(N^{1/2}) \left( \tr\left[ (1+\gamma_{\nu,0}) P \right] \right)^{1/2} \,.\label{eq:densitymatrix16}
\end{equation}

Eq.~\eqref{eq:densitymatrix3g}, $\|\gamma_{\nu,0}\|\lesssim (\beta\omega)^{-1}$ and an application of Lemma~\ref{lem:relativeentropy} prove the bound $\left| \tr [ \gamma - \gamma_{\nu,0} ] \right| \leq o(N)$. Together with Eq.~\eqref{eq:densitymatrix12}, this shows
\begin{equation}
\vert \tr [ \mathds{1}(\gamma \leq \lambda) \gamma - \gamma_{\nu,0} ] \vert \leq \vert \tr [ \gamma - \gamma_{\nu,0} ] \vert + \tr[ \mathds{1}(\gamma > \lambda) \gamma ] \leq o(N). \label{eq:densitymatrix17}
\end{equation}
Having Eq.~\eqref{eq:densitymatrix17} and Eq.~\eqref{eq:densitymatrix16} at hand, we combine them with Eq.~\eqref{eq:densitymatrix9} and $\Vert \gamma_{\nu,0} \Vert_1 \leq N$ to finally obtain
\begin{align}
\left\Vert j_2 \gamma_N j_2 - \gamma_{\nu,0} \right\Vert_1 &\leq  o(N) + o\left(N^{1/2} \right) \left( \tr\left[ (1+\gamma_{\nu,0}) P \right] \right)^{1/2} \label{eq:densitymatrix18} \\
&\quad + O(N^{1/2}) \left( \left\Vert Q \gamma_{\nu,0} Q \right\Vert_1 + o(N^{1/2}) \left( \tr\left[ (1+\gamma_{\nu,0}) P \right] \right)^{1/2} \right)^{1/2} \,,\nonumber
\end{align}
where we inserted $j_2 \gamma_N j_2$ for  $\gamma$. To complete the argument it remains to choose the projection $P$.

We choose $P = \mathds{1}(h \leq \eta T)$ for some large $\eta > 0$. Recall that $g(n)=(n+1)(n+2)/2$ denotes the degree of degeneracy of the energy level $\omega n$ of the harmonic oscillator Hamiltonian $h$.  We then have 
\begin{equation}
\tr [(1+\gamma_{\nu,0})P] \leq \tr \left[P + \gamma_{\nu,0} \right] = \sum_{\substack{ n \geq 0 : \\ \omega n \leq \eta T }} g(n) + O\left(\frac{1}{(\beta \omega)^{3}}\right) \lesssim \frac {1+\eta^3}{(\beta\omega)^3} \,.
\label{eq:densitymatrix18b}
\end{equation}
The term involving $Q$ can be estimated as 
\begin{equation}
\left\Vert Q \gamma_{\nu,0} Q \right\Vert_1 \leq \tr\left[ \gamma_{0} Q \right] = \sum_{ \substack{n \geq 0 : \\ \omega n > \eta T }} \frac{g(n)}{e^{\beta \left( \omega n - \mu_0 \right)}-1} \lesssim \frac{e^{-\eta/2}}{(\beta \omega)^3}.
\label{eq:densitymatrix18c}
\end{equation}
By choosing $\eta \gg 1$ appropriately, this shows 
\begin{equation}
\left\Vert j_2 \gamma_N j_2 - \gamma_{\nu,0} \right\Vert_1 \leq o(N). 
\label{eq:densitymatrix18c2}
\end{equation}

As a final step in the estimate of $j_2 \gamma_N j_2$, we replace $\gamma_{\nu,0}$ by $\gamma_0^{\mathrm{gc}} - N_0^\mathrm{gc} \vert \varphi_0 \rangle\langle \varphi_0 \vert$ with $\gamma_0^{\mathrm{gc}}$ defined in Eq.~\eqref{eq:aprioricondensate6f}. A straightforward computation shows that $\Vert \gamma_0^{\mathrm{gc}} - N_0^\mathrm{gc} \vert \varphi_0 \rangle\langle \varphi_0 \vert -\gamma_{\nu,0} \Vert_1 \lesssim (\beta \omega)^{-1}$. Hence,
\begin{equation}
\left\Vert j_2 \gamma_N j_2 - \left( \gamma_0^{\mathrm{gc}} - N_0^\mathrm{gc} \vert \varphi_0 \rangle\langle \varphi_0 \vert \right) \right\Vert_1 \leq o(N)
\label{eq:densitymatrix18c3}
\end{equation}
holds.
\subsection{The off-diagonal elements of $\gamma_N$ and the final estimate}
To complete the proof of Theorem~\ref{thm:main}, it remains to estimate the trace norm of $j_1^2 \gamma_N j_2^2$. In combination with Corollary~\ref{lem:particlenumbers}, which shows that the trace norm difference of $\gamma_{N,0}$ and $\gamma_0^{\mathrm{gc}}$ is small, this will allow us to conclude the convergence result \eqref{eq:mainresult6} for the one-particle density matrix  in the case $N_0 \sim N$ 
and $N a_N \geq \epsilon \omega^{-1/2}$ for some $\epsilon >0$. Finally, we shall comment on the case where one of these assumptions is not valid.

We define $\tilde{\gamma}_0^{\mathrm{gc}} = \gamma_0^{\mathrm{gc}} - N_0^\mathrm{gc} \vert \varphi_0 \rangle \langle \varphi_0 \vert$ as well as $P^{\mathrm{GP}} = P^{\mathrm{GP}}_{N_0a_N} = \vert \phi^{\mathrm{GP}}_{1,N_0 a_N} \rangle \langle \phi^{\mathrm{GP}}_{1,N_0 a_N} \vert$ for short. The identity $j_1(x)^2 + j_2(x)^2 = 1$ for all $x \in \mathbb{R}^3$ and the triangle inequality allow us to bound
\begin{align}
\left\Vert \gamma_N - \tilde{\gamma}^{\mathrm{gc}}_0 - N_0 P^{\mathrm{GP}} \right\Vert_1 \leq& \left\Vert j_1^2 \gamma_N j_1^2 - N_0 j_1 P^{\mathrm{GP}} j_1 \right\Vert_1 + \left\Vert  j_2^2 \gamma_N j_2^2 - j_2 \tilde{\gamma}^{\mathrm{gc}}_0 j_2 \right\Vert_1 \label{eq:densitymatrix19} \\
&+ 2 \left\Vert j_1^2 \gamma_N j_2^2 \right\Vert_1 + N_0 \left\Vert P^{\mathrm{GP}} - j_1 P^{\mathrm{GP}} j_1 \right\Vert_1 + \left\Vert \tilde{\gamma}^{\mathrm{gc}}_0 - j_2 \tilde{\gamma}^{\mathrm{gc}}_0 j_2 \right\Vert_1. \nonumber
\end{align}
With Eq.~\eqref{eq:aprioricondensate6e}, Eq.~\eqref{eq:densitymatrix18c3} and $j_i(x) \leq 1$ for $i\in\{1,2\}$, we see that the first two terms on the right-hand side of the above equation are of order $o(N)$. To derive an estimate for the first term in the second line of Eq.~\eqref{eq:densitymatrix19}, we bound $\left\Vert j_1^2 \gamma_N j_2^2 \right\Vert_1 \leq \left\Vert P^{\mathrm{GP}} j_1^2 \gamma_N j_2^2 \right\Vert_1 + \left\Vert (1-P^{\mathrm{GP}}) j_1^2 \gamma_N j_2^2 \right\Vert_1$. Since $P^{\mathrm{GP}}$ is a rank one projection, we have (recall that $\Vert \cdot \Vert$ denotes the operator norm)
\begin{equation}
\left\Vert P^{\mathrm{GP}} j_1^2 \gamma_N j_2^2 \right\Vert_1 = \left\Vert P^{\mathrm{GP}} j_1^2 \gamma_N j_2^2 \right\Vert \leq \left\Vert \gamma_N \right\Vert^{1/2} \left\Vert j_2^2 \gamma_N j_2^2 \right\Vert^{1/2} \leq N^{1/2} \left\Vert j_2 \gamma_N j_2 \right\Vert^{1/2}. \label{eq:densitymatrix20}
\end{equation}
We also estimate
\begin{equation}
\left\Vert j_2 \gamma_N j_2 \right\Vert \leq \left\Vert j_2 \gamma_N j_2 - \tilde{\gamma}^{\mathrm{gc}}_0 \right\Vert + \left\Vert \tilde{\gamma}_0^{\mathrm{gc}} \right\Vert. \label{eq:densitymatrix21}
\end{equation}
Recall that the largest eigenvalue of $\tilde{\gamma}_0^{\mathrm{gc}}$ is bounded by a constant times $(\beta \omega)^{-1}$. Making use of Eq.~\eqref{eq:densitymatrix18c3}, this implies $\left\Vert P^{\mathrm{GP}} j_1^2 \gamma_N j_2^2 \right\Vert_1 \leq o(N)$. On the other hand,
\begin{align}
\left\Vert (1-P^{\mathrm{GP}}) j_1^2 \gamma_N j_2^2 \right\Vert_1 &\leq \left\Vert (1-P^{\mathrm{GP}}) j_1^2 \gamma_N^{1/2} \right\Vert_2 \left\Vert \gamma_N^{1/2} j_2^2 \right\Vert_2 \label{eq:densitymatrix22} \\
&\leq \left( \tr\left[ \left( j_1^2 \gamma_N j_1^2 - N_0 P^{\mathrm{GP}} \right) (1-P^{\mathrm{GP}}) \right] \right)^{1/2} N^{1/2} \,, \nonumber
\end{align}
where $\Vert \cdot \Vert_2$ denotes the Hilbert-Schmidt norm. 
We apply Eq.~\eqref{eq:aprioricondensate6e} to see that the trace in the second line of Eq.~\eqref{eq:densitymatrix22} is bounded by $o(N) + N_0 \left\Vert P^{\mathrm{GP}} - j_1 P^{\mathrm{GP}} j_1 \right\Vert_1$. Moreover, the exponential decay of $\phi^{\mathrm{GP}}_{1,N_0 a_N}$, see \cite[Appendix~A]{RobertGPderivation}, implies 
\begin{equation}
\left\Vert P^{\mathrm{GP}} - j_1 P^{\mathrm{GP}} j_1 \right\Vert_1 \lesssim \left( \int_{B(R)^{\mathrm{c}}} \left\vert \phi^{\mathrm{GP}}_{1,N_0 a_N}(x) \right\vert^2 \text{d}x \right)^{1/2}  \lesssim e^{- c \omega^{1/2}R}
\label{eq:densitymatrix22b}
\end{equation}
for an appropriately chosen constant $c>0$. Since $R \gg \omega^{-1/2}$ we know that the last term on the right-hand side of Eq.~\eqref{eq:densitymatrix22b} is of order $o(1)$. Together with Eqs.~\eqref{eq:densitymatrix20}--\eqref{eq:densitymatrix22b}, we therefore see that
\begin{equation}
\left\Vert j_1^2 \gamma_N j_2^2 \right\Vert_1 + N_0 \left\Vert P^{\mathrm{GP}} - j_1 P^{\mathrm{GP}} j_1 \right\Vert_1 \leq o(N)
\label{eq:densitymatrix22c}
\end{equation} 
holds.

It remains to give a bound on the last term on the right-hand side of Eq.~\eqref{eq:densitymatrix19}. We add and subtract $j_2 \tilde{\gamma}^{\mathrm{gc}}_0$ and use $j_2 \leq 1$ to see that
\begin{equation}
\left\Vert \tilde{\gamma}^{\mathrm{gc}}_0 - j_2 \tilde{\gamma}^{\mathrm{gc}}_0 j_2 \right\Vert_1 \leq  2 \left\Vert (1-j_2) \tilde{\gamma}^{\mathrm{gc}}_0 \right\Vert_1 \leq 2 \left\Vert (1-j_2) \left( \tilde{\gamma}^{\mathrm{gc}}_0 \right)^{1/2} \right\Vert_2 \left\Vert \left( \tilde{\gamma}^{\mathrm{gc}}_0 \right)^{1/2} \right\Vert_2
\label{eq:densitymatrix22d}
\end{equation}
holds.  The last factor equals $(N_{\mathrm{th}}^{\mathrm{gc}})^{1/2}$. The square of the first factor, on the other hand, can be bounded by
\begin{equation}
\left\Vert (1-j_2) \left( \tilde{\gamma}^{\mathrm{gc}}_0 \right)^{1/2} \right\Vert_2^2 = \int_{\mathbb{R}^3} \tilde{\gamma}^{\mathrm{gc}}_0(x,x) \left( 1 - j_2(x) \right)^2 \text{d}x \leq \frac{4 \pi (2R)^3}{3} \sup_{x \in \mathbb{R}^3} \tilde{\gamma}^{\mathrm{gc}}_0(x,x).
\label{eq:densitymatrix22e}
\end{equation}
To obtain this bound, we used that the support of $0\leq 1-j_2\leq 1$ is given by $B(2R)$. We claim that $\tilde{\gamma}_0^{\mathrm{gc}}(x,x) \lesssim \beta^{-3/2}$ holds. This can be seen by an analysis similar to a part of the analysis carried out in Lemma~\ref{lem:densitybound}: We first replace the chemical potential $\mu_0$ by a larger one that guarantees $-\frac{3 \omega}{2} - \mu \geq 0$ to hold. In fact, we choose $\mu = -\frac{3 \omega}{2}$. Then there exists a constant $C>0$ such that 
\begin{equation}
\frac{e^{\beta \left( e_{\alpha}(h) - \mu \right)}-1}{e^{\beta \left( e_{\alpha}(h) - \mu_0 \right)}-1} \leq C \label{eq:densitymatrix23b}
\end{equation}
holds for all $\alpha \geq 1$. This can easily be checked when we realize that the expression on left-hand side of the above equation is monotone decreasing in $e_{\alpha}(h)$ if $\mu < \mu_0$. It therefore suffices to check the inequality with $e_{\alpha}(h)$ replaced by $\omega$. Eq.~\eqref{eq:densitymatrix23b} implies
\begin{equation}
\sum_{\alpha=1}^{\infty} \frac{1}{e^{\beta \left( e_{\alpha}(h) - \mu_0 \right)}-1} \left\vert \varphi_{\alpha}(x) \right|^2 \leq \sum_{\alpha=1}^{\infty} \frac{C}{e^{\beta \left( e_{\alpha}(h) - \mu \right)}-1} \left\vert \varphi_{\alpha}(x) \right|^2
\label{eq:densitymatrix23c}
\end{equation}
where $\lbrace \varphi_{\alpha} \rbrace_{\alpha=0}^{\infty}$ denotes the set of eigenfunctions of $h$. Since $\frac{x^2 \omega^2}{4} - \frac{3 \omega}{2} - \mu \geq 0$, we can argue as in the proof of Lemma~\ref{lem:densitybound}, Eqs.~\eqref{eq:densitybound4}--\eqref{eq:densitybound6}, to show that
\begin{equation}
\tilde{\gamma}_{0}^{\mathrm{gc}}(x,x) \lesssim \beta^{-3/2}
\label{eq:densitymatrix23e}
\end{equation}
holds. 

Together with Eq.~\eqref{eq:densitymatrix22d} and Eq.~\eqref{eq:densitymatrix22e}, Eq.~\eqref{eq:densitymatrix23e} implies 
\begin{equation}
\left\Vert \tilde{\gamma}_0^{\mathrm{gc}} - j_2 \tilde{\gamma}_0^{\mathrm{gc}} j_2 \right\Vert_1 \lesssim \left( N_{\mathrm{th}}^{\mathrm{gc}} \right)^{1/2} \frac{R^{3/2}}{\beta^{3/4}}.
\label{eq:densitymatrix23f}
\end{equation}
For our choice of $R$ in Section~\ref{sec:finalestimatelowerbound}, we have $R^{3/2} \ll \beta^{3/4} (\beta\omega)^{3/2}N$ and hence the right-hand side of \eqref{eq:densitymatrix23f} is of order $o(N)$. Together with 
Eqs.~\eqref{eq:densitymatrix19} and~\eqref{eq:densitymatrix22c}  this  proves
\begin{equation}
\left\Vert \gamma_N - \tilde{\gamma}_0^{\mathrm{gc}} - N_0 P^{\mathrm{GP}} \right\Vert_1 \leq o(N). 
\label{eq:densitymatrix24}
\end{equation}
The desired bound in Eq.~\eqref{eq:mainresult6} then follows from Corollary~\ref{lem:particlenumbers}.

Recall that so far we have worked under  the assumptions $N_0 \sim N$ 
 and $N a_N \geq \epsilon \omega^{-1/2}$ for some $\epsilon >0$.
It remains to consider the case where either $N_0 = o(N)$ and/or $a_N \ll \omega^{-1/2} N^{-1}$. 
In each of these cases, we have $\| P^\mathrm{GP} - |\varphi_0\rangle\langle \varphi_0| \| \ll 1$, and also $E^{\mathrm{GP}}(N_0,a_N,\omega) \ll \omega N$. The equivalent of Eq.~\eqref{eq:densitymatrix1} therefore reads
\begin{equation}
\tr\left[ H_N \Gamma_N \right] - TS(\Gamma_N) = F_0(\beta,N,\omega) + o(\omega N)
\label{eq:densitymatrix18k}
\end{equation}
and implies 
\begin{align}
o(\omega N)  \geq& \tr\left[ \left( \text{d}\Upsilon\left( h-\mu_0 \right) \right) \Gamma_{N} \right] - T S\left( \Gamma_{N} \right) - \tfrac{1}{\beta}\tr\left[ \ln\left( 1-e^{-\beta(h-\mu_0)} \right) \right]. 
\label{eq:densitymatrix18l}
\end{align}
With this input, we go through the analysis of Section~\ref{sec:aprioriboundj2gammaNj2}. In case $\mu_0 \lesssim - \omega$, we can directly apply Lemma~\ref{lem:relativeentropy} and the subsequent estimates, and the equivalent of Eq.~\eqref{eq:densitymatrix18c3} tells us that $\| \gamma_N - \gamma_0^{\mathrm{gc}} \|_1 \leq o(N)$. Together with Corollary~\ref{lem:particlenumbers}  this implies the claim. If $|\mu_0| \ll \omega$, we have to first remove the condensate. Let $P=|\varphi_0\rangle\langle \varphi_0|$ denote the projection onto the ground state of $h$, and $Q=1-P$.  From the subadditivity of the entropy and Eq.~\eqref{eq:densitymatrix3a}, we have 
\beq
S(\Gamma_N) \leq s(P\gamma_NP)+s(Q\gamma_NQ) \leq s(Q\gamma_NQ) + 1 + \ln(1+ N)\,.
\eeq
Since $\tr P \ln ( 1-e^{-\beta(h-\mu_0)} ) = \ln (1-e^{\beta\mu_0}) \leq 0$ and $\tr (h-\mu) P\gamma_N P \geq0$, we conclude from Eq.~\eqref{eq:densitymatrix18l} that $\mathcal{S}(Q\gamma_N Q,\tilde\gamma_0^\mathrm{gc}) = o( \beta \omega N)$.  The analysis of Section~\ref{sec:aprioriboundj2gammaNj2} then implies that $\| Q \gamma_N Q - \tilde\gamma_0^{\mathrm{gc}} \|_1 \leq o(N)$. Since $\tr \gamma_N = \tr \gamma_0^\mathrm{gc}$, this also implies that $\| P \gamma_N P - N_0^\mathrm{gc} P \|_1 = o(N)$. Finally, by arguing as in Eq.~\eqref{eq:densitymatrix20}, one easily sees that $\| P \gamma_N Q\|_1 \leq o(N)$. In combination with Corollary~\ref{lem:particlenumbers}, this shows that also in this case $\| \gamma_N - \gamma_{N,0} \|_1 \leq o(N)$. This completes the proof of Eq.~\eqref{eq:mainresult6}.

To conclude the proof of Theorem~\ref{thm:main} it remains to prove Eq.~\eqref{eq:mainresult7}. To that end, we write
\begin{equation}
\left\Vert \gamma_N - N_0 P^{\mathrm{GP}} \right\Vert \leq \left\Vert \gamma_N - N_0 P^{\mathrm{GP}} - \tilde{\gamma}_{N,0} \right\Vert + \left\Vert \tilde{\gamma}_{N,0} - \tilde{\gamma}^{\mathrm{gc}}_0 \right\Vert + \left\Vert \tilde{\gamma}^{\mathrm{gc}}_0 \right\Vert.
\label{eq:densitymatrix18m}
\end{equation} 
The first term on the right-hand side of Eq.~\eqref{eq:densitymatrix18m} can be bounded by the trace norm of the same expression, which is bounded by $o(N)$ according to Eq.~\eqref{eq:mainresult6}. Similarly, Eq.~\eqref{eq:freedensityclose} implies that the second term is bounded by a constant times $(N \ln N)^{1/2}$. From the explicit form of $\tilde{\gamma}^{\mathrm{gc}}_0$ we deduce that its largest eigenvalue is of the order $O((\beta \omega)^{-1})$. Therefore,
\begin{equation}
\left\Vert \gamma_N - N_0 P^{\mathrm{GP}} \right\Vert \leq o(N)
\label{eq:densitymatrix18n}
\end{equation}
and the proof of Theorem~\ref{thm:main} is complete.

\appendix
\section{Some properties of the ideal Bose gas}
\label{sec:appendix1}
In this Appendix we collect several statements about the ideal Bose that are needed in the proof of Theorem~\ref{thm:main} and do not seem to have appeared in the literature before (except for the first part of Proposition~\ref{prop:idealgas}). We start by introducing some notation. Fix a  nondecreasing sequence $\left\lbrace E_j \right\rbrace_{j=0}^{\infty}$ of nonnegative real numbers. A vector $n=(n_0,n_1, \ldots)$ of infinite length is called a \textit{configuration} if all its entries are nonnegative integers and if only a finite number of them is different from zero. The collection $\left\lbrace E_j \right\rbrace_{j=0}^{\infty}$ plays the role of the energy levels of a one-particle quantum system with the temperature factor $\beta$ absorbed, and a configuration $n$ labels an element of the standard basis of the bosonic Fock space. For each configuration $n$, we define its energy to be $E(n) = \sum_{j \geq 0} E_j n_j$. The canonical partition function is given by $Z(N)=\sum_{|n|=N} \exp\left(-E(n)\right)$ for $N \in \mathbb{N}_0$. Here $\left| n \right| = \sum_{j=0}^{\infty} n_j$ denotes the number of particles in the configuration $n$. Let $\langle A \rangle_N$ be the expectation of an operator $A$ in the canonical Gibbs state related to $Z(N)$. By $a_j^*$ and $a_j$ we denote the bosonic creation and annihilation operator of a particle with the energy $E_j$ and we define $\hat{n}_j = a^*_j a_j$. The free energy of the system  is given by
\begin{equation}
F(N)=-\ln\left( Z(N) \right)
\end{equation}
and the expectation of $f(\hat{n}_j)$ for $j \in \mathbb{N}_0$ with a function $f:\mathbb{N}_0 \rightarrow \mathbb{R}$ reads
\begin{equation}
\left\langle f(\hat{n}_j) \right\rangle_N = \frac{ \sum_{|n|=N} f(n_j) e^{-E(n)} }{Z(N)}.
\end{equation}
We assume that the energy levels $\left\lbrace E_j \right\rbrace_{j=0}^{\infty}$ and the function $f$ are such that the partition function and the expectations of all $f(\hat{n}_j)$ are finite. We then have: 
\begin{proposition}
\label{prop:idealgas}
The map $N \mapsto F(N)$ is convex, i.e., $F(N+1) + F(N-1) \geq 2 F(N)$ holds for all $N \geq 1$, and for any nonnegative, nondecreasing function $f:\mathbb{N}_0 \rightarrow \mathbb{R}$ the map $N \mapsto \left\langle f(\hat{n}_j) \right\rangle_N$ is nondecreasing for each $j \in \mathbb{N}_0$. 
\end{proposition}
The proof of the first statement seems to appear for the first time in \cite{LewisPuleZagrebnov1988} and it was later reproven in \cite{Suto}. The second statement has been shown in \cite{PuleZagrebnov} for the functions $f(x) = x^k$ with $k \in \mathbb{N}$ and in \cite{Suto} for $f(x) = x$. The proof in \cite{Suto} still works if one replaces the identity by a nonnegative, nondecreasing function $f$.

Proposition~\ref{prop:idealgas} has two consequences that will be of importance for us. The first is naturally formulated in the following setup: Denote by $Z_{\mathrm{gc}}(\mu) = \sum_{n} \exp(-(E(n) - \mu |n| ))$ the grand canonical partition function and define $\lambda_{N,\mu} = Z(N) \exp(\mu N)/Z_{\mathrm{gc}}(\mu)$. The grand canonical free energy is given by $F_{\mathrm{gc}}(\mu) = \sum_{N \geq 0} \lambda_{N,\mu} F(N) - S(\lambda)$, where $S(\lambda) = -\sum_{N\geq 0} \lambda_{N,\mu} \ln(\lambda_{N,\mu})$. The first consequence of Proposition~\ref{prop:idealgas} is the following statement which quantifies the difference between the canonical free energy and its grand canonical counterpart:
\begin{corollary}
\label{cor:freeenergy}
Assume $\mu$ is such that $\overline{N} = \sum_{N \geq 0} N \lambda_{N,\mu} \in \mathbb{N}$. Then
\begin{equation}
F(\overline{N}) \geq F_{\mathrm{gc}}(\mu) \geq F(\overline{N}) - \ln(1+\overline{N}) - 1.
\label{eq:idealgas15}
\end{equation}
\end{corollary}
\begin{proof}
The convexity of the map $N \mapsto F(N)$, see Proposition~\ref{prop:idealgas}, implies $\sum_{N \geq 0} \lambda_{N,\mu} F(N) \geq F(\overline{N})$. Denote by $\mathcal{M}$ the set of all real, nonnegative sequences $w=\lbrace w_N \rbrace_{N=0}^{\infty}$ with the properties that $\sum_{N \geq 0} w_N = 1$ and $\sum_{N \geq 0} N w_N = \overline{N}$. By definition, we have $S(\lambda) \leq \sup_{w \in \mathcal{M}} S(w)$. The supremum on the right-hand side can be computed explicitly and equals $\ln(1+\overline{N}) + \overline{N} \ln((1+\overline{N})/\overline{N})$. Since $\overline{N} \ln( (1+\overline{N})/\overline{N} ) \leq 1$ this proves the lower bound for $F_{\mathrm{gc}}(\mu)$. The upper bound follows from the Gibbs variational principle.
\end{proof}
The second consequence is an estimate of the canonical one/two-particle density in terms of the grand canonical one-particle density. 
\begin{proposition}
\label{cor:density}
Let $h$ be a one-particle Hamiltonian on $L^2(\mathbb{R}^{d})$, $d \geq 1$, with energy levels $\lbrace E_j \rbrace_{j=0}^\infty$ and eigenfunctions $\lbrace \phi_j \rbrace_{j=0}^{\infty}$, that is, $h \phi_j = E_j \phi_j$. Denote by
\begin{align}
\varrho_N^{\mathrm{c}}(x) &= \sum_{j} \langle \hat{n}_j \rangle_N | \phi_j(x) |^2 \quad \text{ and }  \label{eq:idealgas16} \\
\varrho_N^{(2),\mathrm{c}}(x,y) &= \sum_{j_1,j_2,j_3,j_4} \overline{ \phi_{j_1}(x) } \overline{ \phi_{j_2}(y) } \phi_{j_3}(y) \phi_{j_4}(x) \left\langle a_{j_1}^* a_{j_2}^* a_{j_3} a_{j_4} \right\rangle_N \nonumber
\end{align}
the canonical one-particle and two-particle densities, respectively. The grand canonical one-particle density is given by
\begin{equation}
\varrho^{\mathrm{gc}}_{\mu}(x) = \sum_{N \geq 0} \lambda_{N,\mu} \varrho_N^{\mathrm{c}}(x), \label{eq:idealgas16b} 
\end{equation}
where $\mu$ is chosen such that $\overline{N} = \sum_{N \geq 0} N \lambda_{N,\mu} \in \mathbb{N}$. 
Then
\begin{equation}
\varrho_{\overline{N}}^{\mathrm{c}}(x) \leq \tfrac{40}{1.8} \varrho^{\mathrm{gc}}_{\mu}(x)  \quad \text{ and } \quad \varrho_{\overline{N}}^{(2),\mathrm{c}}(x,y) \leq 4 \left(  \tfrac{40}{1.8} \right)^2 \varrho^{\mathrm{gc}}_{\mu}(x) \varrho^{\mathrm{gc}}_{\mu}(y)
\label{eq:idealgas18}
\end{equation}
holds almost everywhere.
\end{proposition}
\begin{proof}
We start with the proof of the inequality for the one-particle density. From the definition of the one-particle densities, Eqs.~\eqref{eq:idealgas16} and~\eqref{eq:idealgas16b}, and the monotonicity of the map $N \mapsto \langle \hat{n}_j \rangle_N$, see Proposition~\ref{prop:idealgas}, we find
\begin{equation}
\varrho^{\mathrm{gc}}_{\mu}(x) \geq \varrho^{\mathrm{c}}_{\overline{N}}(x) \sum_{N \geq \overline{N}} \lambda_{N,\mu}.
\label{eq:idealgas19}
\end{equation}
To prove the claim, we need to show that $\sum_{N \geq \overline{N}} \lambda_{N,\mu} \geq \tfrac{1.8}{40}$ holds. From Lemma~\ref{lem:fourthmomentbound} below we know that 
\begin{equation}
\sum_{N \geq 0} \lambda_{N,\mu} \left( N-\overline{N} \right)^4 \leq 9 \left( \sum_{N \geq 0} \lambda_{N,\mu} \left( N-\overline{N} \right)^2 \right)^2 + \sum_{N \geq 0} \lambda_{N,\mu} \left( N-\overline{N} \right)^2.
\label{eq:idealgas17}
\end{equation}
The grand canonical Gibbs state is quasi-free and  hence can use Wick's Theorem to see that 
\begin{equation}
\sum_{N \geq 0} \lambda_{N,\mu} \left( N-\overline{N} \right)^2 \geq \overline{N} \geq 1
\label{eq:idealgas17b}
\end{equation}
holds. Hence, we can bound the centered fourth moment of the particle number in Eq.~\eqref{eq:idealgas17} by 10 times the variance squared.

Let us define the new random variable $X$ by $X = (N-\overline{N}) (\textbf{E}((N-\overline{N})^2))^{-1/2}$ which by our assumptions has the following properties:
\begin{equation}
\mathbf{E}(X)=0, \quad \mathbf{E}(X^2)=1 \quad \text{ and } \quad \mathbf{E}(X^4) = Y. \label{eq:idealgas20}
\end{equation} 
Here $\mathbf{E}(X)$ is the expectation of $X$ and we have $Y \leq 10$ by the arguments in the previous paragraph. Denote by $\textbf{P}_X$ the probability measure on $\mathbb{R}$ induced by $X$ and choose $a$, $b$, $d$ such that the function $f(x)=ax+bx^2-dx^4$ obeys $f(x) \leq \chi_{[0,\infty)}(x)$ for all $x \in \mathbb{R}$. Here $\chi_{[0,\infty)}$ denotes the characteristic function of the interval $[0,\infty)$. We then have
\begin{equation}
\mathbf{P}(X \geq 0) = \int_{\mathbb{R}} \chi_{[0,\infty)}(s) \text{d}\mathbf{P}_X(s) \geq \int_{\mathbb{R}} f(s) \text{d}\mathbf{P}_X(s) = b - Y d. 
\label{eq:idealgas21}
\end{equation}
Explicit optimization of the right-hand side of Eq.~\eqref{eq:idealgas21} under the constraint $f(x) \leq \chi_{[0,\infty)}(x)$ yields 
\begin{equation}
\mathbf{P}(X \geq 0) \geq \frac{2 \sqrt{3} - 3}{Y} \geq \frac{1.8}{40}.
\label{eq:idealgas22}
\end{equation}
This proves the claim for the one-particle densities (assuming the validity of Lemma~\ref{lem:fourthmomentbound}).

It remains to prove the bound for the two-particle densities. An application of the Cauchy-Schwarz inequality tells us that
\begin{align}
\varrho^{(2),\mathrm{c}}_{\overline{N}}(x,y) &= \sum_{i,j} \left| \frac{1}{2}   \sum_{\sigma \in S_2} { \phi_{\sigma(i)}(x) } \overline{ \phi_{\sigma(j)}(y) } \right|^2  \left\langle  a_i^* a_j^*  a_j a_i  \right\rangle_{\overline{N}} \label{eq:lemupperbound221} \\
&\leq \sum_{i,j} \left\vert \phi_i(x) \right\vert^2 \left\vert \phi_j(y) \right\vert^2 \left\langle a_i^* a_j^* a_j a_i \right\rangle_{\overline{N}} \nonumber
\end{align}
holds. Here $S_2$ denotes the group of permutations of two elements. Let us denote by $\langle A \rangle_{\mathrm{gc}} = \sum_{N \geq 0} \lambda_{N,\mu} \langle A \rangle_N$ the expectation of an operator $A$ in the grand canonical Gibbs state. We want to derive an upper bound for the expectation value in the second line on the right-hand side of Eq.~\eqref{eq:lemupperbound221}. For $i = j$ we have $a_i^* a_i^* a_i a_i = \hat{n}_i (\hat{n}_i - 1)$ with $\hat n_i = a_i^* a_i$. From Proposition~\ref{prop:idealgas} we know that the map $N \mapsto \langle \hat{n}_j(\hat{n}_j - 1) \rangle_N$ is nondecreasing. Using this fact, we argue as in the case of the one-particle density to see that 
\begin{equation}
\langle \hat{n}_j(\hat{n}_j - 1) \rangle_{\overline{N}} \leq \frac{40}{1.8} \langle \hat{n}_j(\hat{n}_j - 1) \rangle_{\mathrm{gc}}
\label{eq:lemupperbound221a}
\end{equation}
holds. The right-hand side of Eq.~\eqref{eq:lemupperbound221a} can be simplified when we use Wick's Theorem: $\langle \hat{n}_j(\hat{n}_j - 1) \rangle_{\mathrm{gc}} = 2 \langle \hat{n}_j \rangle_{\mathrm{gc}}^2$. If $i \neq j$ we have $a_i^* a_j^* a_j a_i = \hat{n}_i \hat{n}_j$ and \cite{Suto} tells us that
\begin{equation}
\langle \hat{n}_i \hat{n}_j \rangle_{\overline{N}} \leq \langle \hat{n}_i \rangle_{\overline{N}} \langle \hat{n}_j \rangle_{\overline{N}} 
\label{eq:lemupperbound221b}
\end{equation}
 holds. As in the previous case, we use $\langle \hat{n}_j \rangle_{\overline{N}} \leq \frac{40}{1.8} \langle \hat{n}_j \rangle_{\mathrm{gc}}$. Combining these estimates with Eq.~\eqref{eq:lemupperbound221}, we finally obtain
\begin{equation}
\varrho^{(2),\mathrm{c}}_{\overline{N}}(x,y) \leq 4 \left( \frac{40}{1.8} \right)^2 \sum_{i,j} \left\vert \phi_i(x) \right\vert^2 \left\vert \phi_j(y) \right\vert^2 \left\langle \hat{n}_i \right\rangle_{\mathrm{gc}} \left\langle \hat{n}_j \right\rangle_{\mathrm{gc}}.
\label{eq:lemupperbound222}
\end{equation}
This proves the claim \eqref{eq:idealgas18}.
\end{proof}

\begin{remark}\label{rem:app}
The first part of the proof shows that $\langle \hat n_j\rangle_N \leq \frac{40}{1.8}  \langle \hat n_j \rangle_\textrm{gc}$ for all $j$, where $\langle \, \cdot\,\rangle_\mathrm{gc}$ denotes the corresponding grand canonical state with average particle number $N$. In particular, $N_0 \lesssim N_0^\mathrm{gc}$ and $N_\mathrm{th} \lesssim N_\mathrm{th}^\mathrm{gc}$ holds. 
\end{remark}

The next Lemma provides an estimate of the fourth moment of the particle number in the grand canonical ensemble in terms of the second moment. It is needed in the proof of Proposition~\ref{cor:density}.
\begin{lemma}
\label{lem:fourthmomentbound}
Let $\hat{N} = \sum_{j=0}^{\infty} \hat{n}_j$ be the particle number operator and denote by $\overline{N} = \langle \hat{N} \rangle_{\mathrm{gc}}$ the expected number of particles in the grand canonical ensemble. We then have
\begin{equation}
\langle (\hat{N} - \overline{N})^4 \rangle_{\mathrm{gc}} \leq 9 \langle (\hat{N}-\overline{N})^2 \rangle_{\mathrm{gc}}^2 + \langle (\hat{N}-\overline{N})^2 \rangle_{\mathrm{gc}}. 
\label{eq:fourthmomentbound0}
\end{equation}
\end{lemma}
\begin{proof}
If we use that $\tfrac{\partial}{\partial \mu} Z_{\mathrm{gc}}(\mu) = Z_{\mathrm{gc}}(\mu) \overline{N}$ holds, a simple computation leads to
\begin{equation}
\left( \frac{\partial}{\partial \mu} \right)^3 \overline{N} = \left\langle \hat{N}^4 \right\rangle_{\mathrm{gc}} - 4 \left\langle \hat{N}^3 \right\rangle_{\mathrm{gc}} \overline{N} - 3 \left \langle \hat{N}^2 \right\rangle_{\mathrm{gc}}^2 + 12  \left\langle \hat{N}^2 \right\rangle_{\mathrm{gc}} \overline{N}^2 - 6 \overline{N}^4. \label{eq:fourthmomentbound1}
\end{equation}
On the other hand, 
\begin{equation}
\left\langle \left( \hat{N}-\overline{N} \right)^4 \right\rangle_{\mathrm{gc}} = \left\langle \hat{N}^4 \right\rangle_{\mathrm{gc}} - 4 \left\langle \hat{N}^3 \right\rangle_{\mathrm{gc}} \overline{N} + 6 \left\langle \hat{N}^2 \right\rangle_{\mathrm{gc}} \overline{N}^2 - 3 \overline{N}^4, \label{eq:fourthmomentbound2}
\end{equation}
which together with $\tfrac{\partial}{\partial \mu} \overline{N} = \langle \hat{N}^2 \rangle_{\mathrm{gc}} - \overline{N}^2$ allows us to conclude that
\begin{equation}
\left\langle \left( \hat{N}-\overline{N} \right)^4 \right\rangle = \left( \frac{\partial}{\partial \mu} \right)^3 \overline{N} + 3 \left( \frac{\partial}{\partial \mu} \overline{N} \right)^2. \label{eq:fourthmomentbound3}
\end{equation}
To treat the first term on the right-hand side, we need to do a little computation. It yields
\begin{equation}
\frac{\partial}{\partial \mu} \sum_{j=0}^{\infty} \frac{1}{e^{E_j - \mu}-1} = \sum_{j=0}^{\infty} \frac{1}{4 \sinh\left( \frac{E_j - \mu}{2} \right)^2} \label{eq:fourthmomentbound4}
\end{equation}
as well as
\begin{equation}
\left( \frac{\partial}{\partial \mu} \right)^3 \sum_{j=0}^{\infty} \frac{1}{e^{E_j - \mu}-1} = \sum_{j=0}^{\infty} \left[ \frac{3}{8 \sinh\left( \frac{E_j - \mu}{2} \right)^4} + \frac{1}{4 \sinh\left( \frac{E_j - \mu}{2} \right)^2} \right]. \label{eq:fourthmomentbound5}
\end{equation}
With Eq.~\eqref{eq:fourthmomentbound4} we see that the first term on the right-hand side of Eq.~\eqref{eq:fourthmomentbound5} is bounded from above by six times the variance of the particle number squared. Together with Eq.~\eqref{eq:fourthmomentbound3} this proves the claim.
\end{proof}

The last statement of this Appendix is an estimate on the trace norm difference of the one-particle density matrices of the canonical and grand canonical Gibbs states, which we denote by $\gamma_{N}$ and $\gamma_N^\textrm{gc}$, respectively. The latter equals 
\begin{equation}
\gamma_N^\mathrm{gc} = \frac 1{e^{h-\mu} -1} =  \sum_{i\geq 0} \frac 1{e^{E_i-\mu} -1 } | \phi_j\rangle\langle\phi_j|
\end{equation}
where $\mu<0$ is chosen such that $\tr \gamma_N^\mathrm{gc} = N \in \mathbb{N}$. We shall also introduce $\tilde{\gamma}_N^{\mathrm{gc}} = \gamma_N^\mathrm{gc} - N_0^\mathrm{gc} |\phi_0\rangle\langle \phi_0|$, where $N_0^\mathrm{gc} = (e^{-\mu}-1)^{-1}$ is the number of particles in the condensate in the grand canonical ensemble, and similarly $\tilde{\gamma}_N = \gamma_{N} - N_0 \vert \phi_0 \rangle \langle \phi_0 \vert$, with $N_0 = \langle \hat n_0\rangle_N$ the number of particles in the condensate in the  canonical ensemble.

\begin{lemma}
\label{lem:freedensityclose}
With the definitions above,  we have
\begin{equation}
|N_0 - N_0^\mathrm{gc}| \leq \left\Vert \tilde{\gamma}_N^{\mathrm{gc}} - \tilde{\gamma}_{N} \right\Vert_1 \leq \left\Vert {\gamma}_N^{\mathrm{gc}} - {\gamma}_{N} \right\Vert_1 \lesssim \left( \tr \left[\tilde{\gamma}_N^{\mathrm{gc}} (1 + \tilde{\gamma}_N^{\mathrm{gc}}) \right]    \ln N \right)^{1/2} +  \left(1 +\| \tilde\gamma_N^\mathrm{gc}\| \right) \ln N   \,.
\label{eq:freedensityclose}
\end{equation}
\end{lemma}

\begin{proof}
Let $\Gamma_N^\mathrm{G}$ denote the canonical Gibbs state, with one-particle density matrix $\gamma_N$. With $s(\gamma)$ defined in Eq.~\eqref{eq:densitymatrix5b}, we have the entropy inequality (see, e.g.,  \cite[2.5.14.5]{Thirring_4}) $S(\Gamma_{N}^\mathrm{G}) \leq s(\gamma_{N})$. Therefore,
\begin{equation}
F(N) \geq \tr\left[ h \gamma_{N} \right] - s(\gamma_{N}) = F_{\mathrm{gc}}(\mu) + \mathcal{S}(\gamma_N, \gamma_N^\mathrm{gc})
\label{eq:densitymatrix18d}
\end{equation}
where $\mathcal{S}$ denotes the relative entropy defined in \eqref{def:cS}. 
Corollary~\ref{cor:freeenergy} thus implies that
\begin{equation}
\mathcal{S}\left( \gamma_{N}, \gamma_N^{\mathrm{gc}} \right) \lesssim \ln(N)\,.
\label{eq:densitymatrix18e}
\end{equation}
Note that $\gamma_{N}$ and $\gamma_N^\mathrm{gc}$ have the same eigenfunctions $\phi_j$. In particular, they commute. Hence the relative entropy can be written as
\beq
\mathcal{S}\left( \gamma_{N}, \gamma_N^{\mathrm{gc}} \right) = \sum_{\alpha\geq 0} S(a_\alpha,b_\alpha) 
\eeq
where $S$ is defined in the proof of Lemma~\ref{lem:relativeentropy}, and $a_\alpha$ resp. $b_\alpha$ denote the eigenvalues of $\gamma_{N}$ and $\gamma_N^\mathrm{gc}$, respectively.  Using \eqref{eq:relativeentropy4} and the Cauchy-Schwarz inequality, this implies
\begin{align}\nonumber
\left\Vert \tilde{\gamma}_N^{\mathrm{gc}} - \tilde{\gamma}_{N} \right\Vert_1 & = \sum_{\alpha\geq 1} | a_\alpha - b_\alpha| \leq  \left( \sum_{\alpha\geq 1} \frac { (a_\alpha - b_\alpha)^2}{(a_\alpha+b_\alpha)(1+b_\alpha)} \right)^{1/2} 
 \left( \sum_{\alpha\geq 1}  (a_\alpha+b_\alpha)(1+b_\alpha) \right)^{1/2} 
\\ & \lesssim \mathcal{S}\left( \gamma_{N}, \gamma_N^{\mathrm{gc}} \right)^{1/2} \left\{  \left(1 +\| \tilde\gamma_N^\mathrm{gc}\| \right) \left\Vert \tilde{\gamma}_N^{\mathrm{gc}} - \tilde{\gamma}_{N} \right\Vert_1 + \tr \left[\tilde{\gamma}_N^{\mathrm{gc}} (1 + \tilde{\gamma}_N^{\mathrm{gc}}) \right]  \right\}^{1/2}\,.  \label{eq:dm19}
\end{align}
In combination with Eq.~\eqref{eq:densitymatrix18e}, this gives
\beq
\left\Vert \tilde{\gamma}_N^{\mathrm{gc}} - \tilde{\gamma}_{N} \right\Vert_1 \lesssim \left(  \tr \left[\tilde{\gamma}_N^{\mathrm{gc}} (1 + \tilde{\gamma}_N^{\mathrm{gc}}) \right]    \ln N \right)^{1/2} +  \left(1 +\| \tilde\gamma_N^\mathrm{gc}\| \right) \ln N \,.
\eeq
The claim \eqref{eq:freedensityclose} then follows from the fact that $N = \tr \gamma_{N} = \tr \gamma_N^\mathrm{gc}$, which implies that 
\beq
|N_0 - N_0^\mathrm{gc}| = | \tr ( \tilde{\gamma}_N^{\mathrm{gc}} - \tilde{\gamma}_{N}) | \leq \left\Vert \tilde{\gamma}_N^{\mathrm{gc}} - \tilde{\gamma}_{N} \right\Vert_1 \,
\eeq
holds.
\end{proof}


We shall now apply Lemma~\ref{lem:freedensityclose} to the case of the harmonic oscillator Hamiltonian in Eq.~\eqref{eq:idealbosegas1a}. We  re-introduce the inverse temperature $\beta$, and adjust the notation to the one used in the main text. That is, we denote the one-particle density matrices by $\gamma_{N,0}$ and $\gamma_0^\mathrm{gc}$, respectively. Recall also the definitions $N_\mathrm{th} = N - N_0$ and similarly for $N_\mathrm{th}^\mathrm{gc}$. In this case, we have
\begin{equation}
\tr \tilde\gamma_0^\mathrm{gc} = N_0^\mathrm{gc} = O((\beta\omega)^{-3}), \  \ \| \tilde\gamma_0^\mathrm{gc} \| = O ( (\beta\omega)^{-1}) 
\end{equation}
and also 
\begin{equation}
\tr \left(\tilde\gamma_0^\mathrm{gc}\right)^2  \lesssim  \int_0^\infty  \frac {x^2} { \left( e^{\beta\omega x} -1 \right)^2 } \textrm{d}x = O ((\beta\omega)^{-3}) 
\end{equation}
Hence we obtain the following Corollary. 

\begin{corollary}
\label{lem:particlenumbers}
Consider the three-dimensional ideal Bose gas in the harmonic oscillator potential, that is, the one-particle Hamiltonian of the system is given by $h=-\Delta + \tfrac{\omega^2}{4} x^2 - \tfrac{3}{2} \omega$. We assume that the chemical potential $\mu_0$ is chosen such that the expected number of particles in the grand canonical ensemble equals $N \in \mathbb{N}$. Then 
\begin{equation}
|N_0 - N_0^\mathrm{gc}| = |N_\mathrm{th} - N_\mathrm{th}^\mathrm{gc} | \leq \left\Vert {\gamma}_0^{\mathrm{gc}} - {\gamma}_{N,0} \right\Vert_1   \lesssim  (\beta\omega)^{-3/2} \left( \ln N \right)^{1/2} + (\beta\omega)^{-1} \ln N  \,.
\label{eq:freedensitycloseho}
\end{equation}
\end{corollary}

\textbf{Acknowledgments.} 
Partial financial support by the European Research Council (ERC) under the European Union's Horizon 2020 research and innovation programme (grant agreement No 694227), and by the Austrian Science Fund (FWF), project Nr. P 27533-N27, is gratefully acknowledged.

\vspace{0.5cm}

(Andreas Deuchert) Institute of Science and Technology Austria (IST Austria)\\ Am Campus 1, 3400 Klosterneuburg, Austria\\ E-mail address: \texttt{andreas.deuchert@ist.ac.at}

(Robert Seiringer) Institute of Science and Technology Austria (IST Austria)\\ Am Campus 1, 3400 Klosterneuburg, Austria\\ E-mail address: \texttt{robert.seiringer@ist.ac.at}

(Jakob Yngvason) Faculty of Physics, University of Vienna \\ Boltzmanngasse 5, 1090 Vienna, Austria \\ E-mail address: \texttt{jakob.yngvason@univie.ac.at}


\begin{thebibliography}{49}
\addcontentsline{toc}{section}{References}

\bibitem{WieCor1995} M.H. Anderson, J.R. Ensher, M.R. Matthews, C.E. Wieman, E.A. Cornell, \textit{Observation of bose-einstein condensation in a dilute atomic vapor}, Science \textbf{269}, 198 (1995)
 
\bibitem{BenOlivSchl2015} N. Benedikter, G. de Oliveira, B. Schlein, \textit{Quantitative Derivation of the Gross-Pitaevskii Equation}, Comm. Pure Appl. Math. \textbf{68}, 1399 (2015)

\bibitem{BenPorSchl2015} N. Benedikter, M. Porta, B. Schlein, \textit{Effective Evolution Equations from Quantum Dynamics}, Springer, Berlin (2016) 

\bibitem{BratelliRobinson2} O. Bratteli, D. W. Robinson, \textit{Operator Algebras and Quantum Statistical Mechanics 2}, Springer, Berlin (1997)

\bibitem{ChatterjeeDiakonis} S. Chatterjee, P. Diaconis, \textit{Fluctuations of the Bose-Einstein condensate}, J. Phys. A: Math. Theor. \textbf{47}, 085201 (2014) 

\bibitem{SimonSchroedingerOps} L. Cycon, R. G. Froese, W. Kirsch, B. Simon, \textit{Schr\"odinger Operators with Applications to Quantum Mechanics and Global Geometry}, Springer, Heidelberg (1987)

\bibitem{Dalfovo_etal1999}F. Dalfovo, S, Giorgini, L.P. Pitaevskii, S. Stringari, \textit{Theory of Bose-Einstein condensation in trapped gases}, Rev. Mod. Phys. \textbf{71}, 463 (1999)

\bibitem{Kett1995} K.B. Davis, M.-O. Mewes, M.R. Andrews, N.J. van Druten, D.S. Durfee, D.M. Kurn, W. Ketterle, \textit{Bose-Einstein Condensation in a Gas of Sodium Atoms}, 
Phys. Rev. Lett. \textbf{75}, 3969 (1995)

\bibitem{DerezinskiGerard} J. Derezinski, C. Gerard, \textit{Asymptotic completeness in quantum field theory. Massive Pauli-Fierz Hamiltonians}, Rev. Math. Phys. \textbf{11}, 383 (1999)

\bibitem{DHS2015} A. Deuchert, C. Hainzl, R. Seiringer, \textit{Note on a Family of Monotone Quantum Relative Entropies}, Lett. Math. Phys. \textbf{105}, 1449 (2015)

\bibitem{Dyson} F.J. Dyson, \textit{Ground-State Energy of a Hard-Sphere Gas}, Phys. Rev. \textbf{106}, 20 (1957)

\bibitem{ErdSchlYau2009} L. Erd\H os, B. Schlein, H.-T. Yau, \textit{Rigorous derivation of the Gross-Pitaevskii equation with a large interaction potential}, J. Amer. Math. Soc. \textbf{22}, 1099 (2009)

\bibitem{ErdSchlYau2010} L. Erd\H os, B. Schlein, H.-T. Yau, \textit{Derivation of the Gross-Pitaevskii equation for the dynamics of Bose-Einstein condensate}, Ann. of Math. \textbf{172}, 291 (2010)

\bibitem{FKSS2017} J. Fr\"ohlich, A. Knowles, B. Schlein, V. Sohinger, 
\textit{Gibbs Measures of Nonlinear Schr\"odinger Equations as Limits of Many-Body Quantum States in Dimensions 
$d\leq 3$}, Commun. Math. Phys. \textbf{356}, 883 (2017)

\bibitem{Gau2013} A.L. Gaunt, T.F. Schmidutz, I. Gotlibovych, R.P. Smith,  Z. Hadzibabic, \textit{Bose-Einstein Condensation of Atoms in a Uniform Potential}, 
Phys. Rev. Lett. \textbf{110}, 200406 (2013)

\bibitem{HLSThemodynamicLimit} C. Hainzl, M. Lewin, J.P. Solovej, \textit{The thermodynamic limit of quantum Coulomb systems Part II. Applications}, Adv. Math. \textbf{221}, 488 (2009)

\bibitem{Hau1998} Lene Vestergaard Hau, B.D. Busch, Chien Liu, Zachary Dutton, Michael M. Burns, J.A. Golovchenko, \textit{Near Resonant Spatial Images of Confined Bose-Einstein Condensates in the 4-Dee Magnetic Bottle},  
Phys. Rev. A \textbf{58}, R54 (1998)

\bibitem{Jastrow} R. Jastrow, \textit{Many-Body Problem with Strong Forces}, Phys. Rev. \textbf{98}, 1479 (1955)

\bibitem{kocht} H. Koch, D. Tataru, \textit{$L^p$ eigenfunction bounds for the Hermite operator},  Duke Math. J. \textbf{128},  369 (2005)

\bibitem{LewinNamRougerie2015} M. Lewin, P.T. Nam, N. Rougerie, \textit{Derivation of nonlinear Gibbs measures from many-body quantum mechanics}, J. \'Ec. polytech. Math. \textbf{2}, 65 (2015)

\bibitem{LewinNamRougerie2017} M. Lewin, P.T. Nam, N. Rougerie, \textit{Gibbs measures based on 1D (an)harmonic oscillators as mean-field limits}, arXiv:1703.09422 

\bibitem{LewinSabin2014} M. Lewin, J. Sabin, \textit{A family of monotone quantum relative entropies}, Lett. Math. Phys. \textbf{104}, 691 (2014) 

\bibitem{LewisPuleZagrebnov1988} J.T. Lewis, J.V. Pul\'e, V.A. Zagrebnov, \textit{The large deviation principle for the Kac distribution}, Helv. Phys. Acta \textbf{61}, 1063 (1988)

\bibitem{LiebLoss} E.H. Lieb, M. Loss, \textit{Analysis}, AMS, Providence, Rhode Island (2001)

\bibitem{LiSei2002} E.H. Lieb, R. Seiringer, \textit{Proof of Bose-Einstein Condensation for Dilute Trapped Gases}, Phys. Rev. Lett. \textbf{88}, 170409 (2002)

\bibitem{LiSei2006} E.H. Lieb, R. Seiringer, \textit{Derivation of the Gross-Pitaevskii equation for rotating Bose gases}, Commun. Math. Phys. \textbf{264}, 505 (2006)

\bibitem{LiSeiSol2005} E.H. Lieb, R. Seiringer, J.P. Solovej, \textit{Ground State Energy of the Low Density Fermi Gas}, Phys. Rev. A \textbf{71}, 053605 (2005)

\bibitem{Themathematicsofthebosegas} E.H. Lieb, R. Seiringer, J. P. Solovej, J. Yngvason, \textit{The Mathematics of the Bose Gas and its Condensation}, Birkh\"auser, Basel (2005)

\bibitem{RobertGPderivation} E.H. Lieb, R. Seiringer, J. Yngvason, \textit{Bosons in a trap: A rigorous derivation of the Gross-Pitaevskii energy functional}, Phys. Rev. A \textbf{61}, 043602 (2000)

\bibitem{LiYng1998} E.H. Lieb, J. Yngvason, \textit{Ground State Energy of the Low Density Bose Gas}, Phys. Rev. Lett. \textbf{80}, 2504 (1998)

\bibitem{Lopes2017} R. Lopes, C. Eigen, N. Navon, D. Cl\'ement, R. P. Smith, and Z. Hadzibabic, \textit{Quantum Depletion of a Homogeneous Bose-Einstein Condensate}, Phys. Rev. Lett. \textbf{119}, 190404 (2017)

\bibitem{NRS} P.T. Nam, N. Rougerie, R. Seiringer, \textit{Ground states of large Bose systems: the Gross-Pitaevskii limit revisited}, Analysis \& PDE \textbf{9}, 459 (2016)

\bibitem{PethickSmith} C. Pethick, H. Smith, \textit{Bose-Einstein Condensation in Dilute Gases}, Cambridge University Press, New York (2008)

\bibitem{Pickl2015} P. Pickl, \textit{Derivation of the time dependent Gross Pitaevskii equation with external fields}, Rev. Math. Phys. \textbf{27}, 1550003 (2015)

\bibitem{PitaevskiiStringari} L. Pitaevskii, S. Stringari, \textit{Bose-Einstein Condensation and Superfluidity}, Oxford University Press, New York (2016)

\bibitem{PuleZagrebnov} J.V. Pul\'e, V.A. Zagrebnov, \textit{The canonical perfect Bose gas in Casimir boxes}, J. Math. Phys. \textbf{45}, 9 (2004)

\bibitem{Rou2015} N. Rougerie, \textit{De Finetti theorems, mean-field limits and Bose-Einstein condensation}, preprint arXiv:1506.05263, Lecture notes, (2015)

\bibitem{largecoulombsystems} R. Seiringer, \textit{Dilute, Trapped Bose Gases and Bose-Einstein Condensation}, in: J. Derezinski, H. Siedentop, (Eds.), \textit{Large Coulomb systems, lecture notes on mathematical aspects of QED}, Lect. Notes Phys. \textbf{695}, Springer, Berlin (2006)

\bibitem{RobertFermigas} R. Seiringer, \textit{The Thermodynamic Pressure of a Dilute Fermi Gas}, Commun. Math. Phys. \textbf{261}, 729 (2006)

\bibitem{Sei2008} R. Seiringer, \textit{Free Energy of a Dilute Bose Gas: Lower Bound}, Comm. Math. Phys. \textbf{279}, 595 (2008)

\bibitem{SimonFunct} B. Simon, \textit{Functional Integration and Quantum Physics, Second ed.}, AMS, Providence, Rhode Island (2005)

\bibitem{Suto} A. S\"ut\H{o}, \textit{Correlation inequalities for noninteracting Bose gases}, J. Phys. A: Math. Gen. \textbf{37}, 3 (2004)

\bibitem{Tamm_etal 2011} N. Tammuz, R.P. Smith, R.L.D. Campbell, S. Beattie,  S. Moullder, \textit{Can an Bose Gas be Saturated?}, Phys. Rev. Lett. \textbf{106}, 230401 (2011)

\bibitem{Thirring_4} W. Thirring, \textit{Quantum Mathematical Physics}, $2^{nd}$ ed., Springer, New York (2002)

\bibitem{weyl} 	H. Weyl, \textit{Das asymptotische Verteilungsgesetz der Eigenwerte linearer partieller Differentialgleichungen (mit einer Anwendung auf die Theorie der Hohlraumstrahlung)},  Math. Ann. \textbf{71},  441 (1911) 

\bibitem{Yin2010} J. Yin, \textit{Free Energies of Dilute Bose Gases: Upper Gound}, J. Stat. Phys. \textbf{141}, 683 (2010)


\end{thebibliography}
\end{document}